\documentclass[12pt]{article}

\RequirePackage{amsthm,amsmath,amsfonts,amssymb}
\RequirePackage[colorlinks,citecolor=blue,urlcolor=blue]{hyperref}
\RequirePackage{graphicx}

\usepackage{float}
\usepackage{enumitem}
\usepackage{url}
\usepackage{bbm} 
\usepackage[nottoc]{tocbibind}
\usepackage{booktabs,caption,subcaption,natbib}
\usepackage[flushleft]{threeparttable}
\usepackage{multirow}

\usepackage[pdftex,dvipsnames]{xcolor}
\usepackage{comment} 
\usepackage[margin=1in]{geometry}
\usepackage{accents}
\usepackage{xcolor}
\usepackage{pgfplots}
\usepackage{tikz}
\pgfplotsset{compat = 1.18}
\usepackage[nodisplayskipstretch]{setspace}
\usepackage{xr}

\newtheorem{theorem}{Theorem}
\newtheorem{corollary}{Corollary}[theorem]

\pgfmathdeclarefunction{normal}{2}{
\pgfmathparse{
    1/(sqrt(2*pi*#2))*exp(-1/2*((x-#1)^2)/#2)}%
}

\DeclareMathAccent{\wtilde}{\mathord}{largesymbols}{"65}

\newcommand{\bA}{\boldsymbol{A}}

\newcommand{\bB}{\boldsymbol{B}}

\newcommand{\Bias}{\mbox{Bias}}

\newcommand{\bC}{\boldsymbol{C}}

\newcommand{\bD}{\boldsymbol{D}}

\newcommand{\bG}{\boldsymbol{G}}

\newcommand{\bI}{\boldsymbol{I}}
\newcommand{\bK}{\boldsymbol{K}}

\newcommand{\bM}{\boldsymbol{M}}

\newcommand{\bone}{\boldsymbol{1}}

\newcommand{\bQ}{\boldsymbol{Q}}

\newcommand{\bR}{\boldsymbol{R}}

\newcommand{\bu}{\boldsymbol{u}}

\newcommand{\bx}{\boldsymbol{x}}
\newcommand{\bX}{\boldsymbol{X}}

\newcommand{\by}{\boldsymbol{y}}
\newcommand{\bY}{\boldsymbol{Y}}

\newcommand{\bZ}{\boldsymbol{Z}}
\newcommand{\bzero}{\boldsymbol{0}}

\newcommand{\bbeta}{\boldsymbol{\beta}}

\newcommand{\bepsilon}{\boldsymbol{\epsilon}}

\newcommand{\bgamma}{\boldsymbol{\gamma}}

\newcommand{\blambda}{\boldsymbol{\lambda}}
\newcommand{\bLambda}{\boldsymbol{\Lambda}}
\newcommand{\bmu}{\boldsymbol{\mu}}

\newcommand{\bpi}{\boldsymbol{\pi}}
\newcommand{\bPi}{\boldsymbol{\Pi}}

\newcommand{\bsigma}{\boldsymbol{\sigma}}
\newcommand{\bSigma}{\boldsymbol{\Sigma}}

\newcommand{\btheta}{\boldsymbol{\theta}}

\newcommand{\hbbeta}{\boldsymbol{\widehat{\beta}}}
\newcommand{\tbbeta}{\boldsymbol{\tilde{\beta}}}

\newcommand{\hbtheta}{\boldsymbol{\widehat{\theta}}}

\newcommand{\tbgamma}{\boldsymbol{\tilde{\gamma}}}

\newcommand{\hc}{\widehat{c}}
\newcommand{\scaledXTX}{\bX^{\top}\bSigma^{-1}\bX}
\newcommand{\scaledXTMX}{\bX^{\top}\bM\bSigma^{-1}\bM\bX}

\newcommand{\scaledXTMY}{\bX^{\top}\bM\bSigma^{-1}\bM\bY}

\newcommand{\scaledXTXj}{\frac{1}{\sigma_j^2}\bX_j^{\top}\bX_j}
\newcommand{\scaledXTYj}{\frac{1}{\sigma_j^2}\bX_j^{\top}\bY_j}

\newcommand{\parDeriv}[2]{\frac{\partial #1}{\partial #2}}

\newcommand{\tr}{\mathsf{tr}}
\newcommand{\EV}{\mathsf{E}}
\newcommand{\Var}{\mathsf{var}}
\newcommand{\Cov}{\mathsf{cov}}
\newcommand{\MSE}{\mathsf{MSE}}

\newcommand{\BMSE}{\mathsf{BMSE}}

\newcommand{\mle}{\textsf{MLE}}
\newcommand{\ham}{\textsf{HAM}}
\newcommand{\has}{\textsf{HAS}}
\newcommand{\fe}{\textsf{FE}}
\newcommand{\ridge}{\textsf{R}}
\newcommand{\sub}{\textsf{S}}
\newcommand{\ure}{\textsf{URE}}

\title{\textbf{Heterogeneity-adaptive shrinkage estimation and inference for linear panel models}}

\author{Elizabeth M. Davis and Emily C. Hector\\
Department of Biostatistics, University of Michigan}
\date{}

\begin{document}

\maketitle

\begin{abstract} 
Recent shrinkage estimators for linear panel models have focused on increasing flexibility beyond empirical Bayes estimators to deliver mean squared error optimality, but these estimators' performance rarely meets their theoretical guarantees in practice. 
To balance theoretical and practical performance, we develop a heterogeneity-adaptive shrinkage estimator for linear panel models that adapts to the amount of \emph{information} shared between groups. The primary mechanism for the information-sharing is a shrinkage of group-specific distributions towards a new ``centroid'' distribution through a Kullback-Leibler divergence penalty. The Kullback-Leibler divergence is uniquely geometrically suited for measuring relative information between datasets, and leads to relatively simple closed form estimators with intuitive interpretations. Among other results, we show that our estimator has a provably smaller mean squared error than the group-specific maximum likelihood estimators, and establish asymptotically valid inference procedures. We also propose a new estimator of mean squared error for implementation that leads to improved finite sample performance over minimizing Stein's unbiased risk estimate (SURE). A comprehensive set of simulations highlights our estimator's versatility, and an analysis of length of stay in ICUs illustrates its performance in a real-world setting.
\end{abstract}


\section{Introduction}\label{sec:intro}
The econometric literature has long debated the extent of information that should be shared between groups in linear panel data when estimating group-specific linear regression coefficients \citep{maddala_estimation_1997, baltagi_pool_2000}. On the one hand, estimating all group-specific coefficients separately yields unbiased estimators but may discard valuable information that is shared across groups, sacrificing precision. On the other hand, pooling the data to estimate a single set of coefficients reduces variability but can introduce bias when the group-specific parameters are heterogeneous. To trade off between these two extremes, a rich literature has emerged focused on developing shrinkage estimators that shrink group-specific estimates towards pooled estimates, resulting in a weighted average of the two. 
To make our discussion more concrete, consider an estimator for the linear regression coefficients of the $j$th of $k$ groups, $\bbeta_j\in\mathbb{R}^p$. The most general form of the weighted average estimator is 
\begin{align}\label{eqn:general_weighted_avg}
\hbbeta_j = (\bI_p - \bPi_j)\tbbeta_j + \bPi_j \hbtheta_j, \quad j = 1,\dots, k,    
\end{align}
where $\tbbeta_j\in\mathbb{R}^p$ is the $j$th group's maximum likelihood estimator (MLE), $\bI_p$ is the $p\times p$ identity matrix, $\bPi_j$ is the $j$th matrix of shrinkage weights and $\hbtheta_j\in\mathbb{R}^p$ is a pooled estimator of quantity $\btheta_j$, allowed to differ for each group. We refer to the elements in $\bPi_j$ and $\btheta_j$ as \textit{free parameters} when these are not predetermined by the model. 

Existing weighted average shrinkage estimators for the linear panel model vary widely in their flexibility. Among the least flexible approaches are model-based methods, which have zero free parameters. In these, the pooled parameter $\btheta_j$ corresponds to the prior mean of the $\bbeta_j$, and both it and the shrinkage weights $\bPi_j$ are fixed functions of the data; they do not change after the data are observed. For example, in the classical empirical Bayes estimator with normal priors, the functional forms of both the pooled estimate and weights are fully determined by the priors; in particular, $\bPi_j = \bSigma_j(\bLambda + \bSigma_j)^{-1}$,  where $\bSigma_j$ is the within-group variance and $\bLambda$ is the between-group variance \citep{morris_parametric_1983, walters_chapter_2024}. Model-based empirical Bayesian estimators place strong parametric assumptions on the distribution of group effects to deliver inference \citep{armstrong_robust_2022}, which often results in improperly calibrated inference; see the discussion below.

In contrast, more recent methods are data-adaptive and highly flexible. These techniques use the empirical Bayesian set-up to derive the weighted average form, which bestows interpretability on the pooled parameter $\btheta_j$ as the prior mean, but at least one free parameter is selected to minimize the mean squared error (MSE) \citep{xie_sure_2012, hansen_efficient_2016, mehrabani_improved_2020, kwon_optimal_2026}.  
When the MSE, which depends on the model parameters, is known, by design this optimization yields an estimator with lower MSE than can be achieved by less flexible estimators. In practice, as the MSE is unobserved, data-adaptive methods generally estimate the between-group variance ($\bLambda$) and $\btheta_j$ by minimizing Stein's unbiased risk estimator (SURE) \citep{stein_estimation_1981} of the MSE. Intuitively, estimating the free parameters by minimizing SURE rather than the true MSE introduces bias and variability in the estimator $\hbbeta_j$ that jeopardizes its MSE-optimality. While SURE is an unbiased estimator of the MSE, its minimizer need not be an unbiased estimator of the parameters that minimize the MSE. As a result, increasing the flexibility of $\hbbeta_j$ by increasing the number of free parameters incurs a cost, as more parameters need to be estimated based on a fixed sample. These methods' finite sample performance in practice therefore rarely matches the theoretical minimized MSE guarantee. Moreover, the empirical Bayes framework places low prior probability on realistic patterns of heterogeneity, such as when a group's linear regression coefficients deviate substantially from the bulk of the other groups' coefficients. This poses a challenge for inference, as it leads to posterior credible intervals that are not calibrated. When coupled with very small MSEs, we may become overconfident in the performance of the estimator.

In this paper, we propose a new approach to data integration in the linear panel model, one which arrives at the frequently-used weighted average shrinkage estimator \eqref{eqn:general_weighted_avg} without the empirical Bayesian framework. Our approach requires neither a specified underlying generative model for the group-specific parameters nor the existence of an interpretable shared parameter. The primary mechanism for information-sharing is a shrinkage of group-specific outcome distributions towards a floating ``centroid'' distribution through a Kullback-Leibler divergence (KLD) penalty \citep{kullback_information_1951}. Our formulation jointly estimates parameters of both the unknown group-specific distributions \emph{and} the centroid distribution which links the groups. The estimated centroid is a generalization of $\btheta_j$ in \eqref{eqn:general_weighted_avg} that depends on the $k$ shrinkage weights. 
By equipping each group with its own shrinkage weight and selecting these to minimize the MSE of the group-specific estimators, our framework borrows information adaptively and efficiently according to heterogeneity in the observed groups. Our formulation yields a group-specific estimator that is a weighted average of the group-specific MLEs and the centroid, similar in spirit to equation \eqref{eqn:general_weighted_avg}. 

Further, we show that in a common special case of the estimator in \eqref{eqn:general_weighted_avg}, minimizing SURE leads on average to overborrowing and therefore inflated MSE. In practice, this results in a shrinkage estimator that frequently performs worse than the group-specific MLEs, thwarting our attempt to improve MSE by borrowing information. To address this finite-sample limitation of SURE, we propose a surrogate MSE using ideas on bias reduction of MLEs from \cite{firth_bias_1993}. We show that minimizing the surrogate MSE delivers improved finite-sample performance over minimizing SURE.

Our method balances flexibility with interpretability and inference. We show that our use of the KLD to quantify the similarity between groups yields relatively simple parameter estimates with desirable and intuitive interpretations. In particular, for the special case where all shrinkage parameters are equal, the centroid is the familiar classical pooled weighted least squares estimator. When a group's shrinkage parameter is zero, that group is excluded from the information-sharing. In conjunction with the flexibility of our estimated centroid, this has important consequences for the statistical advantages of our estimator, such as asymptotic bias control and smaller mean squared errors. Notably, we show that our estimator outperforms the group-specific maximum likelihood estimators, in a mean squared error sense,  while also yielding asymptotically unbiased inference on the individual group parameters.

This paper is organized as follows.  
Section \ref{sec:estimator} describes the model, presents a closed form solution for our estimator and develops an intuition for its behavior under different shrinkage parameter settings.
In Section \ref{sec:theory}, we investigate the theoretical properties of our estimator. We show that, similar to a Stein shrinkage result, it is \emph{always} better in a mean squared error sense to borrow information across groups, even in the presence of substantial heterogeneity. Under mild conditions, we also show that our estimator is consistent and we derive asymptotically valid inference procedures. Section \ref{sec:implementation} describes the data-driven selection of the amount of information borrowed based on our proposed surrogate MSE, and the theoretical ``gap'' between the optimal information borrowed and its data-driven counterpart. Section \ref{sec:sim} examines the performance of our estimator in simulation studies. An analysis of data from the eICU Collaborative Research Database in Section \ref{sec:data} illustrates our estimator's behavior in a real-life setting. Lengthy derivations, proofs and additional numerical results are deferred to the supplement, and code for the simulation studies and data analysis are available in the first author's \href{https://github.com/emdavis7/heterogeneity_adaptive_meta_analysis}{GitHub repository}.

\section{The heterogeneity-adaptive estimator}\label{sec:estimator}

\subsection{Estimation with common covariates}\label{subsec:est_all}

Consider a scenario with $k$ independent groups, indexed $j = 1,\dots, k$, where each group contains $n_j$ observations consisting of an outcome vector $\bY_j\in \mathbb{R}^{n_j}$ and a matrix of covariates $\bX_j\in\mathbb{R}^{n_j\times p}$. We assume the covariates are associated with the outcomes in each group through the model $\bY_j \sim \mathcal{N}(\bX_j\bbeta_j, \sigma_j^2\bI_{n_j})$, with $\bI_{n_j}$ the $n_j\times n_j$ identity matrix. For simplicity of exposition, we assume that all observations are independent; our method is easily generalizable to the scenario where $\bY_j \sim \mathcal{N}(\bX_j\bbeta_j, \bSigma_j)$, as long as the investigator has a consistent estimate for $\bSigma_j$.
Let $\bX=\text{block-diag}\{\bX_1, \ldots, \bX_k\}$ be the block-diagonal assembly of covariate matrices $\bX_j$ with dimension $\sum_{j=1}^k n_j \times pk$ and $\bY=(\bY_1, \ldots, \bY_k)$ be the outcome vectors appended in the same order. Then, the joint model for all $k$ groups is
\begin{align*}
    \bY = \bX\bbeta + \bepsilon, \quad \bepsilon\sim \mathcal{N}(\bzero, \bSigma),
\end{align*}
where $\bSigma=\text{block-diag}\{ \sigma^2_1\bI_{n_1}, \ldots, \sigma^2_k \bI_{n_k}\}$ is the block-diagonal matrix of group-specific variances and $\bbeta=(\bbeta_1, \ldots, \bbeta_k) \in \mathbb{R}^{pk}$. Estimating $\bbeta$ using maximum likelihood estimation returns a stacked vector of the maximum likelihood estimators (MLEs) for the individual groups, $\tbbeta = (\tbbeta_1,\dots, \tbbeta_k)'$, without any information shared between them; modeling in this way treats $\bbeta_j$ and $\bbeta_{j^\prime}$ for any groups $j$ and $j^\prime$ as though there is complete heterogeneity. In contrast, if we believe it is reasonable to assume perfect homogeneity so that $\bbeta_j = \bbeta_{j^\prime}$ for all $j, j^\prime\in \{1,\dots,k\}$, the MLE is the uniformly minimum variance unbiased estimator used in fixed-effects meta-analysis \citep{glass1976},
\begin{align} \label{eqn:fe}
\tbbeta_{\fe} &= \left(\sum_{j=1}^k\scaledXTXj \right)^{-1}\sum_{j=1}^k \scaledXTYj = \left(\sum_{j=1}^k\scaledXTXj \right)^{-1}\sum_{j=1}^k\scaledXTXj\tbbeta_j.
\end{align}

Clearly, there is a homogeneity/heterogeneity continuum between the assumptions that $\bbeta_j \neq \bbeta_{j^\prime}$ for all $j \neq j^\prime$ and $\bbeta_j = \bbeta_{j^\prime}$ for all $j \neq j^\prime$ that is not captured by $\tbbeta$ or $\tbbeta_{\fe}$. Further, the ends of the spectrum, i.e. the full homogeneity and full heterogeneity assumptions, treat all groups equally by assuming the same level of heterogeneity (no or full heterogeneity) for all parameters. In contrast, we propose an estimator for $\bbeta$ along the homogeneity/heterogeneity continuum that is fully adaptive in the sense that each group may be more or less heterogeneous than the others.

First, we encourage information-sharing between groups by shrinking estimates for $\bbeta_j$ from outcome distributions that are ``far'' from the other groups. We begin by defining a ``centroid'' outcome distribution for group $j$, $\mathcal{N}(\bX_j\btheta, \sigma_j^2\bI)$, where $\btheta\in \mathbb{R}^p$. 
To measure the difference between using $\mathcal{N}(\bX_j\btheta, \sigma_j^2\bI_{n_j})$ and $\mathcal{N}(\bX_j\bbeta_j, \sigma_j^2\bI_{n_j})$ to model $\bY_j$, we seek a topologically appropriate measure that, in contrast to the Euclidean distance between parameters, accounts for other features of the distributions, such as the error and covariate variances. To compare the distributions, we use the Kullback-Leibler divergence (KLD) \citep{amari_information_2020, kullback_information_1951}. The KLD is the unique Bregman divergence on the manifold of Gaussian distributions, and it measures the relative entropy of a distribution $\mathcal{N}(\bX_j\bbeta_j, \sigma_j^2\bI_{n_j})$ with respect to another distribution $\mathcal{N}(\bX_j\btheta, \sigma_j^2\bI_{n_j})$. It also measures the excess surprise from using $\mathcal{N}(\bX_j\btheta, \sigma_j^2\bI_{n_j})$ over $\mathcal{N}(\bX_j\bbeta_j, \sigma_j^2\bI_{n_j})$, when $\mathcal{N}(\bX_j\bbeta_j, \sigma_j^2\bI_{n_j})$ is the true distribution. 
Since a manifold is locally equivalent to a Euclidean space, small divergences between the distributions are equivalent to small distances between the parameters themselves. This deformation only holds locally, however, so that large divergences between $\mathcal{N}(\bX_j\btheta, \sigma_j^2\bI_{n_j})$ and $\mathcal{N}(\bX_j\bbeta_j, \sigma_j^2\bI_{n_j})$ do not necessarily equate to large distances between $\btheta$ and $\bbeta_j$. 

To estimate $\bbeta$, we maximize the joint log-likelihood subject to the constraint that all group-specific distributions be close in relative entropy to the centroid distribution, or equivalently
\begin{align}\label{eqn:obj_func}
   O(\bY, \bX; \bbeta, \btheta, \bsigma^2) =  
    \sum_{j=1}^k \Bigl\{-\frac{1}{2\sigma^2_j}\|\bY_j-\bX_j\bbeta_j \|^2 
    - \Bigl(\frac{\pi_j}{1-\pi_j}\Bigr)\Bigl(\frac{1}{2\sigma_j^2}\Bigr)\|\bX_j\bbeta_j - \bX_j\btheta\|^2\Bigr\},
\end{align}
where the shrinkage parameter for each group is $\pi_j/(1-\pi_j)$, $\pi_j\in [0,1]$. At $\pi_j=1$, the penalty is infinite and the estimate for $\bbeta_j$ is required to equal the estimate for $\btheta$. As the discrepancy between $\mathcal{N}(\bX_j\btheta, \sigma_j^2\bI_{n_j})$ and $\mathcal{N}(\bX_j\bbeta_j, \sigma_j^2\bI_{n_j})$ is measured by group $j$'s Fisher information metric, the divergence is dominated by high-information directions of $\bX_j^\top \bX_j$, which can obscure heterogeneity in low-information directions. To avoid this, we recommend within-group standardization of $\bX_j$ before pooling through the objective function in equation \eqref{eqn:obj_func}. For simplicity of exposition, we treat $\sigma_j^2$ as known in all derivations; in practice, we estimate it using its MLE from group $j$. Define $\bpi = (\pi_1, \dots, \pi_k)^\prime$, the vector of shrinkage parameters. We give the closed-form estimators that maximize $O(\bY, \bX; \bbeta, \btheta, \bsigma^2)$ below, with the non-trivial derivations relegated to Section \ref{appsec:derivation} of the supplement. The centroid that maximizes this objective function is 
\begin{align}\label{eqn:centroid}   
    \hbtheta(\bpi) &= \Bigl(\sum_{j=1}^k \frac{\pi_j}{\sigma_j^2}\bX_j^\prime\bX_j\Bigr)^{-1}\sum_{j=1}^k \frac{\pi_j}{\sigma_j^2}\bX_j^\prime\bY_j = \Bigl(\sum_{j=1}^k \frac{\pi_j}{\sigma_j^2}\bX_j^\prime\bX_j\Bigr)^{-1}\sum_{j=1}^k \frac{\pi_j}{\sigma_j^2}\bX_j^\prime\bX_j\tbbeta_j,
\end{align}
and group $j$'s estimator for $\bbeta_j$, $\hbbeta_j(\bpi)$, can be written as a convex combination of the MLE, $\tbbeta_j$, and the centroid estimator:
\begin{align}\label{eqn:hbbeta_j}
    \hbbeta_j(\bpi) &= (1-\pi_j)\tbbeta_j + \pi_j\hbtheta(\bpi).
\end{align} 
The full vector of $\hbbeta(\bpi) = \{\hbbeta_1(\bpi)^\prime, \dots, \hbbeta_k(\bpi)^\prime\}^\prime$ is 
\begin{align}\label{eqn:hbbeta_reexpr}
    \hbbeta(\bpi) &=(\bI_{pk} - \bPi)\tbbeta + \bPi\bK\hbtheta(\bpi)  
    = \tbbeta + \bPi\{\bK\bA(\bpi) - \bI_{pk}\}\tbbeta, 
\end{align}
where $\bK=\bone_k \otimes \bI_p$ is a matrix which stacks the vector to its right with $\otimes$ the Kronecker product, $\bPi=\text{block-diag}\{\pi_j \bI_p\} \in \mathbb{R}^{pk \times pk}$, and $\bA(\bpi)=(\bK^\prime\bPi\scaledXTX\bK)^{-1}\bK^\prime\bPi\scaledXTX$.  

Both $\hbbeta(\bpi)$ and $\hbtheta(\bpi)$ are extremely flexible because they shift with the choice of shrinkage parameter $\bpi$. The dependence of $\hbtheta(\bpi)$ on the shrinkage parameters $\bpi$ makes the centroid capable of representing the space of heterogeneity patterns. This includes representation of some recognizable estimators. When $\pi_j = \pi_{j^\prime}$ for all $j\neq j^\prime$, $\hbtheta(\bpi)=\tbbeta_{\fe}$ the fixed-effect meta-estimator in equation \eqref{eqn:fe}. If $\bpi$ is zero for $k-x$ groups and non-zero and equal for the remaining $x$ groups ($x \in \mathbb{Z}, 0\leq x \leq k$), then $\hbtheta(\bpi)$ is the fixed-effect meta-estimator for the groups with non-zero $\pi_j$. If $\bpi$ contains zeroes for all groups except the $j$\textsuperscript{th}, then $\hbtheta(\bpi) = \tbbeta_{j}$. Panel A in Figure \ref{fig:combined_ex} shows the possible values of $\hbtheta(\bpi)$ given the maximum likelihood estimates $\tbbeta_{j}$ and information matrices for four groups with $p=2$ covariates. 

The estimates $\hbbeta_j(\bpi)$ for the group-specific parameters also depend on $\bpi$. Panels B and C of Figure \ref{fig:combined_ex} show the set of feasible estimates $\hbbeta(\bpi)$ for the same data as displayed in panel A, with examples for changes in $\bpi$. Each corner is the MLE for a particular group. In panel B, we see that a move from $\bpi = (0.2, 0.2, 0.2, 0.2)$ to $\bpi = (0.8, 0.8, 0.8, 0.8)$ pulls all estimates for $\hbbeta_j(\bpi)$ towards $\tbbeta_{\fe}$. In panel C, which displays a move from $\bpi = (0.5, 0.5, 0.5, 0)$ to $\bpi = (0.5, 0.5, 0.5, 0.3)$, an increase in the contribution from group 4 shifts the estimates for groups 1, 2, and 3 towards $\tbbeta_4$. We aim to select $\bpi$ to minimize the MSE of $\hbbeta(\bpi)$, which we discuss further in Section \ref{sec:theory}.

\begin{figure}[ht]
    \centering
    \includegraphics[width=0.95\linewidth]{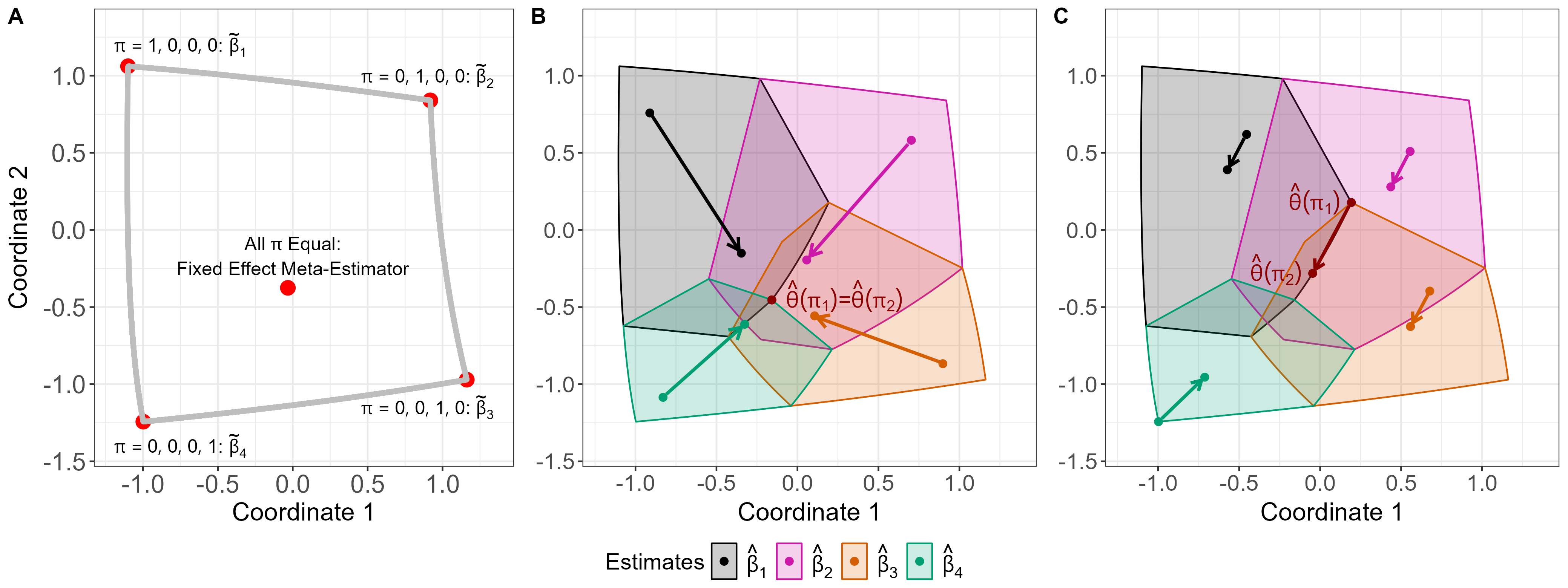}
    \caption{Panel A displays possible values for $\hbtheta(\bpi)$ for four datasets. Any coordinate pair within the outlined shape can be obtained for an appropriate $\bpi$ vector. Panels B and C display possible values for each $\hbbeta_j(\bpi)$ for $k=4$. For a given dataset, denoted by color, any coordinate pair within the polygon may be obtained for an appropriate $\bpi$ vector. In example A, arrows display a shift from $\bpi = (0.2, 0.2, 0.2, 0.2)$ to $\bpi = (0.8, 0.8, 0.8, 0.8)$. In example B, arrows display a shift from $\bpi = (0.5, 0.5, 0.5, 0)$ to $\bpi = (0.5, 0.5, 0.5, 0.3)$.}
    \label{fig:combined_ex}
\end{figure}

Our objective function \eqref{eqn:obj_func} leads to a weighted average estimator like \eqref{eqn:hbbeta_j} even when $\btheta$ is fixed, revealing a connection to empirical Bayes estimators and their newer, more flexible counterparts; prior work notes this relationship only in the special case $\btheta = \bzero$ \citep{tibshirani_excess_2019}.  

We have fewer shrinkage parameters to estimate than the most flexible methods, but in one important respect, our estimator is more adaptive than those proposed in \cite{xie_sure_2012} and \cite{kwon_optimal_2026}: our shrinkage weights are group-specific. In the normal-normal model for linear panel data, these methods' shrinkage weights are a function of a shared variance parameter $\bLambda$ and the group-specific variance, $\bSigma_j$, and $\bLambda$ is selected to minimize SURE; the amount that group $j$'s estimate borrows from the pooled estimate depends on the precision of $\tbbeta_j$. In contrast, our group-specific shrinkage weights can respond to observed heterogeneity to borrow less where a group differs, regardless of the relative precision of its estimates. Further, by selecting our weights to minimize the MSE, we show in Section \ref{sec:theory} that our framework enables asymptotically valid inference, a property absent in \citet{xie_sure_2012}, \citet{mehrabani_improved_2020}, and \citet{kwon_optimal_2026}.

\subsection{Estimation with overlapping covariates}\label{subsec:est_subset}

In Section \ref{subsec:est_all}, we assumed that the same covariates are measured across the $k$ groups, which may not always be the case in practice. Consider instead the setting where the $j$th group has $q_j \geq p$ covariates, $p$ of which are common across groups. The remaining $q_j - p$ covariates may be measuring different features. This adds flexibility to the settings that our method can accommodate: each group may control for different covariates, and the standardized effect sizes for a main effect of interest are combined. 

For group $j$, denote the $n_j \times p$ matrix for the effects of interest as $\bX_j$ and denote the $n_j \times (q_j - p)$ matrix of control covariates as $\bZ_j$. Outcomes for group $j$ are modeled as 
\begin{align}\label{eqn:partitioned_model}
    \bY_j = \bX_j\bbeta_j + \bZ_j\bgamma_j + \epsilon_j,\quad \epsilon_j \sim N(\bzero,\sigma_j^2\bI_{n_j}).
\end{align}

Our goal is to estimate $\bbeta=(\bbeta_1, \ldots, \bbeta_k)$ in the framework proposed in Section \ref{subsec:est_all}.  Denote the MLEs of $\bbeta_j$ and $\bgamma_j$ by $\tbbeta_j$ and $\tbgamma_j$, respectively, with the covariance matrix of the estimators given by
\begin{align}\label{eqn:full_covariance_j}
\sigma^2_j\begin{bmatrix}
      \bX_j^\top \bX_j & \bX_j^\top \bZ_j\\
      \bZ_j^\top \bX_j & \bZ_j^\top \bZ_j
  \end{bmatrix}^{-1}  .
\end{align}
To reformulate equation \eqref{eqn:partitioned_model} with only the covariates in common across datasets, we first project the outcome model onto the null space of $\bZ_j$. Let $\bM_j = \bI_{n_j} - \bZ_j(\bZ_j^{\top}\bZ_j)^{-1}\bZ_j^{\top}$ when $q_j > p$, and $\bM_j=\bI_{n_j}$ when $q_j=p$. In this projected outcome model, additional covariates $\bZ_j$ are removed:
\begin{align}\label{eqn:residual_model}
    \bM_j\bY_j = \bM_j\bX_j\bbeta_j +  \bM_j\epsilon_j,\quad \bM_j\epsilon_j \sim N(\bzero,\sigma_j^2\bM_j).
\end{align}
Given the outcome model in equation \eqref{eqn:residual_model}, the objective function in equation \eqref{eqn:obj_func} becomes 
\begin{align}\label{eqn:obj_func2}
    &\sum_{j=1}^k \Bigl\{-\frac{1}{2\sigma^2_j}(\bM_j\bY_j-\bM_j\bX_j\bbeta_j)^{\top}\bM_j^{+}(\bM_j\bY_j-\bM_j\bX_j\bbeta_j) \nonumber\\
    &\quad 
    - \Bigl(\frac{\pi_j}{1-\pi_j}\Bigr)\Bigl(\frac{1}{2\sigma_j^2}\Bigr)(\bM_j\bX_j\bbeta_j - \bM_j\bX_j\btheta)^{\top}\bM_j^{+}(\bM_j\bX_j\bbeta_j - \bM_j\bX_j\btheta)\Bigr\},
\end{align}
where $\bM_j^{+}$ is the generalized inverse of $\bM_j$; as $\bM_j$ is idempotent, it only has an inverse when $\bM_j=\bI_{n_j}$. In the objective function in equation \eqref{eqn:obj_func2}, the KLD penalty shrinks $\bbeta_j$ towards the centroid for the distribution of $\bM_j\bY_j$. We derive maximizers $\hbtheta_s(\bpi)$ and $\hbbeta_s(\bpi)$ of equation \eqref{eqn:obj_func2} in Section \ref{appsec:derivation} of the supplement, and give their closed forms here: 
\begin{align*}
    \hbtheta_{s}(\bpi) &= \Bigl(\sum_{j=1}^k \frac{\pi_j}{\sigma_j^2 } \bX_j'\bM_j\bX_j\Bigr)^{-1} \sum_{j=1}^k \frac{\pi_j}{\sigma_j^2}\bX_j^{\prime}\bM_j\bX_j\tbbeta_j, \quad
    \hbbeta_s(\bpi) =   (\bI - \bPi)\tbbeta  +  \bPi\bK \hbtheta_s(\bpi),
\end{align*}
where $\bM=\text{block-diag}\{\bM_1, \ldots, \bM_k\}$. Using the covariance matrix in equation \eqref{eqn:full_covariance_j}, we can compute $\bX_j^{\top}\bM_j\bX_j$ following
\begin{align*}
    \bX_j^\top\bM_j\bX_j &= \bX_j^\top \{\bI - \bZ_j(\bZ_j^{\top}\bZ_j)^{-1}\bZ_j^{\top}\}\bX_j= \bX_j^\top\bX_j - \bX_j^\top\bZ_j(\bZ_j^{\top}\bZ_j)^{-1}\bZ_j^{\top}\bX_j.
\end{align*} 
When covariates $\bZ_j$ are orthogonal to $\bX_j$ for all $j=1, \ldots, k$, the above estimators simplify to those from Section \ref{subsec:est_all}: $\hbtheta_{s}(\bpi)=\hbtheta(\bpi)$ and $\hbbeta_s(\bpi)=\hbbeta(\bpi)$.

\section{Theoretical support}\label{sec:theory} 

For simplicity of exposition, we focus on deriving theoretical support for the estimators in Section \ref{subsec:est_all}; similar results can be shown for those in Section \ref{subsec:est_subset} and are omitted. Our goal is to improve estimation across all groups on average by selecting the shrinkage parameter vector $\bpi$ that minimizes the mean squared error of $\hbbeta(\bpi)$. The MSE-minimizing $\bpi$ balances both variance and bias of our estimator and, hopefully, leads to an improvement over the MLE, $\tbbeta$. We begin by demonstrating that such an improvement exists using the simplified setting where $\pi_j\equiv \pi$. We subsequently extend our results to general $\bpi$ vectors by developing an alternative parameterization that yields an interpretable closed-form solution for the optimal scale of borrowing between datasets. Finally, we show that our estimator is consistent and asymptotically normal at the optimal scaling, enabling asymptotically valid inference.

Our objective is to select $\bpi$ to minimize the MSE of $\hbbeta(\bpi)$, derived in Section \ref{appsec:derivation_mse} of the supplement and given by
\begin{align}\label{eqn:mse}
    \MSE\{\hbbeta(\bpi)\} &= \|\bPi\{\bK\bA(\bpi) - \bI_{pk}\}\bbeta\|^2  +  \tr(\Var [\bPi\{\bK\bA(\bpi) - \bI_{pk}\}\tbbeta]) \nonumber\\
    &\quad + 2\tr(\Cov[\tbbeta, \bPi\{\bK\bA(\bpi) - \bI_{pk}\}\tbbeta]) + \tr\{\Var(\tbbeta)\},
\end{align}
where $\tr(\cdot)$ denotes the trace of a matrix. 
A preliminary question is whether minimizing this MSE leads to a choice of $\bpi$ that necessarily offers an improvement in MSE over using the group-specific MLEs. To see that it does, we consider a simpler case where all $\pi_j\equiv\pi$. If an improvement in MSE exists for this simpler case, then selecting $\bpi$ to minimize equation \eqref{eqn:mse} over the full space $[0,1]^k$ must yield an estimator $\hbbeta(\bpi)$ with smaller MSE than the MLE. Theorem \ref{existence_ma} states this result formally: there exists a vector $\bpi$, $\pi_j$ not all zero, such that the MSE of $\hbbeta(\bpi)$ is strictly less than the MSE of $\tbbeta$.
 
\begin{theorem}
    \label{existence_ma} 
    There exists a $\bpi\neq \bzero$ such that $\MSE\{\hbbeta(\bpi)\}< \MSE(\tbbeta).$ In the case where $\pi_j\equiv\pi$, the MSE-minimizing $\pi$ is 
    \begin{align*}
        \pi^{\star} &= \frac{-\tr\{(\bK\bA - \bI_{pk})(\scaledXTX)^{-1}\}}{\tr\{(\bK\bA - \bI_{pk})(\scaledXTX)^{-1}(\bK\bA - \bI_{pk})^\prime\} + \|(\bK\bA - \bI_{pk})\bbeta\|^2}, 
    \end{align*}
    where $\bA = (\bK^\prime \scaledXTX \bK)^{-1} \bK^\prime \scaledXTX$.
\end{theorem}

As discussed in Section \ref{sec:estimator}, when $\pi_j\equiv \pi$, $\hbtheta(\bpi)$ is no longer a function of $\bpi$ and is equal to $\tbbeta_{\fe}$ in equation \eqref{eqn:fe}. Consequently, $\pi$ can be interpreted as the proportion of $\hbbeta_j(\pi)$ coming from the fixed-effect meta-estimator:
\begin{align*}
\hbbeta_j(\pi) =(1-\pi)\tbbeta + \pi\hbtheta = \tbbeta + \pi(\bK\bA - \bI_{pk})\tbbeta.
\end{align*}
When $\pi = 0$, $\hbbeta_j(0) = \tbbeta$ and the estimator is unbiased. As $\pi$ increases to 1, estimation of $\hbbeta_j(\pi)$ leverages more information from the other groups, and the variance of $\hbbeta_j(\pi)$ decreases while the bias increases. The proof of Theorem \ref{existence_ma} shows further that the bias does not increase too rapidly relative to the variance's decrease, implying that a range of $\pi$ exists for which the MSE of $\hbbeta(\bpi)$ is less than that of the MLE. In fact, any $\pi \in (0, 2\pi^{\star})$ yields an estimator $\hbbeta(\pi)$ with smaller MSE than the MLE. Mathematical details are given in Section \ref{appsec:imp_single_tuning} of the supplement. In our implementation, we use the closed form of $\pi^{\star}$ in Theorem \ref{existence_ma} as a starting value when searching for the optimal $\bpi$ vector, plugging $\tbbeta$ in for the unknown $\bbeta$.

Theorem \ref{existence_ma} shows that there exists a $\bpi$ such that our estimator offers a smaller MSE than the MLE, but it does not provide guidance on how to choose the optimal $\bpi$. We next develop a reparameterization of our estimator that decouples the role of $\bpi$ in selecting both the centroid and the amount of shrinkage towards the centroid. While technically involved, the reparameterization delivers a closed-form solution for the optimal amount of information borrowing. With this solution, we will be able to show that our estimator is consistent and asymptotically normal under mild conditions. 

We motivate the reparameterization in Figure \ref{fig:Two-Population Case} using an example with $k=2$ groups. This figure shows that each pair of $\bpi$ values along a ray from the origin gives the same $\hbtheta(\bpi)$. Mathematically, this is due to the fact that, through a simple rescaling,
\begin{align*}
\hbtheta(\bpi) &= \Bigl(\sum_{j=1}^k \pi_j\scaledXTXj)^{-1}\sum_{j=1}^k \pi_j\scaledXTXj\tbbeta = \Bigl(\sum_{j=1}^k a\pi_j\scaledXTXj)^{-1}\sum_{j=1}^k a\pi_j\scaledXTXj\tbbeta \\
&= \hbtheta(a\bpi),\quad a\in\mathbb{R}^+,\quad a\pi_j\leq 1.
\end{align*}
As we move along a ray away from the origin, $\hbbeta(\bpi)$ moves away from $\tbbeta$ towards $\hbtheta(\bpi)$ but $\hbtheta(\bpi)$ remains fixed. 
\begin{figure}[h]
    \centering
    \begin{tikzpicture}
\begin{axis}[
    xmin = 0, xmax = 1,
    ymin = 0, ymax = 1,
    xlabel = \(\pi_1\),
    ylabel = \(\pi_2\)]
\addplot[
    color = black,
    domain = 0:1
    ]
{x};
\addplot[
    color = blue,
    domain = 0:1
    ]
{.1*x};
\addplot[
    color = red,
    domain = 0:1
    ]
{2*x};
\addplot[
    mark = square,
    color = black]
    coordinates {(.25, .25)
        (.5, .5)
        (.75, .75)};
\addplot[mark = square,
    color = red]
    coordinates {(.125, .25)
        (.25, .5)
        (.375, .75)};
\addplot[mark = square,
    color = blue]
    coordinates {(.25, .025)
        (.5, .05)
        (.75, .075)};
\end{axis}
\end{tikzpicture}
    \caption{Each ray represents a line of equal $\hbtheta (\bpi)$. As $\pi_1$ and $\pi_2$ get further from the origin, $\hbbeta(\bpi)$ is more heavily weighted towards the combined estimator, borrowing more information.}
        \label{fig:Two-Population Case}
\end{figure}
Thus, the determination of our estimator can be thought of as a two-step process: first the centroid $\hbtheta(\bpi)$ is selected by choosing a ray, and then a global scaling factor dictates how much each $\hbbeta_j(\bpi)$ shrinks toward the centroid along the ray.
This observation suggests an alternative formulation of our estimator: represent $\bpi$ as the product of a constrained vector $\bpi_r=(\pi_{r,1},\ldots,\pi_{r,k})$ characterizing the centroid-selecting ray (as visualized in Figure \ref{fig:Two-Population Case}), and a scaling factor $c\in[0,1]$, which determines the amount of borrowing. For general $k$, we use the $k$-dimensional equivalent to the constraint in Figure \ref{fig:Two-Population Case} and require that the constrained vector $\bpi_r$ has all $\pi_{r, j}\in[0,1]$ and at least one $\pi_{r, j} = 1$. This restricts $\bpi_r$ to the surface of a unit hypercube, including only surfaces where at least one $\pi_{r, j} = 1$. 
We thus define the reparameterized estimator
\begin{align}\label{eqn:hbbeta_cscale}
    &\hbbeta(c, \bpi_r) = (\bI_{pk} - c\bPi_r)\tbbeta +  c\bPi_r \bK\hbtheta(\bpi_r) = [\bI_{pk} + c\bPi_r\{\bK\bA(\bpi_r) - \bI_{pk}\}]\tbbeta, 
\end{align}
where $\bPi_r=\text{block-diag}\{\pi_{r,1} \bI_p, \ldots, \pi_{r,k} \bI_p\} \in \mathbb{R}^{pk \times pk}$. The restriction on $\bpi_r$ is necessary as it identifies the estimator $\hbbeta(c, \bpi)$, but the choice of constraint is a matter of convenience and, regardless of the constraint, $c= \max_j \pi_j$. We also note that there are $k$ values of $\bpi_r$ that give identical estimates for the full vector $\hbbeta(c, \bpi_r)$; $\hbbeta(c, \bpi) \equiv \tbbeta$ for any $\bpi_r$ with a single element $\pi_{r,j} = 1$ and all other $\pi_{r,j^\prime} = 0$, $j \neq j^\prime, j,j^\prime \in \{1, \ldots, k\}$. These values for $\bpi_r$ are mathematically equivalent to analyzing each group separately, and by Theorem \ref{existence_ma}, are not an optimal choice of $\bpi_r$.

With this reparameterization in hand, we can rewrite the MSE of our estimator as a quadratic function of the scaling parameter $c$: 
\begin{align}\label{eqn:mse_with_c}
    \MSE\{\hbbeta(c, \bpi_r)\} &= c^2~\|\bPi_r\{\bK\bA(\bpi_r) - \bI_{pk}\}\bbeta\|^2  +  c^2~\tr(\Var [\bPi_r\{\bK\bA(\bpi_r) - \bI_{pk}\}\tbbeta]) \nonumber\\
    &\quad + 2c~\tr(\Cov[\tbbeta, \bPi_r\{\bK\bA(\bpi_r) - \bI_{pk}\}\tbbeta]) + \tr\{\Var(\tbbeta)\}
\end{align}
We show in Theorem \ref{optimal_scaling} that there is a closed-form solution for the scaling factor $c$ that minimizes the MSE given $\bpi_r$, the data, and the parameter $\bbeta$.

\begin{theorem} \label{optimal_scaling}
For any constrained $\bpi_r$ such that $\hbbeta(c, \bpi_r) \neq \tbbeta$, there exists a set $\{c^{\dagger} : c^{\dagger}>0\}$ such that $\MSE\{\hbbeta(c^{\dagger}, \bpi_r)\}< \MSE(\tbbeta)$. Given $\bpi_r$, the optimal $c$ is
\begin{align}\label{eqn:c_star}
    c^\star(\bpi_r) &= \min\Bigl[\frac{-\tr(\Cov[\tbbeta, \bPi_r\{\bK\bA(\bpi_r) - \bI_{pk}\}\tbbeta])}{\tr(\Var[\bPi_r\{\bK\bA(\bpi_r) - \bI_{pk}\}\tbbeta]) + \|\bPi_r\{\bK\bA(\bpi_r) - \bI_{pk}\}\bbeta\|^2}, 1\Bigr].
\end{align}
\end{theorem}

The proof is given in Section \ref{appsec:optimal_scaling} of the supplement. 
The result of Theorem \ref{optimal_scaling} consists of two parts. The first part confers on our estimator similar properties to those of James-Stein shrinkage estimators \citep{james_estimation_1961} and the more recent information-based shrinkage estimator proposed by \citet{hector_turning_2024}: as the optimal scaling is strictly positive, the optimal version of our estimator has a smaller MSE than the MLE. The second part is more surprising: this result holds regardless of which $\bpi_r$ defines the centroid. In words, it is always beneficial in a MSE-sense, relative to the MLE, to borrow information across groups, regardless of the interpretation, value and, indeed, existence of a homogeneous parameter!
 
In Corollary \ref{improvement_bpi} below, we identify the complete set of scaling factors $c$ over which $\hbbeta(c, \bpi_r)$ offers an improvement in MSE over the MLE. 

\begin{corollary}\label{improvement_bpi}
 For a given $\bpi_r$, $\MSE\{\hbbeta(c, \bpi_r)\}< \MSE(\tbbeta)$ if and only if $c < 2 c^\star(\bpi_r)$.
\end{corollary}

A proof of Corollary \ref{improvement_bpi} is given in Section \ref{improvement_proof} of the supplement. For each $\bpi_r$, our estimator outperforms the MLE for scaling factors $c$ up to twice the value of $c^{\star}(\bpi_r)$. This complete characterization of the improvement set gives us hope that a reasonable estimate for $c^{\star}(\bpi_r)$ can also improve the MSE. Finally, we show in Theorem \ref{consistency} and Corollary \ref{asymp_distr} below that selecting the scaling factor $c^{\star}(\bpi_r)$ that minimizes the MSE yields an estimator that is consistent and asymptotically normal, enabling asymptotically valid inference.

\begin{theorem}\label{consistency}
    For groups $j\in \{1,\dots, k\}$ assume $n_j = n a_j$, where $a_j\in \mathbb{R}^+$. Assume also that, for each group $j$, 
    \begin{align*}
    \Bigl(\frac{1}{n_j \sigma_j^2}\bX_j^\prime\bX_j \Bigr)^{-1} \to \bQ_j^{-1}
    \quad \text{and}\quad 
    \frac{1}{n_j \sigma_j^2}\bX_j^\prime\bepsilon_j \overset{p}{\to} \bzero,
    \end{align*}
    where $\bQ_j^{-1}$ is a positive definite matrix. Then, for any $\bpi_r$ and any $c^{\star}(\bpi_r)$ as defined in Theorem \ref{optimal_scaling}, $\hbbeta(c^{\star}, \bpi_r)$ is a consistent estimator of $\bbeta$ as $n\to\infty$.
\end{theorem}

\begin{theorem}\label{asymp_distr}
    Assume the conditions in Theorem \ref{consistency} hold. Then 
    \begin{align*}
        &(\nabla_{\bbeta} \EV\{\hbbeta(c^{\star}, \bpi_r)\}(\scaledXTX)^{-1}\nabla_{\bbeta} \EV\{\hbbeta(c^{\star}, \bpi_r)\}')^{-1/2}
        \{\hbbeta(c^{\star}, \bpi_r) - \bbeta\} \overset{d}{\to} \mathcal{N}( \bzero, \bI_{pk}),
        \end{align*}
        where \[\nabla_{\bbeta} \EV\{\hbbeta(c^{\star}, \bpi_r)\} = \bI_{pk} + c^{\star}(\bpi_r)\bPi_r\{\bK\bA(\bpi_r) -\bI_{pk}\}.\]
\end{theorem}

We outline the argument in the proof of Theorem \ref{consistency} and give a complete proof in Section \ref{appsec:consistency} of the supplement. First, for any fixed $\bpi_r$, the total variance of $\hbbeta\{c^{\star}(\bpi_r), \bpi_r\}$ goes to zero as $n\to\infty$. To attain consistency, we need only for the bias of $\hbbeta\{c^{\star}(\bpi_r), \bpi_r\}$, given by $c^{\star}(\bpi_r) \bPi_r\{\bK\btheta(\bpi_r) - \bbeta\}$, to also go to zero. We can consider three cases for this condition. In the first, if all group-specific parameters are equal ($\bbeta_j = \bbeta_{j^\prime}$), then the bias will be zero at any sample size and for any $\bpi_r$. This is because in the scenario with full homogeneity, there is no drawback to combining the groups. In the second case, suppose $\|\bPi_r\{\bK\btheta(\bpi_r) - \bbeta\}\|^2$ does not go to zero as $n\to\infty$: there is heterogeneity among all groups that does not vanish as the sample size grows. In this case, we show that $c^{\star}(\bpi_r)$ goes to zero as $n\to\infty$. This is desirable behavior because, intuitively, the optimal scaling factor and therefore the amount of borrowing shrinks as the sample size increases when there is heterogeneity between all groups. In the final case, $\|\bPi_r\{\bK\btheta(\bpi_r) - \bbeta\}\|^2$ goes to zero as $n\to\infty$. We need not draw a conclusion about $c^{\star}(\bpi_r)$ in this scenario, because this condition implies that the bias will also go to zero. This event only occurs if either $\pi_{r,j} = 0$ or if $\btheta(\bpi_r) = \bbeta_j$ for each group $j$; that is, either our estimator leaves group $j$ out of the centroid or $\bbeta_j$ is equal for all groups included within the centroid. 

It follows from Corollary \ref{asymp_distr} that an asymptotically valid $100(1 - \alpha)$\% confidence interval for the $\ell$\textsuperscript{th} parameter, $\beta_{\ell}$, of $\bbeta$ is
\begin{align}\label{eqn:asymp_CI}
\Bigl(&\widehat{\beta}_{\ell}(c^{\star}, \bpi) \pm z_{1-\alpha/2}\sqrt{[\nabla_{\bbeta} \EV\{\hbbeta(c^{\star}, \bpi)\}(\scaledXTX)^{-1} \nabla_{\bbeta} \EV\{\hbbeta(c^{\star}, \bpi)\}']_{\ell \ell}}\Bigr). 
\end{align} 
When working with a subset of the covariates as described in Section \ref{subsec:est_subset}, the asymptotically valid $100(1-\alpha)\%$ confidence interval for $\bbeta_{\ell}$ is nearly identical to that given in equation \eqref{eqn:asymp_CI}. The only adjustment needed is to use $\scaledXTMX$ in place of $\scaledXTX$ in the calculation of the estimate and its variance.

Put together, the results of this section confer on our estimator several surprising and beneficial theoretical properties. First, selecting a vector of shrinkage parameters with the goal of minimizing the MSE yields an estimator with asymptotically valid inference that outperforms the MLE in finite samples. Even more surprising, consistency and asymptotic normality hold for any $\bpi_r$ with scaling $c^{\star}(\bpi_r)$; that is, we do not require the MSE-minimizing choice for the centroid $\btheta(\bpi_r)$ to achieve these properties, so long as the appropriate scaling is selected. This increases our likelihood of success in estimating the $k$ shrinkage parameters. Rather than needing to correctly select all $k$ parameters to minimize the MSE, we can select $k-1$ freely and are only required to choose the $k^{\textrm{th}}$ correctly to guarantee asymptotic inference.

\section{Bias-reduced shrinkage parameter selection}\label{sec:implementation}
 
Our goal is to estimate $\bbeta$ using the estimator $\hbbeta(c, \bpi_r)$ with $c$ and $\bpi_r$ chosen so as to minimize the MSE displayed in equation \eqref{eqn:mse_with_c}, but this MSE depends on $\bbeta$, which is not available in practice. A reasonable approach would be to select $c$ and $\bpi_r$ by minimizing Stein's unbiased risk estimator (SURE). In this section, we first show that this leads to over-borrowing, and we then propose an alternative objective function for data-driven selection of the shrinkage parameters that mitigates this risk.

The MSE of $\hbbeta(c, \bpi_r)$, displayed in equation \eqref{eqn:mse_with_c}, depends on $\bbeta$ only through the term $c^2~ \|\bPi_r\{\bK\bA(\bpi_r) - \bI_{pk}\}\bbeta\|^2$. Estimating this term by na\"ively plugging in the MLE for $\bbeta$ yields a positively biased estimator for the MSE,
\begin{align}\label{eqn:biased_mse}
 \BMSE\{\hbbeta(c, \bpi_r)\} &=  c^2~ \|\bPi_r\{\bK\bA(\bpi_r) - \bI_{pk}\}\tbbeta\|^2    +  c^2 ~\tr(\Var [\bPi_r\{\bK\bA(\bpi_r) - \bI_{pk}\}\tbbeta]) \nonumber\\
     &\quad + 2c ~\tr(\Cov[\tbbeta, \bPi_r\{\bK\bA(\bpi_r) - \bI_{pk}\}\tbbeta])  + \tr\{\Var(\tbbeta)\}.
\end{align}
Indeed, we show in Section \ref{appsec:expected_BMSE} of the supplement that the expected value of $\BMSE$ is $\MSE\{\hbbeta(c, \bpi_r)\} + c^2~ \tr(\Var [\bPi_r\{\bK\bA(\bpi_r) - \bI_{pk}\}\tbbeta])$. The shrinkage parameters that minimize $\BMSE\{\hbbeta(c, \bpi_r)\}$ are not those that minimize the MSE.

In the spirit of \citet{stein_estimation_1981}, the unbiased risk estimator (URE), i.e. the unbiased estimator of the MSE, is
\begin{align}\label{eqn:ure}
     \ure\{\hbbeta(c, \bpi_r)\} &=  c^2~ \|\bPi_r\{\bK\bA(\bpi_r) - \bI_{pk}\}\tbbeta\|^2 
     + 2c~ \tr(\Cov[\tbbeta, \bPi_r\{\bK\bA(\bpi_r) - \bI_{pk}\}\tbbeta]) \nonumber\\
     &\quad + \tr\{\Var(\tbbeta)\}.
\end{align}
The full derivation is shown in Section \ref{appsec:umse_derivation} of the supplement.
The URE is an unbiased estimator of the MSE and therefore its minimizers are asymptotically unbiased estimators of the minimizers of the MSE under mild regularity conditions. Due to their finite-sample bias, however, these minimizers lead to over-borrowing on average and therefore sub-optimal performance of $\hbbeta(c, \bpi_r)$. To see this, by Jensen's inequality, 
\begin{align*}
    \EV(\widehat{c}^\star) &= \EV\Bigl(\frac{-\tr(\Cov[\tbbeta, \bPi_r\{\bK\bA(\bpi_r) - \bI\}\tbbeta])}{ \|\bPi_r\{\bK\bA(\bpi_r) - \bI\}\tbbeta\|^2 }\Bigr) \geq \frac{-\tr(\Cov[\tbbeta, \bPi_r\{\bK\bA(\bpi_r) - \bI\}\tbbeta])}{ \EV[\|\bPi_r\{\bK\bA(\bpi_r) - \bI\}\tbbeta\|^2] } = c^\star.
\end{align*}
where $c^\star$ and $\widehat{c}^\star$ respectively denote the values of $c$ that minimize the MSE and URE, omitting dependence on $\bpi_r$.

\begin{figure}[ht!]
    \centering
    \begin{tikzpicture}[line cap=round,line join=round, x=10cm, y=8cm]%
    \clip(-0.1,-0.3) rectangle (1.1,0.5);
    \draw[thick,->](0,0) -- (1, 0) node[below]{$c$};
    \draw[thick,<->](0,-.3)--(0,.5) node[anchor = north west, font = \small]{Derivative of Loss Function};
    \draw [line width=2.pt] (0.,-0.2) -- (0.3,0.4);
    \draw [line width=2.pt,dash pattern=on 1pt off 3pt on 1pt off 4pt] (0.,-0.2) -- (0.7,0.15);
    \draw [line width=1.pt] (0.1,0.)-- (0.1,-0.15);
    \draw [line width=2.pt,dash pattern=on 1pt off 2pt] (0.,-0.05) -- (0.7,0.3);
    \draw (0.29,0.45) node[anchor=north west, font = \large] {$\textstyle \frac{\partial \MSE\{\hbbeta(c, \bpi_r)\}}{\partial c}$};
    \draw (0.69,0.21) node[anchor=north west, font = \large] {$\textstyle \frac{\partial \ure\{\hbbeta(c, \bpi_r)\}}{\partial c}$};
    \draw (0.69,0.36) node[anchor=north west, font = \large] {$\textstyle \frac{\partial \ure\{\hbbeta(c, \bpi_r)\}}{\partial c} +\delta$};
    \draw (-0.08, 0.04) node[anchor=north west] {\scriptsize (0,0)};
    \draw (0.09,-0.03) node[anchor=north west] {$\delta$};
    \draw (0.06,0.07) node[anchor=north west] {$c^\star$};
    \draw (0.36,0.07) node[anchor=north west] {$\hat{c}^\star$};
\end{tikzpicture}%
    \caption{A demonstration of the proposed correction to the URE. The solid line yields the optimal scaling, $c^{\star}$, as the root. The lower dashed line over-borrows, on average. Shifting the lower line up by $\delta$ to the parallel dashed line above it gives a root at $c^{\star}$.}
    \label{fig:deriv_plot}
\end{figure}

We propose to reduce over-borrowing with a strategy adapted from \citet{firth_bias_1993} and illustrated in Figure \ref{fig:deriv_plot}. We first describe the main idea. The derivative of the URE has a root at $\hc^{\star}$ that is positively biased for $c^\star$. Therefore, for some $\delta>0$, $[\partial\ure\{\hbbeta(c, \bpi_r)\}]/(\partial c) + \delta$ has a root at $c^\star$. In words, the optimal borrowing $c^{\star}$ is obtained by shifting the URE derivative curve up by $\delta$. The shift needed to obtain a root at $c^\star$ is
\begin{align*}
\delta = \frac{\partial^2 }{\partial c^2}\ure\{\hbbeta(c, \bpi_r)\} (\hc^{\star} - c^{\star}) = 2\|\bPi_r\{\bK\bA(\bpi_r) - \bI\}\tbbeta\|^2 (\hc^{\star} - c^{\star}).
\end{align*}
As $\hc^\star - c^\star$ is unknown, we propose to estimate it with $\hc^\star - \tilde{c}$, where $\tilde{c}$ minimizes BMSE in equation \eqref{eqn:biased_mse}. The root of this corrected derivative will reduce borrowing and be closer, on average, to the optimal $c^{\star}$.
We show in Section \ref{appsec:pseudo_MSE_derivation} of the supplement that finding the root of this corrected derivative and selecting a MSE-minimizing $\bpi_r$ can be achieved by minimizing the loss function
\begin{align}\label{eqn:pseudo_mse}
    \|\bPi\{\bK\bA(\bpi) - \bI\}\tbbeta\|^2  - 2 \frac{\|\bPi\{\bK\bA(\bpi) - \bI\}\tbbeta\}\|^2\times \tr(\Cov[\tbbeta, \bPi\{\bK\bA(\bpi) - \bI\}\tbbeta])}{\tr(\Var[
\bPi\{\bK\bA(\bpi) - \bI\}\tbbeta]) +
\|\bPi\{\bK\bA(\bpi) - \bI\}\tbbeta\|^2}
\end{align} 
as a function of the unconstrained $\bpi$. 

We refer to equation \eqref{eqn:pseudo_mse} as the \textit{surrogate MSE}, so-named to recognize that it is designed and used only for selecting the shrinkage parameters.
The asymptotic results discussed in Section \ref{sec:theory} hold under shrinkage parameter selection with the surrogate MSE, as the minimizer of the surrogate MSE converges to $c^{\star}$ as $n\to\infty$. 
We term our final proposed estimator the \textit{heterogeneity-adaptive shrinkage estimator} (HAS estimator) $\hbbeta(\bpi_{\has})$ with the shrinkage parameters, $\bpi_{\has} = (\pi_{j,\has})_{j=1}^k$, minimizing equation \eqref{eqn:pseudo_mse}.  

We conduct a small simulation study to illustrate the performance of $\hbbeta(\bpi)$ under different shrinkage parameter selection mechanisms. We select $\bpi$ by minimizing the true MSE, the URE, the BMSE, and the surrogate MSE, and denote the resulting estimators as $\hbbeta(\bpi_{\MSE})$, $\hbbeta(\bpi_{\ure})$, $\hbbeta(\bpi_{\BMSE})$, and $\hbbeta(\bpi_{\has})$, respectively. When minimizing the true MSE, we assume that $\bSigma$ is known, and for the remaining measures, $\bSigma$ is estimated by its MLE. We compare the analytical MSE of each estimator to the analytical MSE of the individual group MLEs, where the analytical MSEs are computed using equation \eqref{eqn:mse_with_c}. 
We consider $k=3$ groups with $p=4$ covariates. We generate the $\bbeta_j$'s from a multivariate normal distribution with shared mean $\bbeta_m$ and covariance matrix $0.1\bI_4$, and round to two significant digits. We consider four combinations of sample sizes $(n_1,n_2,n_3)$: ($100$, $100$, $100$), ($100$, $100$, $500$), ($100$, $500$, $500$) and ($500$, $500$, $500$). Covariates $\bX_j$ are sampled following $\bX_{j1}=\bone$, $\bX_{j2}\sim \mathcal{N}(\bzero, \bI_{n_j}), \bX_{j3}\sim \text{Bin}(n_j, 0.5), \bX_{j4} = \bX_{2j} \bX_{3j}$. The parameters $\bbeta_j$ and covariates $\bX_j$ are held constant across Monte Carlo replicates. Within each Monte Carlo replicate, we generate errors $\bepsilon_{j} \sim \mathcal{N}(\bzero, 0.25\bI_{n_j})$. 

\begin{table}[ht]
\centering
\resizebox{\linewidth}{!}{
\begin{tabular}{l|r|rrr|rrr|rrr|rrr} 
 &   \multicolumn{1}{c|}{$\tbbeta$} & \multicolumn{3}{c}{$\hbbeta(\bpi_{\MSE})$} & \multicolumn{3}{|c}{$\hbbeta(\bpi_{\ure})$} & 
 \multicolumn{3}{|c}{$\hbbeta(\bpi_{\BMSE})$} &  \multicolumn{3}{|c}{$\hbbeta(\bpi_{\has})$} \\
 $n$ &  MSE & MSE & PL & Ratio & MSE & PL & Ratio & MSE & PL & Ratio & MSE & PL & Ratio \\ 
  \hline
100, 100, 100 & 9.62 & 8.83 & 0 & 0.92 & 8.97 & 16.30 & 0.93 & 8.93 & 5.10 & 0.93 & 8.93 & 3.40 & 0.93 \\ 
  100, 500, 100 & 6.76 & 5.83 & 0 & 0.86 & 5.93 & 10.50 & 0.88 & 5.91 & 1.80 & 0.87 & 5.91 & 1.40 & 0.87 \\ 
  100, 500, 500 & 4.78 & 4.14 & 0 & 0.87 & 4.19 & 8.70 & 0.88 & 4.19 & 2.50 & 0.88 & 4.19 & 3.00 & 0.88 \\ 
  500, 500, 500 & 1.75 & 1.72 & 0 & 0.98 & 1.72 & 1.60 & 0.98 & 1.72 & 0.80 & 0.98 & 1.72 & 0.80 & 0.98 \\ 
  All & -- & -- & 0 & 0.89 & -- & 9.28 & 0.93 & -- & 2.55 & 0.92 & -- & 2.15 & 0.92  
\end{tabular}}
\caption{``MSE'' columns give the median MSE multiplied by 100; ``PL'' columns show the percent of replicates in which $\hbbeta(\bpi)$ has a larger MSE than $\tbbeta$; and ``Ratio'' columns give the median of the ratio of $\hbbeta(\bpi)$'s MSE to $\tbbeta$'s MSE.}
\label{tab:sim_study0}
\end{table}

Table \ref{tab:sim_study0} displays the median analytical MSE across $1000$ Monte Carlo replicates at each sample size setting, when selecting $\bpi$ by minimizing different values. As the theory guarantees, $\hbbeta(\bpi_{\MSE})$ has a smaller MSE than the MLE in all instances. Across all settings, $\hbbeta(\bpi_{\ure})$ has a larger MSE than the MLE in 9.28\% of the $1000$ Monte Carlo replicates. In contrast, $\hbbeta(\bpi_{\BMSE})$ and $\hbbeta(\bpi_{\has})$ have larger MSEs than the MLE in 2.55\% and 2.15\% of the $1000$ Monte Carlo replicates, respectively. Minimizing the surrogate MSE is at its most beneficial when all groups have smaller sample sizes. In the condition with all sample sizes equal to 100, $\hbbeta(\bpi_{\ure})$ is over four times more likely than $\hbbeta(\bpi_{\has})$ to have a larger MSE than the MLE.

\section{Numerical support}\label{sec:sim}

\subsection{Overview}

In this section, we investigate the performance of the HAS estimator through four simulation settings. In each setting, we compare the HAS estimator $\hbbeta_{\has} = \hbbeta(\bpi_{\has})$ to the vectors of group-specific MLEs $\tbbeta$ and to \cite{kwon_optimal_2026}'s proposed URE-minimizing estimator, $\hbbeta_{\ure}$. We also develop a ridge-like penalized regression competitor $\hbbeta_{\ridge}$ in Section \ref{appsec:ridgelike} of the supplement that uses a Euclidean-distance (instead of KLD) penalty to encourage similar estimates of $\bbeta$, similar in spirit to the fused ridge estimator \citep{lettink_two-dimensional_2023}, which has no inferential guarantees. In each setting, we evaluate $\hbbeta_{\has}, \tbbeta$, $\hbbeta_{\ridge}$, $\hbbeta_{\ure}$ in $1000$ Monte Carlo replicates. We report the empirical MSE (eMSE) and coverage of 95\% confidence intervals using equation \eqref{eqn:asymp_CI} for the HAS estimator and the large sample distribution of the MLE for $\tbbeta$. We estimate $\hbbeta_{\ure}$ using the FEShR package written by \cite{kwon_optimal_2026}, which constrains the pooled estimator in a way that frequently leads to no solutions existing. In our simulation and data analysis below, we modify their code to constrain the elements of the pooled estimator to lie between the minimum and maximum MLEs. We omit the small number of remaining cases where a solution does not exist from our reporting of results.

\subsection{Setting 1: Varying sample size and number of parameters}\label{subsec:sim1}

We consider a low-heterogeneity setting with $k = 3$ groups, $p\in\{2, 4, 10, 20\}$ and ten unique combinations of $n_j\in\{100, 200, 300\}$. Within each Monte Carlo replicate, we generate $n_j$ observations $\bX_{j}$ and error terms $\bepsilon_j\sim \mathcal{N}(\bzero, \bI_{n_j})$: when $p = 2$, $\bX_{1j}=\bone$, $\bX_{2j}\sim \mathcal{N}(\bzero, \bI_{n_j})$; when $p = 4$, we further include $X_{3ij}\sim \text{Bern}(0.5)$, $X_{4ij} = X_{2ij} X_{3ij}$, $i=1,\dots, n_j$; when $p = 10$ and $p = 20$, the additional covariates are simulated independently from standard normal distributions.

The process for generating group-specific parameters intentionally violates a key assumption of the normal-normal empirical Bayesian model: while the group parameters are generated from the same distribution, they do not come from a normal distribution. Specifically, for each $(n_1, n_2, n_3, p)$ combination, we generate $\bbeta_j = \bbeta_m + \boldsymbol{u}_j$ with $\boldsymbol{u}_j \in \mathbb{R}^p$ independent draws from $\text{Unif}(-0.25, 0.25)$ rounded to two decimal places. The values of $\bbeta_j$ are held constant across Monte Carlo replicates in each combination; values of $\bbeta_j$ are included in Table \ref{tab:setting1_parm_values} of the supplement. Table \ref{tab:sim1_emse} displays the eMSE of the estimators $\tbbeta$, $\hbbeta_{\has}$, $\hbbeta_{\ridge}$, and $\hbbeta_{\ure}$. Across all combinations, the HAS estimator has a smaller eMSE than the MLE and the URE estimator.  As sample sizes increase, the ratio of the eMSE for the HAS estimator to the eMSE for the MLE gets closer to 1; at larger sample sizes, the heterogeneity of the datasets is more easily detected and the HAS estimator reverts to the group-specific MLEs, preventing over-borrowing and enabling inference. The HAS estimator also yields a smaller eMSE than the ridge-like estimator across the majority of settings, highlighting the benefits of using the KLD over the Euclidean distance.

\begin{table}[h!]
\centering \centering \resizebox{\linewidth}{!}{
\begin{tabular}{l|rrrr|rrrr|rrrr|rrrr}
  & \multicolumn{4}{|c}{$p=2$} & \multicolumn{4}{|c}{$p=4$} & \multicolumn{4}{|c}{$p=10$} & \multicolumn{4}{|c}{$p=20$}\\
$n_1, n_2, n_3$ & $\tbbeta$ & $\hbbeta_{\ridge}$ & $\hbbeta_{\ure}$ &  $\hbbeta_{\has}$ 
                & $\tbbeta$ & $\hbbeta_{\ridge}$ & $\hbbeta_{\ure}$ &  $\hbbeta_{\has}$
                & $\tbbeta$ & $\hbbeta_{\ridge}$ & $\hbbeta_{\ure}$ &  $\hbbeta_{\has}$
                & $\tbbeta$ & $\hbbeta_{\ridge}$ & $\hbbeta_{\ure}$ &  $\hbbeta_{\has}$\\ 
  \hline
  100, 100, 100 & 6.2 & 4.7 & 5.5 & \textbf{4.5} & 37.8 & 28.5 & 28.6 & \textbf{27.4} & 61.1 & 41.5 & 52.6 & \textbf{40.2} & 106.6 & \textbf{77.5} & 100.6 & 78.0 \\ 
  100, 200, 100 & 5.1 & 4.2 & 4.8 & \textbf{3.8} & 32.6 & 25.4 & 25.6 & \textbf{24.4} & 50.8 & 34.7 & 44.3 & \textbf{33.9} & 87.0 & \textbf{64.2} & 82.5 & 64.8 \\ 
  100, 200, 200 & 4.1 & 3.6 & 3.9 & \textbf{3.2} & 24.5 & 18.7 & 19.2 & \textbf{18.5} & 39.6 & 28.7 & 35.0 & \textbf{27.3} & 67.3 & 52.8 & 64.6 & \textbf{52.5} \\ 
  100, 300, 100 & 4.8 & 3.9 & 4.5 & \textbf{3.5} & 30.2 & 22.9 & 24.1 & \textbf{22.2} & 46.8 & 31.9 & 41.2 & \textbf{31.1} & 82.0 & \textbf{59.9} & 78.2 & 60.6 \\ 
  100, 300, 200 & 3.6 & 3.3 & 3.5 & \textbf{2.8} & 22.0 & 16.8 & 17.5 & \textbf{16.6} & 36.4 & 26.8 & 32.5 & \textbf{25.4} & 61.0 & 48.2 & 58.8 & \textbf{47.6} \\ 
  100, 300, 300 & 3.4 & 3.2 & 3.3 & \textbf{2.7} & 20.7 & \textbf{15.3} & 16.8 & 15.4 & 32.6 & 24.4 & 29.6 & \textbf{22.8} & 55.4 & 44.0 & 53.5 & \textbf{43.2} \\ 
  200, 200, 200 & 3.0 & 2.8 & 2.9 & \textbf{2.5} & 18.8 & 15.8 & 15.6 & \textbf{15.2} & 28.7 & 22.6 & 25.8 & \textbf{21.6} & 47.5 & \textbf{40.3} & 45.9 & 40.4 \\ 
  200, 300, 200 & 2.7 & 2.6 & 2.6 & \textbf{2.3} & 16.1 & 13.6 & 13.7 & \textbf{13.1} & 25.3 & 20.1 & 22.9 & \textbf{19.3} & 41.4 & \textbf{35.3} & 40.2 & 35.4 \\ 
  200, 300, 300 & 2.3 & 2.3 & 2.3 & \textbf{2.0} & 14.3 & 12.0 & 12.0 & \textbf{11.7} & 22.3 & 18.4 & 20.5 & \textbf{17.5} & 35.2 & 31.2 & 34.4 & \textbf{30.8} \\ 
  300, 300, 300 & 2.0 & 1.9 & 2.0 & \textbf{1.7} & 12.2 & 10.9 & 10.7 & \textbf{10.4} & 18.8 & 16.0 & 17.4 & \textbf{15.3} & 30.1 & 27.1 & 29.5 & \textbf{27.0} \\  
\end{tabular} }
\caption{Empirical MSEs for each estimator multiplied by 100 for Setting 1. The smallest eMSE for each $(n_1, n_2, n_3, p)$ combination is bolded.}
\label{tab:sim1_emse}
\end{table}

Table \ref{tab:sim1_coverage} reports coverage rates of 95\% large-sample confidence intervals using $\hbbeta_{\has}$. Group-level coverages grow closer to the nominal level with increasing $p$. As the number of covariates increases, our method more easily distinguishes between the group outcome distributions and appropriately decreases the amount of information sharing. Increasing the sample size for one group while holding others constant increases coverage for that group, but can decrease coverage for the others. For example, moving from group sizes $(100, 100, 100)$ to $(100, 200, 100)$ increases the coverage for group 2 but decreases the coverage for group 3 because the centroid shifts towards group 2's outcome distribution, which moves $\hbbeta_{3,\has}$ further from $\bbeta_3$.

\begin{table}[h!]
\centering
 \resizebox{\linewidth}{!}{
\begin{tabular}{r|rrr|rrr|rrr|rrr}
   & \multicolumn{3}{c}{$p = 2$} &  \multicolumn{3}{|c}{$p = 4$} &  \multicolumn{3}{|c}{$p = 10$} & \multicolumn{3}{|c}{$p = 20$}\\
 \multicolumn{1}{r|}{$n$} & $\bbeta_1$ & $\bbeta_2$ & $\bbeta_3$ & $\bbeta_1$ & $\bbeta_2$ & $\bbeta_3$ &  $\bbeta_1$ & $\bbeta_2$ & $\bbeta_3$  &  $\bbeta_1$ & $\bbeta_2$ & $\bbeta_3$  \\  \hline
100, 100, 100 & 87.8 & 93.2 & 89.5 & 93.0 & 90.6 & 90.3 & 93.4 & 91.7 & 92.8 & 92.4 & 91.8 & 92.0 \\ 
  100, 200, 100 & 87.4 & 93.4 & 88.7 & 92.7 & 91.2 & 86.6 & 92.4 & 92.4 & 91.0 & 92.0 & 92.8 & 91.5 \\ 
  100, 200, 200 & 84.8 & 93.8 & 91.0 & 92.6 & 90.4 & 89.9 & 92.8 & 91.4 & 92.3 & 91.7 & 93.0 & 92.5 \\ 
  100, 300, 100 & 85.9 & 94.6 & 87.2 & 92.3 & 92.0 & 86.3 & 91.7 & 93.0 & 89.9 & 92.0 & 93.5 & 91.3 \\ 
  100, 300, 200 & 85.2 & 94.7 & 89.5 & 93.3 & 90.6 & 88.8 & 92.6 & 91.6 & 90.5 & 91.8 & 93.3 & 92.3 \\ 
  100, 300, 300 & 85.0 & 93.7 & 90.7 & 92.7 & 88.5 & 87.8 & 91.9 & 90.7 & 91.3 & 90.9 & 93.2 & 93.0 \\ 
  200, 200, 200 & 89.2 & 92.8 & 88.2 & 92.3 & 90.5 & 90.1 & 93.4 & 91.7 & 92.4 & 92.4 & 92.4 & 92.8 \\ 
  200, 300, 200 & 86.9 & 93.0 & 88.4 & 93.0 & 92.0 & 90.1 & 92.6 & 92.0 & 91.5 & 92.3 & 93.7 & 92.9 \\ 
  200, 300, 300 & 87.0 & 93.5 & 90.0 & 92.0 & 90.5 & 89.9 & 92.5 & 91.5 & 92.1 & 92.9 & 93.5 & 93.2 \\ 
  300, 300, 300 & 89.6 & 91.5 & 88.9 & 92.3 & 91.5 & 90.5 & 92.8 & 92.1 & 91.6 & 92.9 & 93.7 & 93.3 
\end{tabular}
}
\caption{Average coverage rate (\%) of 95\% confidence intervals using $\hbbeta_{\has}$ in Setting 1.}
\label{tab:sim1_coverage}
\end{table}

\subsection{Setting 2: Varying number of groups and heterogeneity}\label{subsec:sim2}

We now examine the HAS estimator's performance with $p = 4$ parameters as the number of groups $k\in\{5, 10, 15\}$ increases under varying parameter heterogeneity conditions. In the first condition, there is no heterogeneity and $\bbeta_j \equiv \bbeta_m$, $j=1, \ldots, k$. In the second, the $\bbeta_j$ display mild heterogeneity with $\bbeta_j \sim \mathcal{N}(\bbeta_m, .1\bI_4)$. The $\bbeta_j$ in the third condition are moderately heterogeneous: $\bbeta_j \sim \mathcal{N}(\bbeta_m, .5\bI_4)$. Finally, the $\bbeta_j$ in the fourth condition have a mixture of no heterogeneity and high heterogeneity: $\bbeta_j \equiv \bbeta_m$ for $j=1,2,3$ and $\bbeta_j \sim \mathcal{N}(\bbeta_m, \bI_4)$ for $j=4, \ldots, k$. We emulate the group parameter generation process assumed for the normal-normal empirical Bayesian model by drawing new values of $\bbeta_j$ with each Monte Carlo replicate. For each Monte Carlo replicate, we generate $n_j = 200$ observations $\bX_{j}$ and error terms, $\bepsilon_j\sim \mathcal{N}(\bzero, \bI_{n_j})$ following $\bX_{1j}=\bone$, $\bX_{2j}\sim \mathcal{N}(\bzero, \bI_{n_j})$, $X_{3ij}\sim \text{Bern}(.5)$ and $X_{4ij} = X_{2ij} X_{3ij}$, $i = 1,\dots, n_j$. 

Table \ref{tab:sim2_mse} displays the eMSE of $\tbbeta$, $\hbbeta_{\ridge}$, $\hbbeta_{\ure}$, and $\hbbeta_{\has}$ and the coverage rate of $\hbbeta_{\has}$'s large-sample 95\% confidence interval. Across almost all combinations of heterogeneity and $k$, $\hbbeta_{\has}$ has smaller eMSE than $\tbbeta$ and $\hbbeta_{\ure}$. The HAS estimator's eMSE is larger than $\hbbeta_{\ridge}$'s with no, mild or moderate heterogeneity, but is smaller than $\hbbeta_{\ridge}$'s in the mixture condition due to the added flexibility of group-specific shrinkage parameters. A smaller eMSE is only desirable, however, if the inference is calibrated, and the ridge-like estimator has no inferential guarantees. The coverage is closest to its nominal value in the condition with complete homogeneity and is furthest in the condition with mild heterogeneity. In the mild heterogeneity condition, the group outcome distributions are close together, and our method consequently over-borrows. 

\begin{table}[h!]
\centering
\begin{tabular}{rr|rrrr|r} 
 &  & \multicolumn{4}{c|}{eMSE$\times 100$} & Coverage rate  (\%) \\ 
Heterogeneity& $k$ & $\tbbeta$ & $\hbbeta_{\ridge}$& $\hbbeta_{\ure}$ & $\hbbeta_{\has}$ & $\hbbeta_{\has}$ \\ 
  \hline
None & 5 & 30.3 & 8.3 & 15.7 & 12.9 & 94.4 \\ 
   & 10 & 62.1 & 8.5 & 24.6 & 21.2 & 94.4 \\ 
   & 15 & 92.2 & 7.4 & 31.2 & 28.4 & 95.0 \\ 
Mild & 5 & 30.5 & 27.1 & 31.5 & 28.4 & 91.1 \\ 
   & 10 & 61.7 & 52.4 & 62.0 & 56.3 & 91.0 \\ 
   & 15 & 92.7 & 78.3 & 87.7 & 84.3 & 90.8 \\ 
Moderate & 5 & 30.3 & 29.5 & 32.7 & 30.4 & 93.5 \\ 
   & 10 & 61.3 & 58.9 & 61.4 & 61.1 & 93.1 \\ 
   & 15 & 92.8 & 89.3 & 91.8 & 92.2 & 92.8 \\ 
Mixture & 5 & 30.9 & 29.6 & 27.3 & 24.7 & 92.7 \\ 
   & 10 & 61.4 & 59.9 & 61.8 & 55.9 & 92.8 \\ 
   & 15 & 91.7 & 89.3 & 90.5 & 85.7 & 93.3 
\end{tabular}
\caption{Empirical MSEs for each estimator (multiplied by 100) and coverage rate of 95\% confidence intervals using $\hbbeta_{\has}$ in setting 2.}
\label{tab:sim2_mse}
\end{table}

Examining the selected shrinkage parameters $\bpi_{\has}$ provides additional insight into the information-borrowing mechanism. Figure \ref{fig:sim2_pi} shows the distribution of $\bpi_{\has}$ for each heterogeneity condition across replicates when $k=15$. In the condition with no heterogeneity, the selected $\bpi_{\has}$ are centered around $0.5$, which is slightly more conservative than the MSE-minimizing $\pi_j$ of $1$. In the mixture condition, the $\pi_{\has,j}$ associated with the homogeneous groups are separated from the $\pi_{\has,j}$ associated with the heterogeneous groups, as desired. 

\begin{figure}[ht!]
    \centering
    \includegraphics[width=0.95\linewidth]{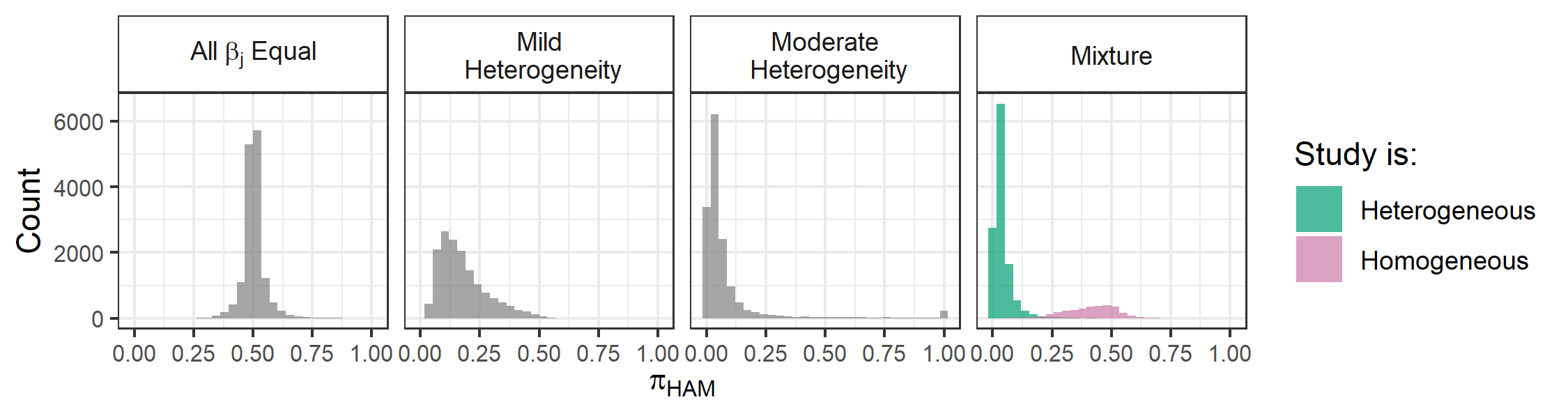}
    \caption{$\bpi_{\has}$ values across heterogeneity conditions when $k = 15$ in setting 2.}
    \label{fig:sim2_pi}
\end{figure}

\subsection{Setting 3: Panel data with overlapping covariates}\label{subsec:sim3} 
As described in Section \ref{subsec:est_subset}, our method can also provide shrinkage estimation for a single covariate of interest, after controlling for relevant, potentially disparate covariates at the group-level.
To demonstrate this, we consider $k=20$ groups ($n_j=200, j=1, \ldots, 20$) with $p=1$ and $q_j=4$ where all groups share the first covariate but the other covariates vary across groups. We consider two scenarios for the generation of $\beta_j$, the coefficient of interest. In the first scenario, $\beta_j$ is drawn from a single distribution, $\beta_j = 3.47 + u_j$ with $u_j \sim \text{Unif}(-0.25, 0.25)$. In the second, more challenging scenario, $\beta_j$ is drawn from one of two distributions, $\beta_j\sim\text{Unif}(-0.25, 0.25)$ for $j=1,\dots, 5$ and $\beta_j = 3.47 + u_j$, $u_j \sim \text{Unif}(-0.25, 0.25)$ for $j=6,\dots, 20$. Values of $\bgamma_j$ are independently drawn from $\mathcal{N}(\bzero, \bI_4)$. All parameter values are rounded to two decimal places and held constant over Monte Carlo replicates. The first, shared covariate is an intercept and the remaining covariates are simulated independently from standard normal distributions. Error terms are generated independently from standard normal distributions.

Table \ref{tab:sim3} displays the eMSE of $\tbbeta$, $\hbbeta_{\ridge}$, $\hbbeta_{\ure}$, and $\hbbeta_{\has}$. In both scenarios, the eMSE of $\hbbeta_{\has}$ is smaller than the eMSE of $\tbbeta$ and $\hbbeta_{\ure}$. When $\beta_j$ are drawn from a single distribution, $\hbbeta_{\has}$ and $\hbbeta_{\ridge}$ have similar eMSEs. As in Setting 2 (Section \ref{subsec:sim2}), $\hbbeta_{\has}$ outperforms $\hbbeta_{\ridge}$ when $\beta_j$ is drawn from one of two distributions.

\begin{table}[ht]
\centering
\begin{tabular}{r|rrrr|r}
  $\beta_j$ scenario &  \multicolumn{4}{c|}{eMSE$\times 100$} & Coverage rate  (\%) \\ 
  & $\tbbeta$ &  $\hbbeta_{\ridge}$ & $\hbbeta_{\ure}$ & $\hbbeta_{\has}$ & $\hbbeta_{\has}$\\
  \hline
Single Distribution & 10.4 & 8.1 & 9.5 & 8.2 & 83.8 \\ 
Two Distributions & 10.4 & 10.3 & 10.3 & 9.1 & 87.3 
\end{tabular}
\caption{Empirical MSEs for each estimator (multiplied by 100) and coverage of 95\% confidence intervals using $\hbbeta_{\has}$ in setting 3.}
\label{tab:sim3}
\end{table}

Table \ref{tab:sim3} also reports the coverage rate of $\hbbeta_{\has}$'s large-sample 95\% confidence interval. We investigate the observed undercoverage at the group-level in Figure \ref{fig:sim3_beta1}. We plot the coverage of $\widehat{\beta}_{\has, j}$ versus $\beta_j$; we use color to denote which quartile the average (across the 1000 Monte Carlo replicates) $\pi_{j, \has}$ belongs to, with the quartiles computed across all groups. When $\beta_j$ is drawn from a single distribution, groups with $\beta_j$ values closest to $3.47$, the center of the distribution, have the largest average $\pi_{j,\has}$ and coverage nearest to the nominal level. When $\beta_j$ is drawn from one of two distributions, the coverage is closer to the nominal level for the five groups with $\beta_j\sim \text{Unif}(-0.25, 0.25)$. These groups are far enough away from the others that their selected shrinkage parameters are close to zero. In both scenarios, coverage is lower for groups where the parameter is not too far from $3.47$. As a rule of thumb, our estimator tends to over-borrow in situations where it struggles to distinguish between groups because the MLEs are close to each other. Increasing the dimension $p$ or the sample size alleviates this issue, as seen in setting 1 (Section \ref{subsec:sim1}).

\begin{figure}[ht!]
    \centering
    \includegraphics[width=0.8\linewidth]{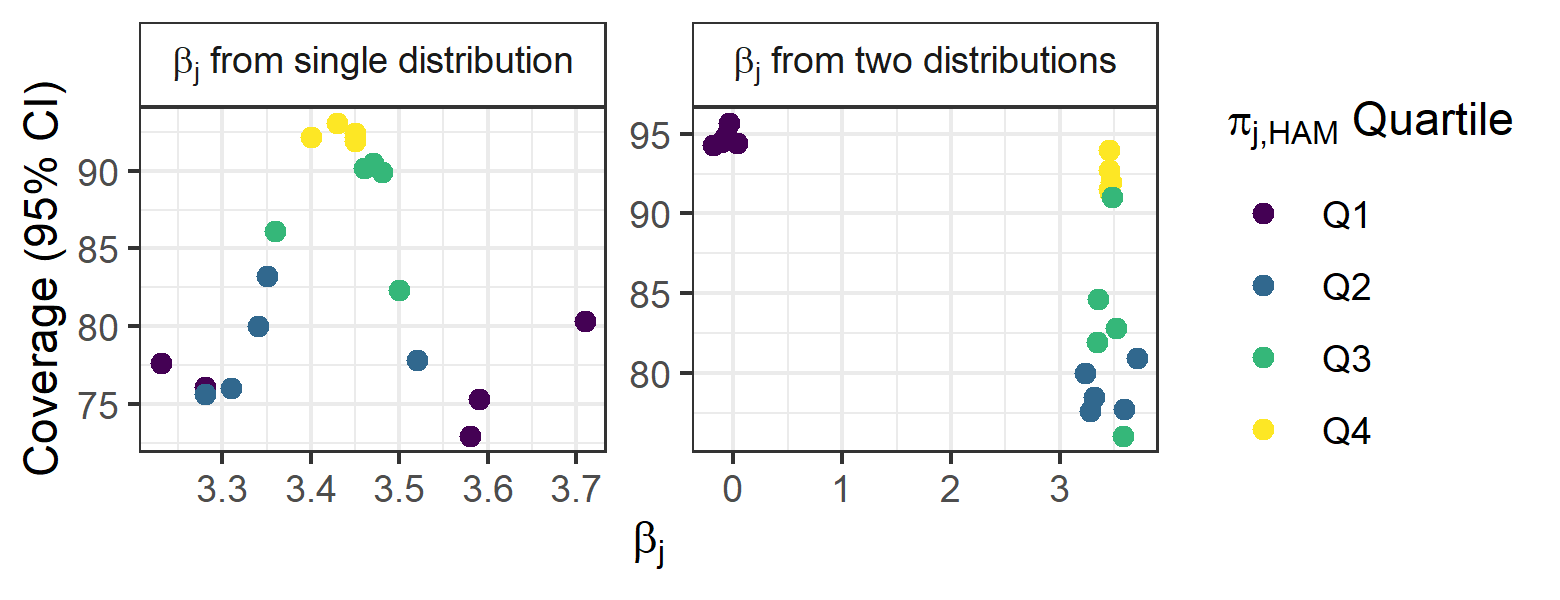}
    \caption{Coverage of $\widehat{\beta}_{\has, j}$ versus $\beta_j$ for each group in setting 3. }
    \label{fig:sim3_beta1}
\end{figure}

\subsection{Setting 4: Varying covariate data generation process}\label{subsec:sim4}
We investigate the performance of $\hbbeta_{\has}$ in the presence of population heterogeneity by varying the covariate data generation process, and illustrate the importance of standardizing covariates when high-information directions of $\bX_j^\top \bX_j$ dominate the divergence. We consider $k=3$ groups with $p = 4$ covariates each, and generate $\bbeta_j \sim \text{Unif}(-0.25, 0.25)$, rounding to two decimal places and holding these values constant over the Monte Carlo replicates. The sample sizes are $n_j\equiv n \in\{20, 50, 100, 200, 500\}$ and the error terms are $\bepsilon_j\sim \mathcal{N}(\bzero, \sigma^2_j\bI_n)$. In each group, $\bX_j$ has a column of ones for the intercept with the remaining covariates independently generated from $\mathcal{N}(\mu_j, b^2\bI_{p-1})$, with $b=1$ when $\mu_j=0$ and $b=\mu_j$ otherwise. We consider four scenarios for $\bmu=(\mu_1, \mu_2, \mu_3)$ and $\bsigma^2=(\sigma_1^2, \sigma_2^2, \sigma_3^2)$: (i) $\bmu=(0, 0, 0)$, $\bsigma^2=(1, 1, 1)$; (ii) $\bmu=(10, 10, 10)$, $\bsigma^2=(100, 100, 100)$; (iii) $\bmu=(0, 1, 2)$, $\bsigma^2=(1, 1, 4)$; (iv) $\bmu=(0, 5, 10)$, $\bsigma^2=(1, 25, 100)$. Across every setting, the variation of the error term is equal to the variation in the non-intercept covariates to allow us to isolate the impact of increasing the scale of the non-intercept covariates. 

To illustrate the consequence of using covariates on their raw scale, Tables \ref{tab:sim4_results1} and \ref{tab:sim4_results1_ci} display the eMSE and coverage rate of 95\% asymptotic confidence intervals, respectively, without covariate standardization. Scenarios (ii) and (iv), each having at least one group with $\mu_j = 10$, have lower coverage than the nominal rate even as the sample size grows large. These scenarios also show a greater reduction in eMSE relative to the MLE at large sample sizes, compared to scenarios (i) and (iii). To illustrate why this occurs, Figure \ref{fig:sim4_mleCI} plots the 95\% confidence intervals of $\tbbeta$ at $n = 500$. The overlap in confidence intervals for the intercept coefficient $\beta_1$ in scenarios (ii) and (iv) is leading our estimator to over-borrow, yielding large reductions in eMSE at the expense of coverage. 

\begin{table}[ht]
\centering \resizebox{\linewidth}{!}{
\begin{tabular}{r|rrrr|rrrr|rrrr|rrrr}
  & \multicolumn{4}{c}{scenario (i)}  & 
   \multicolumn{4}{|c}{scenario (ii)}  & 
 \multicolumn{4}{|c}{scenario (iii)}& \multicolumn{4}{|c}{scenario (iv)} \\
  & \multicolumn{4}{c|}{eMSE$\times 100$} & \multicolumn{4}{c|}{eMSE$\times 100$} & \multicolumn{4}{c|}{eMSE$\times 100$} & \multicolumn{4}{c}{eMSE$\times 100$}  \\
n & $\tbbeta$ & $\hbbeta_{\ridge}$ & $\hbbeta_{\ure}$ & $\hbbeta_{\has}$ 
  & $\tbbeta$ &  $\hbbeta_{\ridge}$  & $\hbbeta_{\ure}$ & $\hbbeta_{\has}$  
  & $\tbbeta$ &  $\hbbeta_{\ridge}$  & $\hbbeta_{\ure}$ & $\hbbeta_{\has}$ 
  & $\tbbeta$ &  $\hbbeta_{\ridge}$  & $\hbbeta_{\ure}$ &$\hbbeta_{\has}$  \\ 
  \hline
20 & 77.5 & 40.2 & 58.0 & 44.8 & 7756.1 & 3569.1 & 5021.0 & 4320.0 & 198.7 & 92.0 & 158.0 & 100.0 & 3262.6 & 997.3 & 2853.0 & 1427.1 \\ 
  50 & 25.9 & 17.7 & 23.4 & 17.7 & 2771.6 & 1383.8 & 1830.6 & 1631.6 & 66.3 & 34.7 & 54.1 & 37.9 & 1120.9 & 355.8 & 959.8 & 477.6 \\ 
  100 & 12.6 & 10.4 & 12.4 & 10.0 & 1240.9 & 593.9 & 842.7 & 719.4 & 31.5 & 17.9 & 26.9 & 20.9 & 539.3 & 172.3 & 464.8 & 233.3 \\ 
  200 & 6.2 & 5.6 & 6.4 & 5.5 & 624.5 & 309.2 & 464.1 & 375.5 & 15.9 & 9.8 & 14.3 & 12.4 & 255.5 & 83.5 & 224.9 & 114.0 \\ 
  500 & 2.4 & 2.3 & 2.5 & 2.3 & 246.3 & 120.7 & 202.9 & 151.5 & 6.0 & 4.3 & 5.7 & 5.8 & 104.2 & 36.7 & 96.4 & 52.5 \\ 
\end{tabular}}
\caption{Empirical MSEs, multiplied by 100, in simulation setting 4.} 
\label{tab:sim4_results1}
\end{table}

\begin{table}[ht]
\centering
\begin{tabular}{r|rrrr}
n & scenario (i) & scenario (ii) & scenario (iii) & scenario (iv) \\ 
  \hline
20 & 91.1 & 91.4 & 90.4 & 90.3 \\ 
  50 & 91.7 & 90.1 & 90.0 & 87.6 \\ 
  100 & 91.1 & 87.3 & 88.4 & 83.2 \\ 
  200 & 91.2 & 81.5 & 87.9 & 76.9 \\ 
  500 & 92.5 & 69.6 & 88.8 & 69.1 \\ 
   
\end{tabular}
\caption{Coverage rate (\%) for 95\% asymptotic confidence intervals using $\hbbeta(\bpi_{\ham})$ in simulation setting 4.} 
\label{tab:sim4_results1_ci}
\end{table}

\begin{figure}[ht!]
    \centering
    \includegraphics[width=0.7\linewidth]{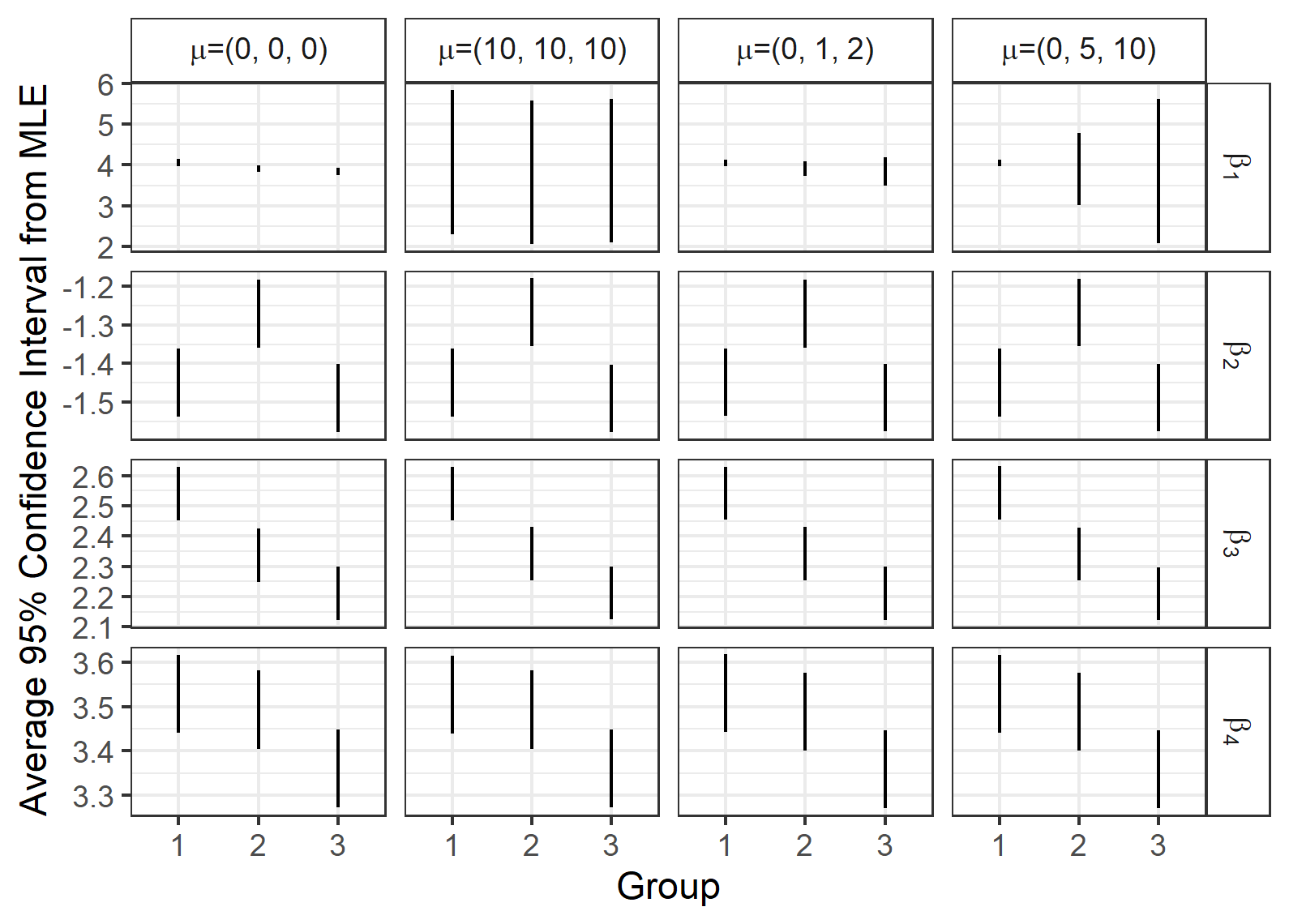}
    \caption{Average 95\% confidence intervals at $n=500$ for $\tbbeta$. In scenarios (ii) and (iv), the confidence intervals for the intercept, $\beta_1$, overlap substantially.}
    \label{fig:sim4_mleCI}
\end{figure}

Table \ref{tab:sim4_results2} shows the results after rescaling each non-intercept covariate to have a standard deviation of $1$ within each group, which prevents the intercept covariate from becoming overly informative. Across all scenarios, the eMSE improvement on the standardized data is smaller than on the original data and the coverage approaches the nominal level faster as $n$ grows.

\begin{table}[ht]
\centering
\resizebox{\linewidth}{!}{
\begin{tabular}{r|rrr|rrr|rrr|rrr}
  & \multicolumn{3}{|c}{scenario (i)}  & 
   \multicolumn{3}{|c}{scenario (ii)}  & 
 \multicolumn{3}{|c}{scenario (iii)}& \multicolumn{3}{|c}{scenario (iv)} \\
  & \multicolumn{2}{c}{eMSE$\times 100$} & CR & \multicolumn{2}{c}{eMSE$\times 100$} & CR & \multicolumn{2}{c}{eMSE$\times 100$} & CR & \multicolumn{2}{c}{eMSE$\times 100$} & CR \\
n & $\tbbeta$ & $\hbbeta_{\has}$ & $\hbbeta_{\has}$ 
  & $\tbbeta$ & $\hbbeta_{\has}$ & $\hbbeta_{\has}$ 
  & $\tbbeta$ & $\hbbeta_{\has}$ & $\hbbeta_{\has}$ 
  & $\tbbeta$ & $\hbbeta_{\has}$ & $\hbbeta_{\has}$ \\ 
  \hline
20 & 71.4 & 58.7 & 89.3 & 13073.0 & 9842.3 & 88.5 & 246.0 & 227.3 & 88.9 & 5447.7 & 5269.2 & 90.2 \\ 
  50 & 25.1 & 21.2 & 91.0 & 4697.8 & 3739.3 & 89.1 & 84.4 & 80.4 & 92.2 & 1919.6 & 1889.2 & 92.8 \\ 
  100 & 12.4 & 10.9 & 91.5 & 2150.5 & 1768.2 & 89.8 & 40.3 & 38.8 & 92.6 & 919.4 & 919.8 & 94.3 \\ 
  200 & 6.1 & 5.7 & 91.5 & 1074.3 & 914.9 & 90.1 & 20.5 & 19.8 & 92.8 & 440.0 & 441.9 & 94.5 \\ 
  500 & 2.4 & 2.3 & 92.8 & 426.2 & 390.8 & 91.0 & 7.8 & 7.7 & 92.8 & 178.2 & 178.7 & 94.9 \\ 
\end{tabular}}
\caption{Empirical MSEs, multiplied by 100, and coverage rate (CR, \%) of 95\% confidence intervals in simulation setting 4 on rescaled data.}
\label{tab:sim4_results2}
\end{table}

\section{Health economics application} \label{sec:data}

We illustrate the practical advantages of the HAS estimator in an analysis of covariates of intensive-care unit (ICU) length of stay using real data from the eICU Collaborative Research Database \citep{eicu2019}. Section \ref{appsec:eicu_access} in the supplement describes the access procedure for these data. The determinants of hospital length of stay are of continuing interest in health economics, and time spent in the ICU represents a particularly resource-intensive period \citep{carey_hospital_2000, david_integration_2013, aragon_long-run_2022, siciliani_living_2025}. Anticipating and reducing the length of stay is valuable to hospitals for resource management, especially in the presence of ICU bed shortages \citep{rapoport_length_2003, bai_operations_2018}. Longer ICU stays are associated with higher healthcare costs and adverse health outcomes, including an increased risk of mortality after discharge \citep{williams_effect_2010, moitra_relationship_2016}. We investigate the association between ICU length of stay and a set of covariates measured at admission to the ICU across multiple hospitals. The HAS estimator provides inference for each hospital while also allowing information borrowing to improve precision.

We conduct this analysis using the log transform of the ICU length of stay as the outcome. We assess the association between length of stay and a patient's age, gender, systolic and diastolic blood pressure when entering the ICU, whether a patient was admitted from the emergency department (ED) or from another source, and the Acute Physiology and Chronic Health Evaluation (APACHE) IV score. The APACHE IV score incorporates clinical and health history data collected from patients within the first day in the ICU, but it is designed with ICU mortality as the primary outcome \citep{zimmerman_acute_2006}. Prior work indicates that the APACHE IV score is a poor predictor of ICU length of stay; we include it in our model to help account for initial health status and to assess whether the association between the score and length of stay varies across hospitals \citep{verburg_comparison_2014, zangmo_validating_2023}. We consider only patients discharged alive from the ICU. Our data consist of $k=29$ hospitals from the West region of the eICU database for which there were at least 100 patients recorded as discharged alive from the ICU in 2014 and 2015. We standardize the continuous predictors as discussed in Sections \ref{subsec:est_all} and \ref{subsec:sim4}.

To implement our method, we find the vector $\bpi_{\has}$ that minimizes the surrogate MSE in equation \eqref{eqn:pseudo_mse}. As visualized in Figure \ref{fig:da_beta_apache}, the majority of hospitals have $\pi_{j,\has}$ between 0.3 and 0.5, while four hospitals have a smaller $\pi_{j,\has}$. The variability in $\bpi_{\has}$ echoes the heterogeneity in the hospital MLEs; we would expect to see $\pi_{j,\has}$ close to 0.5 for all groups in the presence of homogeneity. Using $\bpi_{\has}$, we first compare the estimated centroid $\hbtheta(\bpi_{\has})$ to the pooled estimate from fully combining the data, $\tbbeta_{\fe}$, and the pooled estimate from \cite{kwon_optimal_2026}'s URE-minimizing method, $\hbtheta_{\ure}$. Table \ref{tab:combined_ci} displays the 95\% confidence intervals for the pooled estimates; the confidence interval using $\hbtheta(\bpi_{\has})$ is constructed using the distribution of the MLE $\tbbeta$. While the pooled confidence interval gives inference on a shared parameter $\bbeta_m$, the confidence interval for $\hbtheta(\bpi_{\has})$ is built to capture the true value of the centroid $\btheta$, which is not necessarily equal to $\bbeta_m$. Nevertheless, in this case, the confidence intervals cover similar regions. The URE point estimate frequently lies outside of all of the other confidence intervals.

Figure \ref{fig:da_beta_apache} plots the HAS and MLE confidence intervals for $\beta_{\text{APACHE IV}}$ with the corresponding $\pi_{j, \has}$, and point estimates $\hbtheta_{\ure}$; empirical Bayes confidence intervals are excluded because we have too few groups to estimate the necessary covariance parameters without additional assumptions. Estimates for hospitals with high $\pi_{j,\has}$ experience a large shift in their point estimate toward the centroid estimate, as well as a reduction in the width of the corresponding confidence interval. \cite{kwon_optimal_2026}'s estimates are much closer to the MLE than the HAS estimates. Figure \ref{fig:coeffs95} displays the 95\% confidence intervals for all covariates, with hospitals sorted from smallest to largest $\pi_{j,\has}$. Changes in significance between the MLE and HAS estimators are more frequent as $\pi_{j,\has}$ increases, reflecting the decrease in variance with more information sharing. Table \ref{tab:ht_results} reports the result of hypothesis tests by covariate and, as expected, in all but one case the HAS estimates are more often significant at level $\alpha = 0.05$. In particular, the ``age'' and ``admission from an ED'' effect estimates gain statistical significance using HAS for several hospitals. These results highlight the efficiency gain from borrowing information across groups, and the importance of accounting for finite-sample bias in minimizers of UREs.

Overall, an increasing APACHE IV score, indicative of worse clinical condition at ICU entry, is associated with a longer ICU length of stay for all hospitals. In fact, out of the selected covariates, the APACHE IV score is the only one which is significantly related to the log-transformed length of stay for all hospitals using the HAS estimator. All else held constant, patients admitted from the ED tend to have shorter stays in the ICU than other patients, for most hospitals. Counterintuitively, for most hospitals, older patients also have shorter stays in the ICU after accounting for other covariates. This result is perhaps a by-product of keeping only patients who are discharged alive; filtering disproportionately removes older individuals with longer ICU stays.

\begin{table}[ht]
\centering
\begin{tabular}{r|rr|r|rr} 
 & \multicolumn{2}{c|}{$\tbbeta_{\fe}$} &  
 \multicolumn{1}{c|}{$\hbtheta_{\ure}$} &\multicolumn{2}{c}{$\hbtheta(\bpi_{\has})$}\\ 
Covariate & Lower & Upper &  Point Est. & Lower & Upper  \\ 
  \hline
Intercept & 0.475 & 0.527 & 1.092 & 0.498 & 0.553 \\ 
Gender & -0.016 & 0.040 & 0.365 & -0.003 & 0.056 \\ 
Admission from ED & -0.083 & -0.027 & 0.098 & -0.126 & -0.067 \\ 
 APACHE IV score & 0.344 & 0.374 & 0.077 & 0.348 & 0.379 \\ 
 Age & -0.091 & -0.060 & -0.116 & -0.089 & -0.057 \\ 
 Diastolic BP & 0.020 & 0.057 & 0.198 & 0.024 & 0.064 \\ 
 Systolic BP & -0.055 & -0.017 & 0.049 & -0.057 & -0.018 \\ 
\end{tabular}
\caption{Lower and upper bounds of 95\% confidence intervals for pooled estimators, including the pooled URE point estimate.} 
\label{tab:combined_ci}
\end{table}

\begin{figure}[ht!]
    \centering
    \includegraphics[width=\linewidth]{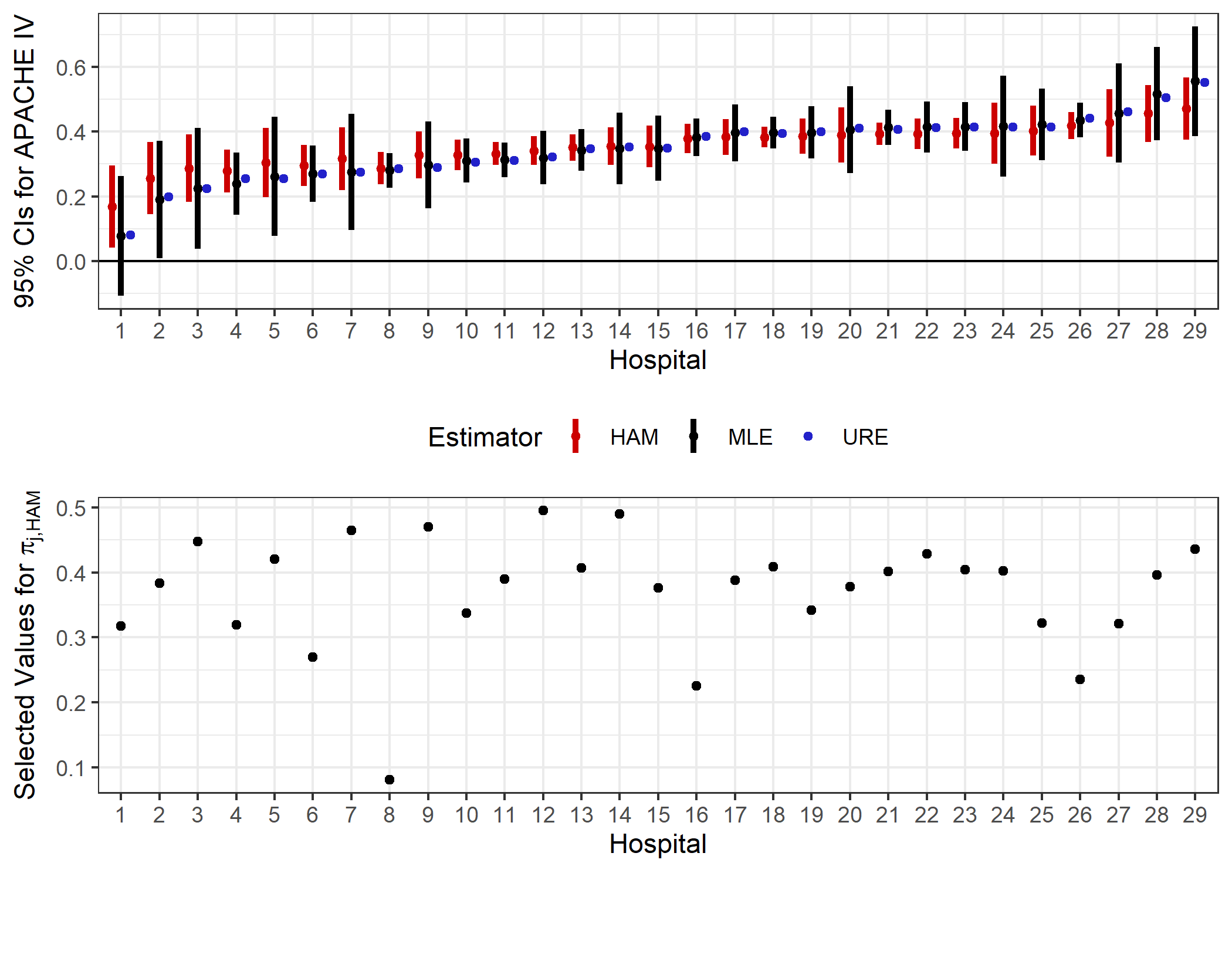}
    \caption{Hospital-level 95\% confidence intervals for $\beta_{\textrm{APACHE IV}}$ using $\hbbeta(\bpi_{\has})$ and $\tbbeta$, and corresponding $\pi_{j, \has}$. Point estimates for the shrinkage estimator proposed by \cite{kwon_optimal_2026}, $\hbbeta_\ure$, are shown. Hospitals are ordered from smallest to largest MLE.}
    \label{fig:da_beta_apache}
\end{figure}

\begin{figure}[ht!]
    \centering
    \includegraphics[width=0.9\linewidth]{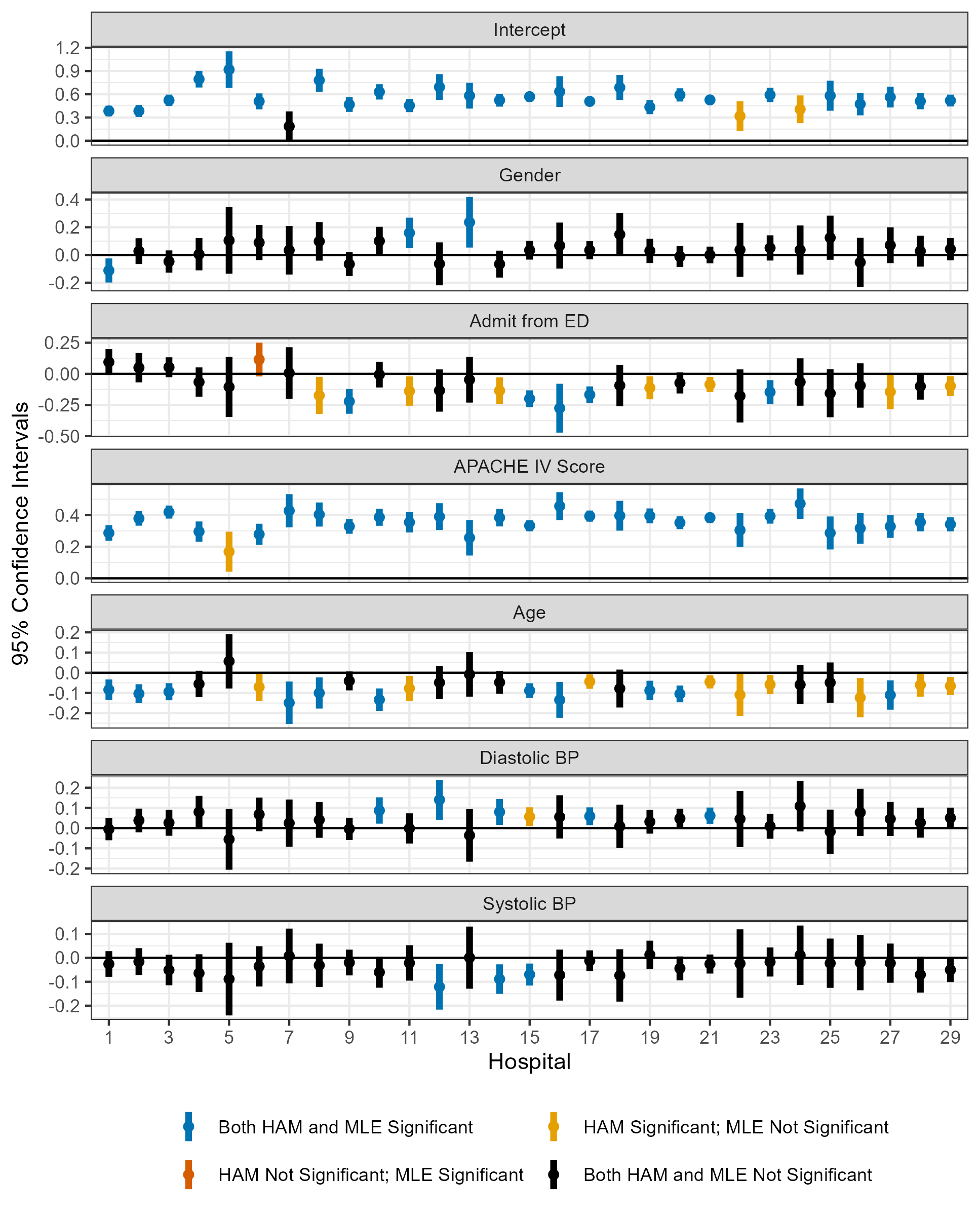}
    \caption{95\% confidence intervals for each coefficient using the HAS estimator with significance at level $0.05$ using both HAS and MLE estimators. Hospitals are ordered from smallest  to largest $\pi_{j,\has}$.}
    \label{fig:coeffs95}
\end{figure}

\begin{table}[ht]
\centering
\begin{tabular}{l|rr|rr} 
 & \multicolumn{2}{c}{$p_{\has} \leq 0.05$}& \multicolumn{2}{c}{$p_{\has} > 0.05$}\\
Covariate & $p_{\mle} \leq 0.05$ & $p_{\mle} > 0.05$ & $p_{\mle} \leq 0.05$ & $p_{\mle} > 0.05$ \\ 
  \hline
Intercept            &   1 &   0 &   2 &  26 \\ 
Gender               &  26 &   0 &   0 &   3 \\ 
Admission from ED    &  16 &   1 &   7 &   5 \\ 
APACHE IV Score      &   0 &   0 &   1 &  28 \\ 
Age                  &   9 &   0 &   9 &  11 \\ 
Diastolic BP         &  23 &   0 &   1 &   5 \\ 
Systolic BP          &  26 &   0 &   0 &   3 \\ 
\end{tabular}
   \caption{Hypothesis test results across $k = 29$ hospitals for each coefficient versus a null hypothesis of $0$. The $p_{\has}$ and $p_{\mle}$ denote $p$-values for each covariate effect using the HAS and MLE, respectively.}
    \label{tab:ht_results}
\end{table}

\section{Conclusion}

Several unique aspects of our proposal bear emphasizing. First, we introduce a new centroid as the data-sharing mechanism that links estimates from individual groups. This centroid's estimation introduces important flexibility for data-adaptive borrowing. Group-specific shrinkage parameters further allow our estimator to differentially borrow according to a group's similarity to the centroid. Second, our use of a Kullback-Leibler divergence shrinkage penalty leads to intuitive and interpretable estimators that share a form with the classical empirical Bayesian estimator, without requiring strong parametric assumptions. Third, by carefully choosing the amount of shrinkage, our estimator is both more efficient than the MLE in finite samples \emph{and} has large-sample inferential guarantees critical for practical use. The proof of these properties relies on an important new parameterization that decouples the selection of the centroid from the selection of the amount of shrinkage. Finally, rather than minimizing SURE, we develop a bias-reduced estimator of the shrinkage parameters that provides improved finite-sample performance. Put together, these results bestow on our estimator data-adaptive flexibility with statistical efficiency guarantees that preserve practical advantages. 

Future work should focus on extending the HAS estimator to generalized linear models to enable its implementation in a wider variety of models. Future work may also consider different inferential goals and investigate the theoretical and empirical properties of the HAS estimator when selecting shrinkage parameters using a different loss function than the MSE. For example, one could instead select $\bpi$ to minimize the mean squared prediction error or the MSE for a subset of group effects. We leave these open questions to the interested reader.

\section*{Acknowledgments}

This work was supported by a grant from the National Science Foundation (DMS~2337943).

\bibliographystyle{apalike}
\bibliography{references3}

\begin{thebibliography}{}

\bibitem[Amari, 2020]{amari_information_2020}
Amari, S.-i. (2020).
\newblock {\em Information geometry and its applications}.
\newblock Springer Japan.

\bibitem[Aragón et~al., 2022]{aragon_long-run_2022}
Aragón, M.~J., Chalkley, M., and Kreif, N. (2022).
\newblock The long-run effects of diagnosis related group payment on hospital lengths of stay in a publicly funded health care system: {Evidence} from 15 years of micro data.
\newblock {\em Health Economics}, 31(6):956--972.

\bibitem[Armstrong et~al., 2022]{armstrong_robust_2022}
Armstrong, T.~B., Kolesár, M., and Plagborg-Møller, M. (2022).
\newblock Robust empirical {Bayes} confidence intervals.
\newblock {\em Econometrica}, 90(6):2567--2602.

\bibitem[Bai et~al., 2018]{bai_operations_2018}
Bai, J., Fügener, A., Schoenfelder, J., and Brunner, J.~O. (2018).
\newblock Operations research in intensive care unit management: a literature review.
\newblock {\em Health Care Management Science}, 21(1):1--24.

\bibitem[Baltagi et~al., 2000]{baltagi_pool_2000}
Baltagi, B.~H., Griffin, J.~M., and Xiong, W. (2000).
\newblock To pool or not to pool: {Homogeneous} versus heterogeneous estimators applied to cigarette demand.
\newblock {\em The Review of Economics and Statistics}, 82(1):117--126.

\bibitem[Carey, 2000]{carey_hospital_2000}
Carey, K. (2000).
\newblock Hospital cost containment and length of stay: {An} econometric analysis.
\newblock {\em Southern Economic Journal}, 67(2):363--380.

\bibitem[David et~al., 2013]{david_integration_2013}
David, G., Rawley, E., and Polsky, D. (2013).
\newblock Integration and task allocation: {Evidence} from patient care.
\newblock {\em Journal of Economics \& Management Strategy}, 22(3):617--639.

\bibitem[Firth, 1993]{firth_bias_1993}
Firth, D. (1993).
\newblock Bias reduction of maximum likelihood estimates.
\newblock {\em Biometrika}, 80(1):27--38.

\bibitem[Frisch and Waugh, 1933]{frisch_partial_1933}
Frisch, R. and Waugh, F.~V. (1933).
\newblock Partial time regressions as compared with individual trends.
\newblock {\em Econometrica}, 1(4):387--401.

\bibitem[Glass, 1976]{glass1976}
Glass, G.~V. (1976).
\newblock Primary, secondary, and meta-analysis of research.
\newblock {\em Educational Researcher}, 5(10):3--8.

\bibitem[Hansen, 2016]{hansen_efficient_2016}
Hansen, B.~E. (2016).
\newblock Efficient shrinkage in parametric models.
\newblock {\em Journal of Econometrics}, 190(1):115--132.

\bibitem[Hector and Martin, 2024]{hector_turning_2024}
Hector, E.~C. and Martin, R. (2024).
\newblock Turning the information-sharing dial: efficient inference from different data sources.
\newblock {\em Electronic Journal of Statistics}, 18:2974--3020.

\bibitem[Horn and Johnson, 1985]{horn_chapter_1985}
Horn, R.~A. and Johnson, C.~R. (1985).
\newblock Chapter 7: {P}ositive definite matrices.
\newblock In {\em Matrix {Analysis}}. Cambridge University Press, Cambridge.

\bibitem[James and Stein, 1961]{james_estimation_1961}
James, W. and Stein, C. (1961).
\newblock Estimation with quadratic loss.
\newblock In {\em Proceedings of the {Fourth} {Berkeley} {Symposium} on {Mathematical} {Statistics} and {Probability}, {Volume} 1: {Contributions} to the {Theory} of {Statistics}}, volume 4.1, pages 361--380. University of California Press.

\bibitem[Kullback and Leibler, 1951]{kullback_information_1951}
Kullback, S. and Leibler, R.~A. (1951).
\newblock On information and sufficiency.
\newblock {\em The Annals of Mathematical Statistics}, 22(1):79--86.

\bibitem[Kwon, 2026]{kwon_optimal_2026}
Kwon, S. (2026).
\newblock Optimal shrinkage estimation of fixed effects in linear panel data models.
\newblock {\em Econometrica}, 94(2):663--677.

\bibitem[Lettink et~al., 2023]{lettink_two-dimensional_2023}
Lettink, A., Chinapaw, M., and van Wieringen, W.~N. (2023).
\newblock Two-dimensional fused targeted ridge regression for health indicator prediction from accelerometer data.
\newblock {\em Journal of the Royal Statistical Society, Series C}, 72:1064--1078.

\bibitem[Lovell, 1963]{lovell_seasonal_1963}
Lovell, M.~C. (1963).
\newblock Seasonal adjustment of economic time series and multiple regression analysis.
\newblock {\em Journal of the American Statistical Association}, 58(304):993--1010.

\bibitem[Maddala et~al., 1997]{maddala_estimation_1997}
Maddala, G.~S., Trost, R.~P., Li, H., and Joutz, F. (1997).
\newblock Estimation of short-run and long-run elasticities of energy demand from panel data using shrinkage estimators.
\newblock {\em Journal of Business \& Economic Statistics}, 15(1):90--100.

\bibitem[Mehrabani and Ullah, 2020]{mehrabani_improved_2020}
Mehrabani, A. and Ullah, A. (2020).
\newblock Improved average estimation in seemingly unrelated regressions.
\newblock {\em Econometrics}, 8(2).

\bibitem[Moitra et~al., 2016]{moitra_relationship_2016}
Moitra, V.~K., Guerra, C., Linde-Zwirble, W.~T., and Wunsch, H. (2016).
\newblock Relationship between {ICU} length of stay and long-term mortality for elderly {ICU} survivors.
\newblock {\em Critical care medicine}, 44(4):655--662.

\bibitem[Morris, 1983]{morris_parametric_1983}
Morris, C.~N. (1983).
\newblock Parametric empirical {Bayes} inference: {Theory} and applications.
\newblock {\em Journal of the American Statistical Association}, 78(381):47--55.

\bibitem[Pollard et~al., 2019]{eicu2019}
Pollard, T., Johnson, A., Raffa, J., Celi, L.~A., Badawi, O., and Mark, R. (2019).
\newblock {eICU} {Collaborative} {Research} {Database} (version 2.0).
\newblock PhysioNet.

\bibitem[Rapoport et~al., 2003]{rapoport_length_2003}
Rapoport, J., Teres, D., Zhao, Y., and Lemeshow, S. (2003).
\newblock Length of stay data as a guide to hospital economic performance for {ICU} patients.
\newblock {\em Medical Care}, 41(3):386--397.

\bibitem[Siciliani et~al., 2025]{siciliani_living_2025}
Siciliani, L., Wen, J., and Gaughan, J. (2025).
\newblock Living alone and provider behaviour in public and private hospitals.
\newblock {\em Journal of Health Economics}, 102:103016.

\bibitem[Stein, 1981]{stein_estimation_1981}
Stein, C.~M. (1981).
\newblock Estimation of the mean of a multivariate normal distribution.
\newblock {\em The Annals of Statistics}, 9(6).

\bibitem[Tibshirani and Rosset, 2019]{tibshirani_excess_2019}
Tibshirani, R.~J. and Rosset, S. (2019).
\newblock Excess optimism: {How} biased is the apparent error of an estimator tuned by {SURE}?
\newblock {\em Journal of the American Statistical Association}, 114(526):697--712.

\bibitem[Verburg et~al., 2014]{verburg_comparison_2014}
Verburg, I. W.~M., Keizer, N. F.~d., Jonge, E.~d., and Peek, N. (2014).
\newblock Comparison of regression methods for modeling intensive care length of stay.
\newblock {\em PLOS ONE}, 9(10):e109684.

\bibitem[Walters, 2024]{walters_chapter_2024}
Walters, C. (2024).
\newblock Chapter 3 - {Empirical} {Bayes} methods in labor economics.
\newblock In Dustmann, C. and Lemieux, T., editors, {\em Handbook of {Labor} {Economics}}, volume~5, pages 183--260. Elsevier.

\bibitem[Williams et~al., 2010]{williams_effect_2010}
Williams, T.~A., Ho, K.~M., Dobb, G.~J., Finn, J.~C., Knuiman, M., Webb, S. A.~R., and {on behalf of the Royal Perth Hospital ICU Data Linkage Group} (2010).
\newblock Effect of length of stay in intensive care unit on hospital and long-term mortality of critically ill adult patients.
\newblock {\em BJA: British Journal of Anaesthesia}, 104(4):459--464.

\bibitem[Xie et~al., 2012]{xie_sure_2012}
Xie, X., Kou, S.~C., and Brown, L.~D. (2012).
\newblock {SURE} estimates for a heteroscedastic hierarchical model.
\newblock {\em Journal of the American Statistical Association}, 107(500):1465--1479.

\bibitem[Zangmo and Khwannimit, 2023]{zangmo_validating_2023}
Zangmo, K. and Khwannimit, B. (2023).
\newblock Validating the {APACHE} {IV} score in predicting length of stay in the intensive care unit among patients with sepsis.
\newblock {\em Scientific Reports}, 13:5899.

\bibitem[Zhan, 2002]{zhan_inequalities_2002}
Zhan, X. (2002).
\newblock Inequalities in the {L}\"{o}wner {P}artial {O}rder.
\newblock In {\em Matrix {I}nequalities}, volume 1790. Springer Berlin Heidelberg, Berlin, Heidelberg.

\bibitem[Zimmerman et~al., 2006]{zimmerman_acute_2006}
Zimmerman, J.~E., Kramer, A.~A., McNair, D.~S., and Malila, F.~M. (2006).
\newblock Acute {Physiology} and {Chronic} {Health} {Evaluation} ({APACHE}) {IV}: {Hospital} mortality assessment for today's critically ill patients.
\newblock {\em Critical Care Medicine}, 34(5):1297--1310.

\end{thebibliography}

\appendix
\newpage
\section{Derivation of estimators} \label{appsec:derivation}

In this section, we derive the estimators in equations \eqref{eqn:centroid} and \eqref{eqn:hbbeta_j} of the main manuscript, as well as those in Section \ref{subsec:est_subset}. We derive these estimators for the general case where group $j$ has $q_j$ covariates, $p$ of which are common across groups.

 In this context, for group $j$, denote the $n_j \times p$ matrix for the effects of interest as $\bX_j$ and denote the $n_j \times (q_j - p)$ matrix of control covariates as $\bZ_j$. Outcomes for group $j$ are modeled as 
 \begin{align}\label{eqn:partitioned_model2}
     \bY_j = \bX_j\bbeta_j + \bZ_j\bgamma_j + \epsilon_j,\quad \epsilon_j \sim N(\bzero,\sigma_j^2\bI_{n_j}).
 \end{align} 

 Let $\bK=\bone_k \otimes \bI_p \in \mathbb{R}^{pk\times p}$ matrix which stacks the vector to its right, and define $\bPi =\text{block-diag} \{ \pi_1\bI_p, \ldots, \pi_k\bI_p \} \in \mathbb{R}^{pk\times pk}$. As in Section \ref{subsec:est_subset}, define $\bM_j = \bI_{n_j} - \bZ_j(\bZ_j^{\top}\bZ_j)^{-1}\bZ_j^{\top}$ when $q_j > p$, and $\bM_j=\bI_{n_j}$ when $q_j=p$. Let $\bY=(\bY_1, \ldots, \bY_k)$, $\bX=\text{block-diag}\{\bX_1, \ldots, \bX_k\}$, $\bM = \text{block-diag}\{\bM_1,\ldots, \bM_k\}$, and $\bSigma = \text{block-diag}\{\sigma_1^2\bI_{n_1},\ldots, \sigma^2_k\bI_{n_k}\}$. 

First, note that we can change the multiplication order for some sets of matrices:
\begin{align*}
\bPi\scaledXTMX &= \scaledXTMX\bPi\\
(\bI_{pk}-\bPi)\scaledXTMX &= \scaledXTMX(\bI_{pk}-\bPi).
\end{align*}
Next, as the objective function in equation \eqref{eqn:obj_func2} contains the generalized inverse for $\bM_j$, we can write
\begin{align*}
    (\bM_j\bY_j - \bM_j\bX_j\bbeta_j)^\top \bM_j^{+} &(\bM_j\bY_j - \bM_j\bX_j\bbeta_j) \\
    &= (\bY_j - \bX_j\bbeta_j)^{\top}\bM_j\bM_j^{\top}\bM_j(\bY_j - \bX_j\bbeta_j)\\
    &= (\bY_j - \bX_j\bbeta_j)^{\top}\bM_j(\bY_j - \bX_j\bbeta_j)\\
    &= (\bM_j\bY_j - \bM_j\bX_j\bbeta_j)^{\top}(\bM_j\bY_j - \bM_j\bX_j\bbeta_j).
\end{align*}
Using this, we maximize the objective function 
\begin{align*}
        O(\bbeta, \btheta; \bPi) =& -\frac{1}{2}(\bM\bY - \bM\bX\bbeta)^\prime\bSigma^{-1}(\bM\bY - \bM\bX\bbeta)\\ 
        &\quad - \frac{1}{2}(\bbeta - \bK\btheta)^\prime\bPi(\bI - \bPi)^{-1}\scaledXTMX(\bbeta - \bK\btheta), 
\end{align*}
with respect to $\btheta$ and $\bbeta$. This is the objective function in equation \eqref{eqn:obj_func} from Section \ref{sec:estimator} in matrix form. We begin by taking the partial derivative of the objective function with respect to $\btheta$:
\begin{align*}
    \nabla_{\btheta} O(\bbeta, \btheta; \bPi) &=
    -\frac{1}{2}\{-2\bK^\prime\bPi(\bI_{pk} - \bPi)^{-1}\scaledXTMX(\bbeta 
        - \bK\btheta)\}.
\end{align*}
The root of this gradient is given by
\begin{align} \label{eqn:theta-complicated}
 \hbtheta(\bpi) &= 
    \{\bK^\prime\bPi(\bI_{pk} - \bPi)^{-1}\scaledXTMX\bK\}^{-1}
    \bK^\prime\bPi(\bI_{pk} - \bPi)^{-1}\scaledXTMX\hbbeta(\bpi) \nonumber \\
 &= 
    \biggl\{\sum_{j=1}^d 
        \frac{\pi_j}  {(1-\pi_j)\sigma_j^2}\bX_{j}^\prime\bM_j\bX_j
    \biggl\}^{-1}
    \sum_{j=1}^d \frac{\pi_j}{(1-\pi_j)\sigma_j^2}\bX_j^\prime\bM_j\bX_j\hbbeta_j(\bpi).
\end{align}
We will show that this can be simplified after obtaining a closed form for $\hbbeta(\bpi)$, discussed next.

The next step is to derive the $\hbbeta(\bpi)$ that maximizes $O(\bbeta, \btheta; \blambda)$.  In the derivation below, we suppress the dependence of $\hbbeta(\bpi)$ and $\hbtheta(\bpi)$ on $\bpi$ for brevity, and write $\bI$ for $\bI_{pk}$.  The gradient of the objective function with respect to $\bbeta$ is
\begin{align*}
    \nabla_{\bbeta} O(\bbeta, \btheta; \bpi) &=
    -\frac{1}{2}\{-2\bX^\prime\bM\bSigma^{-1}(\bM\bY - \bM\bX\bbeta)\}\\&\quad 
    -\frac{1}{2}\{2\bPi(\bI - \bPi)^{-1}\scaledXTMX(\bbeta - \bK\btheta)\}.
\end{align*}
We again find the root of the gradient, denoted by  $\hbbeta$:
\begin{align*}
    \bzero &= \scaledXTMY - \scaledXTMX\hbbeta -\bPi(\bI - \bPi)^{-1}\scaledXTMX\hbbeta \\
    &\quad + \bPi(\bI - \bPi)^{-1}\scaledXTMX\bK\hbtheta .\end{align*}
Rearranging, we obtain
    \begin{align*}
    \scaledXTMY &= (\bI - \bPi)^{-1}\scaledXTMX\hbbeta   - \bPi(\bI - \bPi)^{-1}\scaledXTMX\bK\hbtheta, \end{align*}
or, equivalently,
    \begin{align*}
    \scaledXTMY &= (\bI-\bPi)^{-1}\scaledXTMX(\hbbeta - \bPi\bK\hbtheta).
\end{align*}
Plugging in the form of $\hbtheta$, the above can be rewritten as
\begin{align*}
    &\scaledXTMY \\
    &= (\bI-\bPi)^{-1}\scaledXTMX\\
    &\quad\cdot[\bI - \bPi\bK\{\bK^\prime\bPi(\bI - \bPi)^{-1}\scaledXTMX\bK\}^{-1}
        \bK^\prime\bPi(\bI - \bPi)^{-1}\scaledXTMX]\hbbeta.\end{align*}
Rearranging, the estimator of $\bbeta$ is
\begin{align} \label{eqn:beta-complicated}
    \hbbeta(\bpi) =& [\bI - \bPi\bK\{\bK^\prime\bPi(\bI - \bPi)^{-1}\scaledXTMX\bK\}^{-1}
        \bK^\prime\bPi(\bI - \bPi)^{-1}\scaledXTMX]^{-1} \nonumber \\
        & \cdot(\bI-\bPi)(\scaledXTMX)^{-1}\scaledXTMY.
\end{align}

We now proceed to simplify this expression. First, the Frisch-Waugh-Lovell theorem states that for a partitioned model like the one in equation \eqref{eqn:partitioned_model}, the MLE for $\bbeta_j$ is $\tbbeta_j = (\bX_j'\bM_j\bX_j)^{-1}\bX_j'\bM_j\bY_j$ \citep{frisch_partial_1933, lovell_seasonal_1963}. In our block-diagonalized notation, the MLE for $\bbeta$ is
\begin{align*}
\tbbeta  = (\scaledXTMX)^{-1}\scaledXTMY.
\end{align*}
Substituting this MLE into equation \eqref{eqn:beta-complicated} gives
\begin{align*}
    \hbbeta(\bpi) =& [\bI - \bPi\bK\{\bK^\prime\bPi(\bI - \bPi)^{-1}\scaledXTMX\bK\}^{-1}
        \bK^\prime\bPi(\bI - \bPi)^{-1}\scaledXTMX]^{-1}\\
        &\cdot(\bI-\bPi)\tbbeta.
\end{align*}
Next, we use the Woodbury identity to simplify the first inverse term of $\hbbeta(\bpi)$:
\begin{align*}
    &[\bI - \bPi\bK\{\bK^\prime\bPi(\bI - \bPi)^{-1}\scaledXTMX\bK\}^{-1}
        \bK^\prime\bPi(\bI - \bPi)^{-1}\scaledXTMX]^{-1} \\
    &= \bI + \bPi\bK  \{\bK^\prime\bPi(\bI - \bPi)^{-1}\scaledXTMX\bK - \bK^\prime\bPi^2(\bI - \bPi)^{-1}\scaledXTMX\bK\}^{-1}\\
    &\quad \cdot \bK^\prime\bPi(\bI - \bPi)^{-1}\scaledXTMX \\
    &= \bI + \bPi\bK\{\bK^\prime\bPi(\bI - \bPi)(\bI - \bPi)^{-1}\scaledXTMX\bK\}^{-1} \bK^\prime\bPi(\bI - \bPi)^{-1}\scaledXTMX \\
    &= \bI +\bPi\bK (\bK^\prime\bPi\scaledXTMX\bK)^{-1} \bK^\prime\bPi(\bI - \bPi)^{-1}\scaledXTMX .
\end{align*}
Plugging this into equation \eqref{eqn:beta-complicated} yields a simplified form of our estimator,
\begin{align}\label{eqn:beta-less-complicated}
    \hbbeta(\bpi) &= \{\bI +\bPi\bK (\bK^\prime\bPi\scaledXTMX\bK)^{-1} \bK^\prime\bPi(\bI - \bPi)^{-1}\scaledXTMX\} \nonumber \\
        &\quad \cdot (\bI-\bPi)\tbbeta \nonumber \\
    &= (\bI-\bPi)\tbbeta  + \bPi\bK(\bK^\prime\bPi\scaledXTMX\bK)^{-1} \bK^\prime\bPi\scaledXTMX\tbbeta.
\end{align}

Before simplifying this further, we turn our attention back to our estimator of the centroid $\btheta$. Plugging the form of $\hbbeta(\bpi)$ in equation \eqref{eqn:beta-less-complicated} into equation \eqref{eqn:theta-complicated}, we obtain
\begin{align*}
    \hbtheta(\bpi) &= \{\bK^\prime\bPi(\bI - \bPi)^{-1}\scaledXTMX\bK\}^{-1}
    \bK^\prime\bPi(\bI - \bPi)^{-1}\scaledXTMX\hbbeta(\bpi)\\
    &= \{\bK^\prime\bPi(\bI - \bPi)^{-1}\scaledXTMX\bK\}^{-1}
    \bK^\prime\bPi(\bI - \bPi)^{-1}\scaledXTMX\\
    &\quad \cdot \{(\bI-\bPi) + \bPi\bK(\bK^\prime\bPi\scaledXTMX\bK)^{-1} \bK^\prime\bPi\scaledXTMX\}\tbbeta.
\end{align*}
Denoting $\bPi^2 = \bPi ~ \bPi$, and using the commutative nature of $\bPi$, this can be simplified as
\begin{align*}
    \hbtheta(\bpi) &= \{\bK^\prime\bPi(\bI - \bPi)^{-1}\scaledXTMX\bK\}^{-1}\\
    &\quad \cdot
    \{\bK^\prime\bPi\scaledXTMX + 
    \bK^\prime\bPi^2(\bI - \bPi)^{-1}\scaledXTMX\bK\\
    &\quad \cdot(\bK^\prime\bPi\scaledXTMX\bK)^{-1} \bK^\prime\bPi\scaledXTMX\}\tbbeta \\
    &=  \{\bK^\prime\bPi(\bI - \bPi)^{-1}\scaledXTMX\bK\}^{-1}\\
    &\quad \cdot
    \{\bI + 
    \bK^\prime\bPi^2(\bI - \bPi)^{-1}\scaledXTMX\bK(\bK^\prime\bPi\scaledXTMX\bK)^{-1} \} \\
    &\quad \cdot\bK^\prime\bPi\scaledXTMX\tbbeta .
\end{align*}
We focus next on the term on the second-to-last line. In the first step, write \[
\bI = \bK^\prime\bPi\scaledXTMX\bK(\bK^\prime\bPi\scaledXTMX\bK)^{-1},
\] and then factor, pulling the inverse term out to the right.
\begin{align*}
    &\bI + 
    \bK^\prime\bPi^2(\bI - \bPi)^{-1}\scaledXTMX\bK(\bK^\prime\bPi\scaledXTMX\bK)^{-1} \\
    &= \{\bK^\prime\bPi\scaledXTMX\bK + \bK^\prime\bPi^2(\bI-\bPi)^{-1}\scaledXTMX\bK\} \\
    &\quad \cdot(\bK^\prime\bPi\scaledXTMX\bK)^{-1}
\end{align*}

To further simplify, we multiply by $(\bI-\bPi)(\bI-\bPi)^{-1}$ within the first term. This allows us to group the first and second terms together, and the $\bPi^2$ term cancels out.
\begin{align*}
    &[\bK^\prime\{(\bPi - \bPi^2)(\bI - \bPi)^{-1} + \bPi^2(\bI - \bPi)^{-1}\}\scaledXTMX\bK]\\
    &\quad \cdot(\bK^\prime\bPi\scaledXTMX\bK)^{-1} \\
    &= \{\bK^\prime\bPi(\bI - \bPi)^{-1}\scaledXTMX\bK\}(\bK^\prime\bPi\scaledXTMX\bK)^{-1}.
\end{align*}
Finally, we substitute this term back into the expression for $\hbtheta(\bpi)$. After cancellation of terms, we obtain
\begin{align}\label{eqn:theta-simplified}
    \hbtheta(\bpi)  &=  \{\bK^\prime\bPi(\bI - \bPi)^{-1}\scaledXTMX\bK\}^{-1}
     \{\bK^\prime\bPi(\bI - \bPi)^{-1}\scaledXTMX\bK\} \nonumber\\ 
     &\quad \cdot (\bK^\prime\bPi\scaledXTMX\bK)^{-1} \bK^\prime\bPi\scaledXTMX\tbbeta \nonumber \\
     &= (\bK^\prime\bPi\scaledXTMX\bK)^{-1} \bK^\prime\bPi\scaledXTMX\tbbeta .
\end{align}
Equation \eqref{eqn:theta-simplified} reveals that $\hbbeta(\bpi)$ in equation \eqref{eqn:beta-less-complicated} can be written as a weighted average of the individual group MLEs and the estimated centroid:
\begin{align*}
\hbbeta(\bpi) &= (\bI-\bPi)\tbbeta + \bPi\bK\hbtheta(\bpi).
\end{align*}

\section{Derivation of MSE}  \label{appsec:derivation_mse}
In this section, we derive the mean squared error (MSE) of $\hbbeta(\bpi)$. Recall from equation \eqref{eqn:hbbeta_reexpr} in Section \ref{subsec:est_all} that we may write $\hbbeta(\bpi)$ in terms of the MLE $\tbbeta$ through
\begin{align*}
\hbbeta(\bpi) &= \tbbeta + \bPi\{\bK\bA(\bpi)-\bI_{pk}\}\tbbeta.
\end{align*}
We begin by rewriting the MSE of $\hbbeta(\bpi)$ in terms of the bias and the variance:
\begin{align*}
    \MSE\{\hbbeta(\bpi)\} &= \EV[\{\hbbeta(\bpi) - \bbeta\}^\prime \{\hbbeta(\bpi) - \bbeta\}]\\
    &=  \EV([\hbbeta(\bpi) - \EV\{\hbbeta(\bpi)\} + \EV\{\hbbeta(\bpi)\} - \bbeta]^\prime [\hbbeta(\bpi) -\EV\{\hbbeta(\bpi)\} + \EV\{\hbbeta(\bpi)\} - \bbeta])\\
    &= \EV([\hbbeta(\bpi) - \EV\{\hbbeta(\bpi)\}]^\prime[\hbbeta(\bpi) - \EV\{\hbbeta(\bpi)\}]) \\ 
    &\quad + 
    2\EV([\hbbeta(\bpi) - \EV\{\hbbeta(\bpi)\}]^\prime[\EV\{\hbbeta(\bpi)\} - \bbeta]) \\
    &\quad + \EV([\EV\{\hbbeta(\bpi)\} - \bbeta]^\prime[\EV\{\hbbeta(\bpi)\}-\bbeta]).
\end{align*}
The first term in the MSE expression above is a scalar, and thus is unchanged by the trace operator:
\begin{align} \label{eqn:MSE-first}
     \tr\{\EV([\hbbeta(\bpi) - \EV\{\hbbeta(\bpi)\}]^\prime[\hbbeta(\bpi) - \EV\{\hbbeta(\bpi)\}])\} &= \tr\{\EV([\hbbeta(\bpi) - \EV\{\hbbeta(\bpi)\}][\hbbeta(\bpi) - \EV\{\hbbeta(\bpi)\}]^\prime)\} \nonumber\\
     &= \tr[\Var\{\hbbeta(\bpi)\}].
\end{align} 
The second term in the MSE expression is $0$:
\begin{align} \label{eqn:MSE-second}
    \EV([\hbbeta(\bpi) - \EV\{\hbbeta(\bpi)\}]^\prime[\EV\{\hbbeta(\bpi)\} - \bbeta]) &=\EV\{\hbbeta(\bpi)\}^\prime[\EV\{\hbbeta(\bpi)\} - \bbeta] - \EV\{\hbbeta(\bpi)\}^\prime[\EV\{\hbbeta(\bpi)\} - \bbeta] \nonumber \\
    &= 0.
\end{align}
The third term in the MSE expression is the inner product of the bias of $\hbbeta(\bpi)$. Together with equations \eqref{eqn:MSE-first} and \eqref{eqn:MSE-second}, this allows us to rewrite the MSE of $\hbbeta(\bpi)$ as
\begin{align} \label{eqn:beta-MSE}
    \MSE\{\hbbeta(\bpi)\} &= \tr[\Var\{\hbbeta(\bpi)\}] 
        + \|\Bias\{\hbbeta(\bpi)\}\|^2.
\end{align}
 The variance of $\hbbeta(\bpi)$ is
\begin{align} \label{eqn:beta-var}
    \Var\{\hbbeta(\bpi)\} &= \Var(\tbbeta) + 2\Cov[\bPi\{\bK\bA(\bpi) - \bI_{pk}\}\tbbeta, \tbbeta] + \Var[\bPi\{\bK\bA(\bpi) - \bI_{pk}\}\tbbeta],
\end{align}
and the bias is
\begin{align} \label{eqn:beta-bias}
    \Bias\{\hbbeta(\bpi)\} &= \EV[\tbbeta + \bPi\{\bK\bA(\bpi) - \bI_{pk}\}\tbbeta] - \bbeta \nonumber \\
     &= \bPi\{\bK\bA(\bpi) - \bI_{pk}\}\bbeta.
\end{align}
Plugging equations \eqref{eqn:beta-var} and \eqref{eqn:beta-bias} into equation \eqref{eqn:beta-MSE}, the MSE of $\hbbeta(\bpi)$ and, equivalently, the MSE of $\hbbeta(c, \bpi_r)$, are 
\begin{equation}
\begin{split}
\label{eqn:general_mse}
    & \tr\{\Var(\tbbeta)\} 
    + \tr(\Var[\bPi\{\bK\hbtheta(\bpi) - \tbbeta\}])
    + 2\cdot\tr(\Cov[\tbbeta, \bPi\{\bK\hbtheta(\bpi) - \tbbeta\}])\\
    &\quad + \|\bPi\{\bK\bA(\bpi) - \bI_{pk}\}\bbeta\|^2.
\end{split}
\end{equation}

\section{Proof of Theorem \ref{existence_ma}} \label{appsec:imp_single_tuning}
In this proof, we demonstrate that there is a vector of $\bpi \neq \boldsymbol{0}$ such that $\MSE\{\hbbeta(\bpi)\}<\MSE(\tbbeta)$. We show that such a vector exists by restricting $\bpi$ to the subset where $\pi_j\equiv \pi$.

First, we write the MSE of $\hbbeta(\bpi)$, given in Equation \eqref{eqn:general_mse}, for the case where $\pi_j\equiv\pi$. To do so, first note that when all $\pi_j$ are equal, $\bA(\bpi) = (\bK^\prime\scaledXTX\bK)^{-1}\bK^\prime\scaledXTX$: $\bA(\bpi)$ is no longer a function of $\pi$ and we write $\bA(\bpi)=\bA$. Then, 
\begin{align*}
    \MSE\{\hbbeta(\pi)\} &=  \tr\{(\scaledXTX)^{-1}\} + 2\pi\cdot\tr\{(\bK\bA - \bI_{pk})(\scaledXTX)^{-1}\}\\&\quad + \pi^2\cdot\tr\{(\bK\bA - \bI_{pk})(\scaledXTX)^{-1}(\bK\bA - \bI_{pk})^\prime\} + \pi^2\cdot\|(\bK\bA - \bI_{pk})\bbeta\|^2.
\end{align*}

To select the $\pi$ that minimizes the MSE, we find the root of the derivative of the $\MSE\{\hbbeta(\pi)\}$ with respect to $\pi$:
\begin{align*}
    \frac{\partial}{\partial \pi} \MSE\{\hbbeta(\pi)\} &= 2\pi\cdot\|(\bK\bA - \bI_{pk})\bbeta\|^2 + 2\pi\cdot\tr\{(\bK\bA - \bI_{pk})(\scaledXTX)^{-1}(\bK\bA - \bI_{pk})^\prime\} \\
    &\quad + 2\cdot\tr\{(\bK\bA - \bI_{pk})(\scaledXTX)^{-1}\}.
\end{align*}
    The root $\pi^\star$ of this derivative satisfies
\begin{align*}
    0 &= 2\pi^{\star}\cdot\tr\{(\bK\bA - \bI_{pk})(\scaledXTX)^{-1}(\bK\bA - \bI_{pk})^\prime\} + 2\cdot\tr\{(\bK\bA - \bI_{pk})(\scaledXTX)^{-1}\}  \\
     &\quad +  2\pi^{\star}\cdot\|(\bK\bA - \bI_{pk})\bbeta\|^2.
\end{align*}
Rearranging, $\pi^\star$ is given by
\begin{align*}
    \pi^{\star} &= \frac{-\tr\{(\bK\bA - \bI_{pk})(\scaledXTX)^{-1}\}}{\tr\{(\bK\bA - \bI_{pk})(\scaledXTX)^{-1}(\bK\bA - \bI_{pk})^\prime\} + \|(\bK\bA - \bI_{pk})\bbeta\|^2}. 
\end{align*}
The second derivative of the MSE with respect to $\pi$ is 
\begin{align*}
    2\|(\bK\bA - \bI_{pk})\bbeta\|^2 + 2 \tr\{(\bK\bA - \bI_{pk})(\scaledXTX)^{-1}(\bK\bA - \bI_{pk})^\prime\} \geq 0,
\end{align*}
and so $\pi^{\star}$ minimizes the MSE.

The final step is to show that $\pi^{\star}>0$. The denominator for $\pi^{\star}$ is strictly positive, as it consists of the sum of the trace of a variance matrix and an inner product. Thus, we need only consider the numerator,
\begin{align*}
    -\tr\{(\bK\bA - \bI_{pk})(\scaledXTX)^{-1}\} &= 
    -\tr\{\bK(\bK^\prime \scaledXTX \bK)^{-1}\bK^\prime\} + \tr\{(\scaledXTX)^{-1}\}\\
    &= -\tr\{\bK^\prime\bK(\bK^\prime \scaledXTX \bK)^{-1}\} + \tr\{\bK^\prime(\scaledXTX)^{-1}\bK\} \\
    &= -k\cdot\tr\{(\bK^\prime \scaledXTX \bK)^{-1}\} + \tr\{\bK^\prime(\scaledXTX)^{-1}\bK\},
\end{align*}
where we used the fact that $\bK^\prime\bK = k\bI_p$ to factor out $k$ in the first term. The second term in the numerator is
\begin{align*}
    \tr\{\bK^\prime(\scaledXTX)^{-1}\bK\} &= \tr\{\sum_{j=1}^k (\scaledXTXj)^{-1}\} = \sum_{\ell=1}^p\sum_{j=1}^k (\scaledXTXj)^{-1}_{\ell \ell}\\
    &= \sum_{\ell=1}^{pk}(\scaledXTX)^{-1}_{\ell \ell} = \tr\{(\scaledXTX)^{-1}\}.
\end{align*}

We now show that, for $k \geq 2$, 
$\tr\{\bK^\prime(\scaledXTX)^{-1}\bK - k (\bK^\prime \scaledXTX \bK)^{-1}\}> 0$, which completes the proof.
\begin{proof}
It will suffice to show that the interior of the trace, 
\begin{align*}
    \bK^\prime(\scaledXTX)^{-1}\bK  - k (\bK^\prime \scaledXTX \bK)^{-1},
\end{align*}
is positive definite. Following the proof technique in \citet{zhan_inequalities_2002}, define a block matrix:
\begin{align*}
    \bD &= \begin{bmatrix}
        k^{-1}\bK^\prime\scaledXTX\bK & \bI_p\\
        \bI_p & k^{-1}\bK^\prime(\scaledXTX)^{-1}\bK
    \end{bmatrix}\\
    &=
    \begin{bmatrix}
        k^{-1}\bK^\prime\scaledXTX\bK & k^{-1}\bK^\prime\bK\\
        k^{-1}\bK^\prime\bK & k^{-1}\bK^\prime(\scaledXTX)^{-1}\bK
    \end{bmatrix}
\end{align*}
If $\bD$ is positive semi-definite, then $k^{-1}\bK^\prime(\scaledXTX)^{-1}\bK - (k^{-1}\bK^\prime\scaledXTX\bK)^{-1}$ is also positive semi-definite \citep{horn_chapter_1985}:
\begin{align*}
    \bK^\prime(\scaledXTX)^{-1}\bK &\succeq k^{2}(\bK^\prime\scaledXTX\bK)^{-1}
\end{align*}
This relationship would imply that 
\begin{align*}
\bK^\prime(\scaledXTX)^{-1}\bK \succ k(\bK^\prime\scaledXTX\bK)^{-1},
\end{align*}
for $k \geq 2$. Thus, we need to show that $\bD$ is positive semi-definite.

To see that $\bD$ is positive semi-definite, first note that $\scaledXTX = \sum_{\ell=1}^{pk}\xi_{\ell} \bu_{\ell} \bu_{\ell}^\prime$, where $u_{\ell}$ is the normalized eigenvector for the $\ell$th eigenvalue, $\xi_{\ell}$, of $\scaledXTX$. Further, as $\sum_{\ell=1}^{pk}\bu_{\ell} \bu_{\ell}^\prime = \bI_{pk}$, we can rewrite $\bD$ as
\begin{align*}
    \bD &= \begin{bmatrix}
    k^{-1}\bK^\prime\sum_{\ell=1}^k\xi_{\ell} \bu_{\ell}\bu_{\ell}^\prime \bK   & k^{-1}\bK^\prime\sum_{\ell=1}^k \bu_{\ell}\bu_{\ell}^\prime \bK \\
    k^{-1}\bK^\prime\sum_{\ell=1}^k  \bu_{\ell}\bu_{\ell}^\prime \bK & k^{-1}\bK^\prime\sum_{\ell=1}^k\xi_{\ell}^{-1} \bu_{\ell}\bu_{\ell}^\prime \bK  
    \end{bmatrix}\\
    &=\sum_{\ell = 1}^{pk} \biggl( \begin{bmatrix}
        \xi_{\ell} & 1 \\
        1 & \xi_{\ell}^{-1}
    \end{bmatrix} \otimes k^{-1}\bK^\prime \bu_{\ell} \bu_{\ell}^\prime\bK \biggr).
\end{align*}

We argue that the matrices to the left and right of the Kronecker product are positive semi-definite, and then show that the eigenvalues of a Kronecker product of positive semi-definite matrices are non-negative. The matrix to the left of the Kronecker product has eigenvalue $\xi_{\ell} + \xi_{\ell}^{-1}$, $\xi_{\ell}>0$, and so is positive semi-definite. The matrix to the right of the Kronecker product is diagonal with non-negative values along the diagonal, so it is positive semi-definite: 
\begin{align*}
    \bK^\prime\bu_{\ell}\bu_{\ell}^\prime\bK &= \bK^\prime\begin{bmatrix}
        u_{11}\\
        \vdots\\
        u_{kp}
    \end{bmatrix}\begin{bmatrix}
        u_{11}\\
        \vdots\\
        u_{kp}
    \end{bmatrix}^\prime\bK \\
    &= \begin{bmatrix}
        \sum_{i = 1}^k u_{i1}^2 & \dots & 0 \\
        \vdots & \ddots &\vdots \\
        0 & \dots & \sum_{i = 1}^k u_{ip}^2.
    \end{bmatrix}
\end{align*}
Now we show that the eigenvalues of a Kronecker product of matrices $\bA$ and $\bB$ are the product of the eigenvalues of $\bA$ and $\bB$. Let $\bx$ be an eigenvector for $\bA$ with corresponding eigenvalue $\lambda$ and $\by$ be an eigenvector for $\bB$ with corresponding eigenvalue $\mu$:
\begin{align*}
(\bA\otimes\bB)(\bx\otimes \by) &= (\bA\bx)\otimes(\bB\by)\\
&= \lambda\bx \otimes \mu \by = \lambda\mu(\bx\otimes \by)
\end{align*}
Then, the Kronecker product of two positive semi-definite matrices is also positive semi-definite, so $\bD$ is positive semi-definite and $\bK'(\scaledXTX)^{-1}\bK \succ k\times(\bK^\prime\scaledXTX\bK)^{-1}$. Consequently, 
\begin{align*}
    -k\tr\{(\bK^\prime \scaledXTX \bK)^{-1}\} + \tr\{\bK^\prime(\scaledXTX)^{-1}\bK\} > 0,
\end{align*}
and $\pi^{\star} > 0$. As a result, for any $\bbeta$ and $\scaledXTX$, there exists a $\hbbeta(\bpi)$ that has a smaller MSE than the MLE. When all $\pi_j\equiv\pi$, $\pi^{\star}$ minimizes the MSE.
\end{proof}

When $\pi_j\equiv\pi$, we can further show that $\MSE\{\hbbeta(\pi) \} < \MSE(\tbbeta)$ for $\pi \in (0,2\pi^{\star})$. Indeed, $\MSE\{\hbbeta(\pi)\}$ is a quadratic function of $\pi$ and is thus symmetric about its minimum at $\pi^{\star}$; since $\MSE\{\hbbeta(0)\} = \MSE(\tbbeta)$, we know that $\MSE\{\hbbeta(2\pi^{\star})\} = \MSE(\tbbeta)$. As a consequence, any $\pi \in (0,2\pi^{\star})$ yields an improvement in MSE over the MLE.

\section{Proof of Theorem \ref{optimal_scaling}} \label{appsec:optimal_scaling}
Define $\bK = (\bone_k \otimes \bI_p)$ and let $c\in [0, 1]$. The estimators are 
\begin{align*}
\hbtheta(\bpi_r) &= 
        \Bigl(\sum_{j=1}^k \pi_{r,j} \scaledXTXj\Bigr)^{-1} \sum_{j=1}^k\pi_{r,j}\scaledXTXj\tbbeta_j\\
\hbbeta(c, \bpi_r) &= (\bI_{pk} - c\bPi_r)\tbbeta - c\bPi_r \bK\hbtheta(\bpi),
\end{align*}
where at least one $\pi_{r,j}=1$

The MSE of $\hbbeta(c, \bpi)$ is 
    \begin{align*}
    & \tr\{\Var(\tbbeta)\} 
    + \tr(\Var[\bPi\{\bK\hbtheta(\bpi) - \tbbeta\}])
    + 2~\tr(\Cov[\tbbeta, \bPi\{\bK\hbtheta(\bpi) - \tbbeta\}])\\
    &\quad + \|\bPi\{\bK\bA(\bpi) - \bI_{pk}\}\bbeta\|^2\\
    &= \tr\{\Var(\tbbeta)\}  
    + 2c~\tr(\Cov[\tbbeta, \bPi_r\{\bK\bA(\bpi_r) - \bI_{pk}\}\tbbeta])\\
    &\quad + c^2~\tr(\Var[\bPi_r\{\bK\bA(\bpi_r) - \bI_{pk}\}\tbbeta]) + c^2~\|\bPi_r\{\bK\bA(\bpi_r) - \bI_{pk}\}\bbeta\|^2.
\end{align*}
The root of the derivative of this MSE with respect to $c$ is
\begin{align*}
    c^\star(\bpi_r) &= \frac{-\tr(\Cov[\tbbeta, \bPi_r\{\bK\bA(\bpi_r) - \bI_{pk}\}\tbbeta])}{\tr(\Var[\bPi_r\{\bK\bA(\bpi) - \bI_{pk}\}\tbbeta]) + \|\bPi_r\{\bK\bA(\bpi_r) - \bI_{pk}\}\bbeta\|^2}
\end{align*}
This will be a minimum, as the second derivative of the MSE with respect to $c$ is positive.

The denominator for $c^{\star}(\bpi_r)$ is always positive, as it is the sum of an inner product and the trace of a variance. Then, if $\tr(\Cov[\tbbeta, \bPi_r\{\bK\bA(\bpi_r) - \bI_{pk}\}\tbbeta])< 0$ for all constrained $\bpi_r$, every $\bpi_r$ offers an improvement over the MLE with sufficiently small scaling. To see the conditions under which this term will be strictly negative, we first rewrite it in terms of $\scaledXTX$.
\begin{align*}
    \tr&(\Cov[\tbbeta, \bPi_r\{\bK\bA(\bpi_r) - \bI_{pk}\}\tbbeta]) \\
    &= \tr(\Cov[\tbbeta, \bPi_r\{\bK(\bK^\prime\bPi_r\scaledXTX\bK)^{-1}\bK^\prime\bPi_r\scaledXTX - \bI_{pk}\}\tbbeta])\\
   &= \tr[\Var(\tbbeta)
   \{\scaledXTX\bPi_r\bK(\bK^\prime\bPi_r\scaledXTX\bK)^{-1}\bK^\prime - \bI_{pk}\}\bPi_r]
\end{align*}
The variance of $\tbbeta$ is $(\scaledXTX)^{-1}$. We substitute this term in:
\begin{align*}
   &= \tr[
   \{\bPi_r\bK(\bK^\prime\bPi_r\scaledXTX\bK)^{-1}\bK^\prime - (\scaledXTX)^{-1}\}\bPi_r]\\
   &= \tr\{
   \bK^\prime\bPi_r^2\bK(\bK^\prime\bPi_r\scaledXTX\bK)^{-1}\} - \tr\{\bPi_r(\scaledXTX)^{-1}\}
\end{align*}
The trace of $\bPi_r(\scaledXTX)^{-1}$ is equal to $\bK^\prime\bPi_r(\scaledXTX)^{-1}\bK$. Then, to show that the numerator of $c^{\star}(\bpi_r)$ is negative, we need to show
\[\bK^\prime\bPi_r(\scaledXTX)^{-1}\bK \succ \bK^\prime\bPi_r^2\bK(\bK^\prime\bPi_r\scaledXTX\bK)^{-1}.\]
This term will always be less than or equal to zero, no matter the value of $\bpi_r$ - it follows from the same logic that allowed us to say that $\sum_{j=1}^k (\scaledXTXj)^{-1} \succ k\times \{\sum_{j=1}^k (\scaledXTXj)^{-1}\}$ in Section \ref{appsec:imp_single_tuning}. First, write a matrix $M$:
\begin{align*}
    \bD = \begin{bmatrix}
        \bK^\prime\bPi_r\scaledXTX\bK(\bK^\prime\bPi_r\bK)^{-1} & \bI_p = \bK^\prime\bPi_r\bK(\bK^\prime\bPi_r\bK)^{-1}\\
        \bI_p= \bK^\prime\bPi_r\bK(\bK^\prime\bPi_r\bK)^{-1} & \bK^\prime\bPi_r(\scaledXTX)^{-1}\bK(\bK^\prime\bPi_r\bK)^{-1}
    \end{bmatrix} 
\end{align*}

As in the proof technique for \ref{existence_ma}, if $\bD$ is positive semi-definite, then
\begin{align*}
     \bK^\prime\bPi_r(\scaledXTX)^{-1}\bK(\bK^\prime\bPi_r\bK)^{-1} &
    - (\bK^\prime\bPi_r\bK)(\bK^\prime\bPi_r\scaledXTX\bK)^{-1} \succeq 0\\
     \bK^\prime\bPi_r(\scaledXTX)^{-1} &\succeq (\bK^\prime\bPi_r\bK)(\bK^\prime\bPi_r\scaledXTX\bK)^{-1} (\bK^\prime\bPi_r\bK)\\
     \bK^\prime\bPi_r(\scaledXTX)^{-1} &\succeq (\sum_{j=1}^k \pi_j)^2(\bK^\prime\bPi_r\scaledXTX\bK)^{-1} \\
    & \succ \sum_{j=1}^k \pi_j^2 (\bK^\prime\bPi_r\scaledXTX\bK)^{-1}
\end{align*}
The last ordering holds as long as at least two $\pi_j \notin \{0, 1\}$; that is, as long as $\hbbeta(c, \bpi_r)\neq \tbbeta$. To see that $\bD$ is positive definite, as before write $\scaledXTX = \sum_{k=1}^{pk}\xi_k u_k u_k^\prime$, where $u_k$ is the normalized eigenvector for the $k$th eigenvalue. The sum, $\sum_{k=1}^{pk}u_k u_k^\prime$ is the $pd \times pd$ identity matrix. Then rewrite $\bD$ as the following sum:
\begin{align*}
    \bD &= \sum_{k = 1}^{pk} \biggl( \begin{bmatrix}
        \xi_k & 1 \\
        1 & \xi_k^{-1}
    \end{bmatrix} \otimes \bK^\prime\bPi_r u_k u_k^\prime\bK(\bK^\prime\bPi_r\bK)^{-1} \biggr)
\end{align*}

This is the Kronecker product of two positive semi-definite matrices, so it is positive semi-definite.

\section{Proof of Corollary \ref{improvement_bpi}} \label{improvement_proof}

Consider the MSE of $\hbbeta(c, \bpi_r)$.
We want to find $c$ such that $\MSE(\tbbeta) - \MSE\{\hbbeta(c, \bpi_r)\}\geq 0$, i.e.
\begin{align*}
    0 & \leq \MSE(\tbbeta) - \MSE\{\hbbeta(c, \bpi_r)\} \\
    &=- c^2\times\tr(\Var[\bpi_r\{\bK\hbtheta(\bpi_r) - \tbbeta\}])
    - 2c\times\tr(\Cov[\tbbeta, \bpi_r\{\bK\hbtheta(\bpi_r) - \tbbeta\}])\\
    & \quad ~ - c^2\times\|\bpi_r\{\bK\btheta(\bpi_r) - \bbeta\}\|^2.
\end{align*}
Rearranging, we want to find $c$ such that
\begin{align*}
    0&\geq c\times\tr(\Var[\bpi_r\{\bK\hbtheta(\bpi_r) - \tbbeta\}])
     + 2\times\tr(\Cov[\tbbeta, \bpi_r\{\bK\hbtheta(\bpi_r) - \tbbeta\}])\\
    & \quad ~ + c \times\|\bpi_r\{\bK\btheta(\bpi_r) - \bbeta\}\|^2.
\end{align*}
The solution is the set
\begin{align*}
    &\Bigl\{c: c \leq \frac{-2 \tr(\Cov[\tbbeta, \bpi_r\{\bK\hbtheta(\bpi_r) - \tbbeta\}]) }
    {\tr(\Var[\bPi\{\bK\hbtheta(\bpi_r) - \tbbeta\}]) + \|\bpi_r\{\bK\btheta(\bpi_r) - \bbeta\}\|^2}\Bigr\} = \{c:c \leq 2c^\star\}.
\end{align*}
That is, for a given $\bpi = c(\bpi_r)\bpi_r$, $\MSE\{\hbbeta(c, \bpi_r)\}\leq \MSE(\tbbeta)$ if and only if $c \leq 2 c^\star(\bpi_r)$.

\section{Proof of Theorem \ref{consistency}} \label{appsec:consistency}

In this section, we show that the estimator $\hbbeta(c^\star, \bpi_r)$ is consistent under the assumptions of Theorem \ref{consistency}. The proof proceeds as follows. We first show the exact distribution for $\hbbeta(c, \bpi_r)$ with a general scaling factor $c$ and then demonstrate that the estimator's bias goes to zero as $n\to \infty$ for the optimal scaling $c^{\star}(\bpi_r)$. Finally, we show that the estimator's total variance also goes to zero as $n \to \infty$. We occasionally suppress $c^\star$'s dependence on $\bpi_r$ for notational brevity.

To begin, write $\hbbeta(c, \bpi_r)$ in terms of $\tbbeta$, as in equation \eqref{eqn:hbbeta_cscale}: 
\begin{align*}
    \hbbeta(c, \bpi) = [\bI_{pk} + c\bPi_r\{\bK\bA(\bpi_r) - \bI_{pk}\}]\tbbeta.
\end{align*}

As $\tbbeta\sim \mathcal{N} \{\bbeta, (\scaledXTX)^{-1}\}$, the estimator $\hbbeta(c, \bpi_r)$ is also normally distributed with mean 
\begin{align*}
&\bbeta +  c\bPi_r\{\bK \bA(\bpi_r) - \bI_{pk}\}\bbeta
\end{align*}
and variance
\begin{align*}
    [\bI_{pk} + 
        c\bPi_r\{\bK\bA(\bpi_r)-\bI_{pk}\}](\scaledXTX)^{-1} 
        [\bI_{pk} + 
        c\{\bK\bA(\bpi_r)-\bI_{pk}\}^\prime\bPi_r]
\end{align*}
We now argue that the bias goes to zero asymptotically.
The bias vector for $\hbbeta(c, \bpi_r)$ is $c\bPi_r\{\bK\bA(\bpi_r) - \bI_{pk}\}\bbeta$. We begin by showing that the matrix $\bK\bA(\bpi_r) - \bI_{pk}$ is asymptotically a constant for a fixed $\bpi_r$. Denote the scaled Gram matrices, $\bG_j = n_j^{-1} \scaledXTXj$. As assumed above, $\bG_j^{-1}\to \bQ_j^{-1}$, a positive definite matrix, as $n \to \infty$. The matrix $\bA(\bpi_r)$ can be re-expressed in terms of the scaled Gram matrices, 
\begin{align*}
\bA(\bpi_r) &= (\bK^\prime\bPi_r \scaledXTX\bK)^{-1}\bK^\prime\bPi_r\scaledXTX\\
&= \begin{bmatrix}
    (\sum_{j=1}^k \pi_{r,j}  a_j \bG_j)^{-1}\pi_{r,1} a_1 \bG_1 &\dots &(\sum_{j=1}^k \pi_{r,j} a_j \bG_j)^{-1}\pi_{r,k} a_k \bG_k 
\end{bmatrix}
\end{align*}

To simplify the notation, let $\bB_n(\bpi_r) = \bK\bA(\bpi_r) - \bI_{pk}$. We examine the behavior of $\bB_n(\bpi_r)$ as $n\to\infty$, for an arbitrary and fixed value of $\bpi_r$. By the continuous mapping theorem, $\bG_j \to \bQ_j$ and therefore, again by the continuous mapping theorem, $\sum_j \pi_{r,j} a_j \bG_j \to \sum_j \pi_{r,j} a_j \bQ_j$. 
As we are considering a fixed $\bpi_r$, we may take a weighted sum of the matrices $\bG_j$ and they will converge to the weighted sum of $\bQ_j$. Then, the inverse of the weighted sum of $\bG_j$ converges to the weighted sum of $\bQ_j$:
\begin{align*}
\Bigl(\sum_j \pi_{r,j} a_j \bG_j \Bigr)^{-1} \to \Bigl(\sum_j \pi_{r,j} a_j \bQ_j \Bigr)^{-1}.
\end{align*}
This implies that
\begin{align*}
&\Bigl(\sum_j \pi_{r,j} a_j \bG_j \Bigr)^{-1}\pi_{r, i} a_i \bG_i \to
\Bigl(\sum_j \pi_{r,j} a_j \bQ_j \Bigr)^{-1}\pi_{r, i} a_i \bQ_i,
\end{align*}
which again implies that
\begin{align*}
\Bigl(\sum_j \pi_{r, j} a_j \bG_j \Bigr)^{-1}\pi_{r, i} a_i \bG_i - \bI \to
\Bigl(\sum_j \pi_{r, j} a_j \bQ_j \Bigr)^{-1}\pi_{r, i} a_i \bQ_i  - \bI.
\end{align*}
The left-hand side of the above equation is a single sub-matrix on the diagonal of $\bB_n(\bpi_r)$:
\begin{align*}
    \bB_n(\bpi_r)&=\begin{bmatrix}
        (\sum_j \pi_{r, j} a_j \bG_j)^{-1}\pi_{r, 1} a_1 \bG_1 - \bI  & \dots &(\sum_j \pi_{r, j} a_j \bG_j)^{-1}\pi_{r, k} a_k \bG_k\\
        \vdots & \ddots & \vdots\\
        (\sum_j \pi_{r, j} a_j \bG_j)^{-1}\pi_{r, 1} a_1 \bG_1 & \dots & (\sum_j \pi_{r, j} a_j \bG_j)^{-1}\pi_{r, k} a_k \bG_k - \bI
    \end{bmatrix}.
\end{align*}
Each sub-matrix of $\bB_n(\bpi_r)$ converges to a constant matrix for an arbitrary fixed $\bpi_r$, so $\bB_n(\bpi_r) \to \bB(\bpi_r)$. Denote $\btheta_{n\to\infty} = \lim_{n \to \infty} \btheta_n$. Then, for arbitrary $\bpi_r$, 
\begin{align*}
    \Bias\{\hbbeta(c, \bpi_r)\} &= c\bPi_r\bB_n(\bpi_r)\bbeta  \to c\bPi_r\bB(\bpi_r)\bbeta = \begin{bmatrix}
        c\pi_{r,1}\{\btheta_{n\to\infty}(\bpi_r) - \bbeta_1\}\\
        \vdots \\
        c\pi_{r,k}\{\btheta_{n\to\infty}(\bpi_r) - \bbeta_k\}
    \end{bmatrix}
\end{align*}
We first consider the bias in the case with general scaling, $c$. The value of $\btheta_{n\to\infty}(\bpi_r)$ is finite, as its components are finite. Then, the bias goes to zero asymptotically if (i) $\bbeta_j = \bbeta_{j^\prime}$ for all groups $j$, $j^\prime$, (ii) $\pi_{r,j}\{\btheta_{n\to\infty}(\bpi_r) - \bbeta_j \}\to \bzero$ for every group $j$, or (iii) if $c\to 0$. For case (i), $\bbeta = \bK\bbeta_1$ (where $\bbeta_1$ is selected without loss of generality). Then,
\begin{align*}
    c\bPi_r\{\bK\bA(\bpi_r) - \bI_{pk}\}\bK\bbeta_1 &= c\bPi_r\{\bK\bA(\bpi_r)\bK - \bK\}\bbeta_1 .
\end{align*}
The bracketed term on the right-hand side is zero as $\bA(\bpi_r)\bK = \bI_p$, and so the bias is zero.

Now we consider cases (ii) and (iii), using the optimal scaling, $c^{\star}(\bpi_r)$ at an arbitrary $\bpi_r$. Recall that
\begin{align*}
    c^{\star}(\bpi_r) &= \frac{-\tr[\Cov\{\tbbeta, \bPi_r\bB_n(\bpi_r)\tbbeta\}]}{\tr[\Var\{\bPi_r\bB_n(\bpi_r)\tbbeta)\}] + \|\bPi_r\bB_n(\bpi_r)\bbeta\|^2}
    \\
    &= \frac{-\tr\{\bPi_r\bB_n(\bpi_r)(\scaledXTX)^{-1}\}}{\tr\{\bPi_r\bB_n(\bpi_r)(\scaledXTX)^{-1}\bB_n(\bpi_r)\bPi_r\} + \|\bPi_r\bB_n(\bpi_r)\bbeta\|^2}.
\end{align*}
We rescale trace terms in the numerator and denominator to put them more clearly in terms of $n$:
\begin{align*}
    c^{\star}(\bpi_r) &= \frac{-n^{-1}\tr\{\bPi_r\bB_n(\bpi_r)\bG^{-1}\}}{n^{-1}\tr\{\bPi_r\bB_n(\bpi_r)\bG^{-1}\bB_n(\bpi_r)\bPi_r\} + \|\bPi_r\bB_n(\bpi_r)\bbeta\|^2}.
\end{align*}
The traces contain only matrices that are finite as $n\to \infty$, so that
\begin{align*}
    n^{-1}\tr\{\bPi_r\bB_n(\bpi_r)\bG^{-1}\} &\to 0,\\
    n^{-1}\tr\{\bPi_r\bB_n(\bpi_r)\bG^{-1}\bB_n(\bpi_r)\bPi_r\}  &\to 0.
\end{align*}
Then, in case (iii) where the inner product of the bias does not go to zero, $c^{\star}=o(1)$, and $\hbbeta(c, \bpi_r)$ is asymptotically unbiased. In case (ii), $\|\bPi_r\bB_n(\bpi_r)\bbeta\|^2$ goes to zero, implying $\pi_{r,j}\{\btheta_{n\to\infty}(\bpi_r) - \bbeta_j\} \to \bzero$ for all $j$; this also yields asymptotic unbiasedness of $\hbbeta(c, \bpi_r)$. Because we consider these at an arbitrary $\bpi_r$, $c^\star(\bpi_r)$ is asymptotically unbiased for all values of $\bpi_r$, including the optimal choice $\bpi_r^{\star}$, which may vary in $n$. 

Finally, we show that the total variance goes to zero. The total variance of $\hbbeta(c, \bpi_r)$ is
\begin{align*}
    \tr[\Var\{\hbbeta(c, \bpi_r)\}] &= \tr[\{\bI + c\bPi_r\bB_n(\bpi_r) \}(\scaledXTX)^{-1} \{\bI + c\bB_n(\bpi_r)^\prime\bPi_r\}] \\
    &= \tr\{(\scaledXTX)^{-1}\} 
        + 2c \times\tr\{\bPi_r\bB_n(\bpi_r)(\scaledXTX)^{-1}\}  \\
    &\quad +
        c^2\tr\{\bPi_r\bB_n(\bpi_r)(\scaledXTX)^{-1}\bB_n(\bpi_r)^\prime\bPi_r\}\\
    &=  \frac{1}{n}\tr(\bG^{-1}) + \frac{2c}{n} \tr\{\bPi_r\bB_n(\bpi_r)\bG^{-1}\}  +
        \frac{c^2}{n}\tr\{\bPi_r\bB_n(\bpi_r)\bG^{-1}\bB_n(\bpi_r)^\prime\bPi_r\}.
\end{align*}
Each trace term converges to a finite scalar because $\bB_n(\bpi_r)\to \bB(\bpi_r)$. Then, as long as $c = o(\sqrt{n})$, the total variance converges to zero as $n\to\infty$. This condition is met for general $c$, as the $c$ is at most 1.

\section{Expected value of biased estimator of the MSE} \label{appsec:expected_BMSE}
Plugging in the MLE $\tbbeta$ for $\bbeta$ within the MSE for $\hbbeta(c, \bpi_r)$ yields a biased MSE, 
\begin{align*}
     \BMSE\{\hbbeta(c, \bpi_r)\} &=  c^2  \|\bPi_r\{\bK\bA(\bpi_r) - \bI_{pk}\}\tbbeta\|^2    +  c^2 \tr(\Var [\bPi_r\{\bK\bA(\bpi_r) - \bI_{pk}\}\tbbeta]) \nonumber\\
     &\quad + 2c\times \tr(\Cov[\tbbeta, \bPi_r\{\bK\bA(\bpi_r) - \bI_{pk}\}\tbbeta])  + \tr\{\Var(\tbbeta)\}.
\end{align*}
The expected value of this function is 
\begin{align*}
     \EV[\BMSE\{\hbbeta(c, \bpi_r)\}] &=  c^2  \EV[\|\bPi_r\{\bK\bA(\bpi_r) - \bI_{pk}\}\tbbeta\|^2]    +  c^2 \tr(\Var [\bPi_r\{\bK\bA(\bpi_r) - \bI_{pk}\}\tbbeta]) \nonumber\\
     &\quad + 2c \times\tr(\Cov[\tbbeta, \bPi_r\{\bK\bA(\bpi_r) - \bI_{pk}\}\tbbeta])  + \tr\{\Var(\tbbeta)\}.
\end{align*}

\section{Derivation for Unbiased MSE Estimator}\label{appsec:umse_derivation}
Our derivation follows \cite{stein_estimation_1981}. We first rewrite the MSE in terms of $\tbbeta$.
\begin{align*}
\MSE\{\hbbeta(c,\bpi_r)\} &= \EV\|\hbbeta(c,\bpi_r) - \bbeta\|^2\\
    &= \EV\|\hbbeta(c, \bpi_r) - \tbbeta + \tbbeta - \bbeta\|^2\\
    &= \EV\|\hbbeta(c, \bpi_r) -\tbbeta\|^2 + 2\{\hbbeta(c, \bpi_r) - \tbbeta\}^{\prime}(\tbbeta - \bbeta) + \tr\{\Var(\tbbeta)\}\\
    &= \EV\|\hbbeta(c, \bpi_r) - \tbbeta\|^2 + 2\tr[\Cov\{\tbbeta, \hbbeta(c, \bpi_r) - \tbbeta\}] + \tr\{\Var(\tbbeta)\}.
\end{align*}
The expression $\hbbeta(c, \bpi_r) - \tbbeta$ is equal to $c\bPi_r\{\bK\bA(\bpi_r) - \bI\}\tbbeta$.
\begin{align*}
    c^2\EV\|\bPi_r\{\bK\bA(\bpi_r) - \bI\}\tbbeta\|^2  + 2c~\tr(\Cov[\tbbeta, \bPi_r\{\bK\bA(\bpi_r) - \bI_{pk}\}\tbbeta])
     + \tr\{\Var(\tbbeta)\}.
\end{align*}
which yields the unbiased risk estimator for the MSE (URE),
\begin{align*}
  c^2~ \|\bPi_r\{\bK\bA(\bpi_r) - \bI_{pk}\}\tbbeta\|^2 
     + 2c~ \tr(\Cov[\tbbeta, \bPi_r\{\bK\bA(\bpi_r) - \bI_{pk}\}\tbbeta]) + \tr\{\Var(\tbbeta)\}.
\end{align*}

\section{Surrogate MSE for Implementation} \label{appsec:pseudo_MSE_derivation}
In this section, we derive the objective function to estimate $\bpi_r$ and $c$, or equivalently, $\bpi$. We start from the adjustment to the derivative of the URE with respect to $c$, proposed in Section \ref{sec:implementation}. We refer to this shifted derivative as $f(c, \bpi_r)$, and it takes the form
\begin{align*}
    f(c, \bpi_r) &= \parDeriv{\ure\{\hbbeta(c, \bpi_r)\}}{c} + \frac{\partial^2\ure\{\hbbeta(c, \bpi_r)\}}{\partial c^2} (\hc^{\star} - \tilde{c}).
\end{align*}

The $c$ that minimizes the URE in equation \eqref{eqn:ure} is 
\begin{align*}
    \hc^{\star} &= \frac{-\tr(\Cov[\tbbeta, \bPi_r\{\bK\bA(\bpi_r) - \bI\}\tbbeta])}{ \|\bPi_r\{\bK\bA(\bpi_r) - \bI\}\tbbeta\|^2 },
\end{align*}
and the minimizer of the biased MSE in equation \eqref{eqn:biased_mse} is
\begin{align*}
    \tilde{c} &=  \frac{-\tr(\Cov[\tbbeta, \bPi_r\{\bK\bA(\bpi_r) - \bI\}\tbbeta])}{ \tr[\Var\{\bPi_r\{\bK\bA(\bpi_r) - \bI\}\tbbeta)\}] + \|\bPi_r\{\bK\bA(\bpi_r) - \bI\}\tbbeta\|^2 }.
\end{align*}

To simplify the notation, denote $\bB(\bpi) = \bK\bA(\bpi_r) - \bI_{pk}$. Plugging the form of these minimizers into the function $f(c, \bpi_r)$,
\begin{align*}
     f(c,  \bpi_r)&=  \parDeriv{\ure\{\hbbeta(c, \bpi_r)\}}{c} + \frac{\partial^2\ure\{\hbbeta(c, \bpi_r)\}}{\partial c^2}\\
     &\qquad \Bigl(\frac{-\tr\{\Cov(\tbbeta, \bPi_r\bB\tbbeta)\}}{ 
     \|\bPi_r \bB\tbbeta\|^2 }  + \frac{\tr\{\Cov(\tbbeta, \bPi_r\bB\tbbeta)\}}{ \tr\{\Var(\bPi_r\bB\tbbeta)\} + \|\bPi_r \bB\tbbeta\|^2 }\Bigr)  \\
     &=   \parDeriv{\ure\{\hbbeta(c, \bpi_r)\}}{c} + \frac{\partial^2\ure\{\hbbeta(c, \bpi_r)\}}{\partial c^2} \cdot\frac{-\tr\{\Var(\bPi_r\bB\tbbeta)\}\tr\{\Cov(\tbbeta, \bPi_r\bB\tbbeta)\}}{  \|\bPi_r\bB\tbbeta\|^2 [\tr\{\Var(\bPi_r\bB\tbbeta)\} + \|\bPi_r\bB\tbbeta\|^2]}.
\end{align*}
We substitute in the first and second derivatives for the URE with respect to $c$:
\begin{align*}
    f(c,\bpi) &=   2c \|\bPi_r\bB\tbbeta\|^2 + 2\tr(\Cov[\tbbeta, \bPi_r\bB\tbbeta])  \\ &\qquad +  2\|\bPi_r\bB\tbbeta\|^2  \frac{-\tr\{\Var(\bPi_r\bB\tbbeta)\}\tr\{\Cov(\tbbeta, \bPi_r\bB\tbbeta)\}}{  \|\bPi_r\bB\tbbeta\|^2 [\tr\{\Var(\bPi_r\bB\tbbeta)\} + \|\bPi_r\bB\tbbeta\|^2]}\\
    &=   2c \|\bPi_r\bB\tbbeta\|^2 + 2\tr(\Cov[\tbbeta, \bPi_r\bB\tbbeta]) -  2 \frac{\tr\{\Var(\bPi_r\bB\tbbeta)\}\tr\{\Cov(\tbbeta, \bPi_r\bB\tbbeta)\}}{  \tr\{\Var(\bPi_r\bB\tbbeta)\} + \|\bPi_r\bB\tbbeta\|^2}\\
        &=   2c \|\bPi_r\bB\tbbeta\|^2  +  2 \frac{\|\bPi_r\bB\tbbeta\|^2\tr\{\Cov(\tbbeta, \bPi_r\bB\tbbeta)\}}{  \tr\{\Var(\bPi_r\bB\tbbeta)\} + \|\bPi_r\bB\tbbeta\|^2}.
\end{align*}
Plugging back in the form of $\bB$, the shifted derivative is therefore
\begin{align*}
    f(c, \bpi_r) &=   2c \|\bPi_r\{\bK\bA(\bpi_r) - \bI_{pk}\}\tbbeta\|^2 \\
    &\quad +  2 \frac{\|\bPi_r\{\bK\bA(\bpi_r) - \bI_{pk}\}\tbbeta\|^2\tr\{\Cov(\tbbeta, \bPi_r\{\bK\bA(\bpi_r) - \bI_{pk}\}\tbbeta)\}}{  \tr(\Var[\bPi_r\{\bK\bA(\bpi_r) - \bI_{pk}\}\tbbeta]) + \|\bPi_r\{\bK\bA(\bpi_r) - \bI_{pk}\}\tbbeta\|^2}.
\end{align*}
The $c$ at which this function is $0$ is $\tilde{c}$, the minimizer for the biased MSE in equation \eqref{eqn:biased_mse}, but the vector $\bpi_r$ need not be the same as that which minimizes the biased MSE.

We next construct an objective function whose minimum is given by the root of $f(c, \bpi_r)$. Integrating $f(c, \bpi_r)$ over $c$, an objective function is given by
\begin{align*}
   & c^2  \|\bPi_r\{\bK\bA(\bpi_r) - \bI_{pk}\}\tbbeta\|^2 \\
    &\quad +  2c \frac{\|\bPi_r\{\bK\bA(\bpi_r) - \bI_{pk}\}\tbbeta\|^2\tr\{\Cov(\tbbeta, \bPi_r\{\bK\bA(\bpi_r) - \bI_{pk}\}\tbbeta)\}}{  \tr(\Var[\bPi_r\{\bK\bA(\bpi_r) - \bI_{pk}\}\tbbeta]) + \|\bPi_r\{\bK\bA(\bpi_r) - \bI_{pk}\}\tbbeta\|^2}.
\end{align*}

Finally, in practice, to avoid constraints on the shrinkage parameters, we minimize the objective function
\begin{align*}
    O(\bpi) &= \|\bPi\{\bK\bA(\bpi) - \bI_{pk}\}\tbbeta\|^2 \\
    &\quad +  2 \frac{\|\bPi\{\bK\bA(\bpi) - \bI_{pk}\}\tbbeta\|^2\tr\{\Cov(\tbbeta, \bPi\{\bK\bA(\bpi) - \bI_{pk}\}\tbbeta)\}}{  \tr(\Var[\bPi\{\bK\bA(\bpi) - \bI_{pk}\}\tbbeta]) + \|\bPi\{\bK\bA(\bpi) - \bI_{pk}\}\tbbeta\|^2},
\end{align*}
as a function of the unrestricted $\bpi$.

\section{Ridge-like estimator} \label{appsec:ridgelike}

In the simulation studies, we demonstrate the advantages of shrinking towards a centroid distribution by comparing the heterogeneity-adaptive estimator $\hbbeta(\bpi_{\has})$ to an estimator that uses an L2 penalty to shrink the estimators towards each other. We refer to this as the ``ridge-like'' estimator. Let $\bbeta = (\bbeta_1^\prime, \dots, \bbeta_k^\prime )^\prime \in \mathbb{R}^{pk}$. The ridge-like estimator maximizes 
\[
O(\bbeta, \lambda) = \sum_{j=1}^k \log f_j(\bY_j; \bbeta_j,\gamma_j) - \lambda\sum_{j\neq j^\prime}^k\|\bbeta_j - \bbeta_{j^\prime}\|^2\].

We re-express $\sum_{j\neq j^\prime}\|\bbeta_j - \bbeta_{j^\prime}\|^2$ as follows. Let $\bC$ be a $k\times {k\choose 2}$ matrix containing every combination of contrasts:
\[ \bC = \begin{bmatrix}
1 & 1 &\dots & 0\\
-1 & 0 &\dots & 0\\
0 & -1 & \dots & 0\\
\vdots & \vdots & \vdots & \vdots \\
0 & 0 &\dots & 1\\
0 & 0 & \dots & -1
\end{bmatrix}.\] 
The sum of pairwise parameter differences can be rewritten as
\begin{align*}
  \sum_{j\neq j^\prime}^k\|\bbeta_j - \bbeta_{j^\prime}\|^2 &=
    \bbeta'(\bC\otimes\bI_p)
    (\bC\otimes\bI_p)'\bbeta\\
    &= \bbeta'(\bC\bC' \otimes \bI_p)\bbeta.
\end{align*}
The objective function is therefore
\begin{align*}
O(\bbeta) = -\frac{1}{2}(\bY - \bX\bbeta)'\Sigma^{-1}(\bY - \bX\bbeta) - \lambda\bbeta'(\bC\bC'\otimes \bI_p)\bbeta.
\end{align*}

We maximize the objective function with respect to $\bbeta$ by finding the root of its derivative,
\begin{align*}
\nabla_{\bbeta} O(\bbeta, \lambda)' &= 
    -\frac{1}{2}(-2\bX'\bSigma^{-1}\bY + 2\scaledXTX\bbeta) 
    - 2\lambda (\bC\bC'\otimes \boldsymbol{I}_p) \bbeta.
\end{align*}
The root is the solution to the equation
\begin{align*}
0 &= 
    -\bX'\bSigma^{-1}\bY + \scaledXTX\hbbeta + 2\lambda (\bC\bC'\otimes \boldsymbol{I}_p)\hbbeta.
\end{align*}
Rearranging, we obtain the estimator
\begin{align*}
\hbbeta_{\ridge}(\lambda) &=
    \{\scaledXTX+2\lambda(\bC\bC'\otimes \bI_p)\}^{-1}\bX'\bSigma^{-1}\bY\\
    &=
    \{\scaledXTX+2\lambda(\bC\bC'\otimes \bI_p)\}^{-1}\scaledXTX\tbbeta = \bR(\lambda)\tbbeta.
\end{align*}

This estimator resembles the traditional ridge estimator, but instead of constraining the length of $\bbeta$ it constrains the sum of pairwise distances between combination of $\bbeta_j$s. We define a matrix $\bR(\lambda) =  \{\scaledXTX+2\lambda(\bC\bC'\otimes \bI_p)\}^{-1}\scaledXTX$ for the extent of this discussion, to simplify the notation.

To have a fair comparison to our estimator, we need to find the $\lambda$ that minimizes the MSE. The bias component is
\begin{align*}
    \Bias\{\hbbeta_{\ridge}(\lambda)\} &= \mathbb{E}\{\hbbeta_{\ridge}(\lambda)\} - \bbeta = \{\bR(\lambda)  - \bI_{pk}\}\bbeta.
\end{align*}
The variance component is 
\begin{align*}
    \Var\{\hbbeta_{\ridge}(\lambda)\} &= \bR(\lambda)\times \Var(\tbbeta)\times \bR^\prime(\lambda) = \bR(\lambda)(\scaledXTX)^{-1}\bR'(\lambda).
\end{align*}
The MSE is thus
\begin{align*}
    \MSE\{\hbbeta_{\ridge}(\lambda)\} &= 
        \tr[\Var\{\hbbeta_{\ridge}(\lambda)\}] + \|\Bias\{\hbbeta_{\ridge}(\lambda)\}\|^2\\
     &= \|(\bR(\lambda)  - \bI_{pk})\bbeta\|^2 +\tr[\Var\{\hbbeta_{\ridge}(\lambda)\}]. 
\end{align*}

As we do not know $\bbeta$, we minimize the unbiased estimator of the MSE using the identity
\begin{align*}
    \|\{\bR(\lambda)  - \bI_{pk}\}\bbeta\|^2 &= \EV[\|\{\bR(\lambda)  - \bI_{pk}\}\tbbeta\|^2] -\tr(\Var[\{\bR(\lambda)  - \bI_{pk}\}\tbbeta ]) \\
    &=\EV[\|\{\bR(\lambda)  - \bI_{pk}\}\tbbeta\|^2]  - \tr[\Var\{\hbbeta_{\ridge}(\lambda)\}] + 2\tr[\Cov\{\hbbeta_{\ridge}(\lambda),\tbbeta\}] - \tr\{\Var(\tbbeta)\}\\
      &= \EV[\|\{\bR(\lambda)  - \bI_{pk}\}\tbbeta\|^2]  - \tr[\Var\{\hbbeta_{\ridge}(\lambda)\}] + 2\tr[\Cov\{\bR(\lambda)\tbbeta,\tbbeta\}] - \tr\{\Var(\tbbeta)\}\\
    &= \EV[\|\{\bR(\lambda)  - \bI_{pk}\}\tbbeta\|^2]   - \tr[\Var\{\hbbeta_{\ridge}(\lambda)\}] + \tr[\{2\bR(\lambda)-\bI_{pk}\}\Var(\tbbeta)] .
\end{align*}
We may then rewrite the MSE for the ridge-like estimator as
\begin{align*}
   \EV\{\|(\bR  - \bI_{pk})\tbbeta\|^2\} + \tr[\{2\bR(\lambda)-\bI_{pk}\}(\scaledXTX)^{-1}].
\end{align*}
Thus, an unbiased estimator of the MSE of $\hbbeta_{\ridge}(\lambda)$ is 
\begin{align}
    \|\{\bR(\lambda)  - \bI_{pk}\}\tbbeta\|^2+ \tr[\{2\bR(\lambda)-\bI_{pk}\}(\scaledXTX)^{-1}].
\end{align}
We select $\lambda$ to minimize this estimator. 

In the setting of Section \ref{subsec:est_subset} where we conduct analysis on a subset of covariates, the ridge-like estimator minimizes a modified objective function,
\begin{align*}
O(\bbeta) = -\frac{1}{2}(\bM\bY - \bM\bX\bbeta)'\Sigma^{-1}(\bM\bY - \bM\bX\bbeta) - \lambda\bbeta'(\bC\bC'\otimes \bI_p)\bbeta,
\end{align*}
where $\bM$ is the block-diagonal matrix of projection matrices $\bM_j = \bI - \bZ(\bZ^\prime\bZ)^{-1}\bZ^\prime$ defined in \ref{subsec:est_subset}. This yields the estimator
\begin{align*}
\hbbeta_{\ridge,\sub}(\lambda) &=
    \{\scaledXTMX+2\lambda(\bC\bC'\otimes \bI_p)\}^{-1}\scaledXTMY\\
    &=
    \{\scaledXTMX+2\lambda(\bC\bC'\otimes \bI_p)\}^{-1}\scaledXTMX\tbbeta
\end{align*}

\section{Additional numerical support}
In simulation setting 1, we hold $\bbeta$ across Monte Carlo replicates while varying the group sample size. Table \ref{tab:setting1_parm_values} displays the parameter values at each level of $p$.
\begin{table}[ht!]
\centering
\begin{tabular}{rrrrr} 
$p$ & Parameter Index & Group 1 & Group 2 & Group 3 \\ 
  \hline
2 & 1 & 3.53 & 3.37 & 3.27 \\ 
  & 2 & 4.50 & 4.49 & 4.52 \\ \hline
  4 & 1 & 3.37 & 3.18 & 3.59 \\ 
   & 2 & 4.49 & 4.53 & 4.30 \\ 
    & 3 & 2.17 & 2.16 & 2.41 \\ 
    & 4 & 0.14 & $-$0.13 & 0.26 \\ \hline
  10 & 1 & 3.26 & 3.24 & 3.42 \\ 
   & 2 & 4.26 & 4.24 & 4.36 \\ 
   & 3 & 2.48 & 2.09 & 2.30 \\ 
   & 4 & $-$0.08 & 0.18 & 0.05 \\ 
   & 5 & $-$0.92 & $-$1.18 & $-$0.85 \\ 
   & 6 & 0.15 & 0.22 & $-$0.01 \\ 
   & 7 & $-$2.97 & $-$2.93 & $-$2.95 \\ 
   & 8 & 0.46 & 0.43 & 0.46 \\ 
   & 9 & $-$4.49 & $-$4.64 & $-$4.52 \\ 
   & 10 & 0.66 & 0.65 & 0.69 \\ \hline
  20 & 1 & 3.24 & 3.23 & 3.17 \\ 
   & 2 & 4.24 & 4.69 & 4.74 \\ 
   & 3 & 2.09 & 2.52 & 2.06 \\ 
   & 4 & 0.18 & 0.08 & 0.10 \\ 
   & 5 & $-$1.18 & $-$1.13 & $-$0.89 \\ 
   & 6 & 0.22 & 0.06 & 0.23 \\ 
   & 7 & $-$2.93 & $-$2.80 & $-$2.63 \\ 
   & 8 & 0.43 & 0.68 & 0.37 \\ 
   & 9 & $-$4.64 & $-$4.61 & $-$4.60 \\ 
   & 10 & 0.65 & 0.77 & 1.06 \\ 
   & 11 & $-$3.07 & $-$3.21 & $-$3.22 \\ 
   & 12 & $-$4.82 & $-$4.55 & $-$4.85 \\ 
   & 13 & 3.44 & 3.30 & 3.40 \\ 
   & 14 & $-$3.87 & $-$3.68 & $-$3.98 \\ 
   & 15 & 2.07 & 1.68 & 2.01 \\ 
   & 16 & 3.13 & 3.09 & 3.37 \\ 
   & 17 & $-$2.40 & $-$2.31 & $-$2.35 \\ 
   & 18 & $-$2.61 & $-$2.49 & $-$2.61 \\ 
   & 19 & 3.12 & 3.12 & 3.07 \\ 
   & 20 & $-$3.70 & $-$3.50 & $-$3.49 \\ 
   \hline
\end{tabular}
\caption{$\bbeta$ values for setting 1, for each number of covariates $p$.}
\label{tab:setting1_parm_values}
\end{table}

\section{eICU data access}\label{appsec:eicu_access}

Obtaining access to the eICU Collaborative Research Database requires credentialing on PhysioNet and completion of the required CITI training. To replicate the results we present here, one should download the patient, ApachePredVar, hospital, and ApachePatientResult tables. The data are structured with each patient allocated a different unit stay ID (patientunitstayid) when they transfer to a new unit, and each unit stay ID is allowed multiple results from the APACHE IV score calculation. Some patients do not receive an APACHE IV score calculation or any APACHE-related predictor variables. While the patient ID is ostensibly unique, we should conceive of each patient as a distinct `patient visit'; that is, an individual may be assigned multiple `uniquepid' if they are not recognized as the same person across visits.

We conduct the following data cleaning:
\begin{itemize}[itemsep=-0.25em, leftmargin=*]
    \item Keep the largest result for each unit stay ID in the Apache patient result table.
    \item Each patient visit in the patient table has multiple rows indicating transfer to a new unit. Keep the first instance of the patient in the hospital.
    \item Merge these de-duplicated patients and patient results with the APACHE predictor variables table and the hospital table.
    \item Keep only observations where the APACHE version is ``IVa'' and the APACHE score is not missing. Treat all APACHE scores of -1 as missing.
    \item Keep only the following variables: hospital id, gender, age, hospital discharge status, diastolic and systolic blood pressure, APACHE score, and actual ICU length of stay. Keep only observations that have data for all of these variables. Individuals aged 89 years and older will be removed from analysis, as these ages are censored.
\end{itemize}
\end{document}